\documentclass[a4paper,UKenglish,cleveref, autoref, thm-restate, numberwithinsect]{lipics-v2021}

\pdfoutput=1 
\hideLIPIcs  

\usepackage{graphicx}
\usepackage{makecell} 
\usepackage{tikz}
\usetikzlibrary{positioning}
\usetikzlibrary{backgrounds,fit,shapes.misc,hobby}
\usetikzlibrary{arrows.meta, calc}

\newif\iffullappendix
\fullappendixtrue

\tikzset{
	every picture/.style={line width=1pt},
	point/.style args={#1}{circle, draw=#1, fill=#1, minimum size=9pt, inner sep=0pt},
	center/.style args={#1}{rectangle, draw=black, fill=#1, line width= .5, minimum size=10pt, inner sep=0pt},
	outlier/.style args={#1}{circle, draw=#1, minimum size=9pt, inner sep=0pt},
    anchormark/.style={
        append after command={\pgfextra{%
        \draw[#1, thin, rotate around={45:(\tikzlastnode)}] (\tikzlastnode) 
        +(-7pt,-7pt) rectangle +(7pt,7pt);
    }}},
}

\newcommand{\convexpath}[2]{
  [
    create hullcoords/.code={
      \global\edef\namelist{#1}
      \foreach [count=\counter] \nodename in \namelist {
        \global\edef\numberofnodes{\counter}
        \coordinate (hullcoord\counter) at (\nodename);
      }
      \coordinate (hullcoord0) at (hullcoord\numberofnodes);
      \pgfmathtruncatemacro\lastnumber{\numberofnodes+1}
      \coordinate (hullcoord\lastnumber) at (hullcoord1);
    },
    create hullcoords
  ]
  ($(hullcoord1)!#2!-90:(hullcoord0)$)
  \foreach [
    evaluate=\currentnode as \previousnode using \currentnode-1,
    evaluate=\currentnode as \nextnode using \currentnode+1
  ] \currentnode in {1,...,\numberofnodes} {
    -- ($(hullcoord\currentnode)!#2!-90:(hullcoord\previousnode)$)
    let \p1 = ($(hullcoord\currentnode)!#2!-90:(hullcoord\previousnode) - (hullcoord\currentnode)$),
        \n1 = {atan2(\y1,\x1)},
        \p2 = ($(hullcoord\currentnode)!#2!90:(hullcoord\nextnode) - (hullcoord\currentnode)$),
        \n2 = {atan2(\y2,\x2)},
        \n{delta} = {-Mod(\n1-\n2,360)}
    in {arc [start angle=\n1, delta angle=\n{delta}, radius=#2]}
  }
  -- cycle
}

\NewDocumentCommand{\fairletarea}{O{gray!27} O{8pt} m}{%
  \scoped[on background layer]{%
    \fill[#1] \convexpath{#3}{#2};%
  }%
}

\usepackage[utf8]{inputenc}
\usepackage{xcolor}
\usepackage{color}

\definecolor{red}{RGB}{213, 94, 0}
\definecolor{blue}{RGB}{0, 114, 178}
\definecolor{purple}{RGB}{255, 133, 207}
\definecolor{green}{RGB}{30, 210, 138}

\usepackage{amsmath,amsthm,amsfonts}
\usepackage{mathtools}
\usepackage{MnSymbol}
\usepackage{mathrsfs}

\usepackage[linesnumbered,ruled,vlined]{algorithm2e}

\SetCommentSty{mycommfont}

\usepackage{caption}
\usepackage{subcaption}

\usepackage{hyperref}

\usepackage{booktabs}
\usepackage{longtable}
\usepackage{paracol}
\usepackage{tabularray}
\UseTblrLibrary{booktabs}

\DeclareMathOperator{\anc}{anc}
\DeclareMathOperator{\OPT}{OPT}

\DeclareMathOperator{\F}{\mathcal{F}} 

\DeclareMathOperator{\C}{\mathcal{C}} 
\DeclareMathOperator{\Representatives}{REP}
\renewcommand{\Representatives}{A}

\DeclareMathOperator{\OPTfair}{OPT_{fair}}

\DeclareMathOperator{\cost}{cost}

\title{Exact ratio preservation via outliers for fair \texorpdfstring{\(k\)}{k}-center clustering}

\author{Anna Arutyunova}{Heinrich Heine University Düsseldorf, Faculty of Mathematics and Natural Sciences, Germany}{anna.arutyunova@hhu.de}{}{}
\author{Irina Fast}{Heinrich Heine University Düsseldorf, Faculty of Mathematics and Natural Sciences, Germany}{irfas101@hhu.de}{}{}
\author{Annika Hennes}{Heinrich Heine University Düsseldorf, Faculty of Mathematics and Natural Sciences, Germany}{annika.hennes@hhu.de}{https://orcid.org/0000-0001-9109-3107}{}
\author{Carsten Krollmann}{Heinrich Heine University Düsseldorf, Faculty of Mathematics and Natural Sciences, Germany}{cakro105@hhu.de}{}{}
\author{Daniel R. Schmidt}{Heinrich Heine University Düsseldorf, Faculty of Mathematics and Natural Sciences, Germany}{dschmidt@hhu.de}{https://orcid.org/0000-0001-7381-912X}{}
\author{Melanie Schmidt}{Heinrich Heine University Düsseldorf, Faculty of Mathematics and Natural Sciences, Germany}{mschmidt@hhu.de}{https://orcid.org/0000-0003-4856-3905}{}

\authorrunning{A. Arutyunova, I. Fast, A. Hennes, C. Krollmann, D. Schmidt and M. Schmidt} 

\Copyright{Anna Arutyunova, Irina Fast, Annika Hennes, Carsten Krollmann, Daniel R. Schmidt, and Melanie Schmidt} 

\ccsdesc[500]{Theory of computation~Facility location and clustering}
\ccsdesc[500]{Theory of computation~Approximation algorithms analysis}

\keywords{Fairness, k-center, approximation algorithms}

\supplement{\hfill}

\supplementdetails[subcategory={Source Code}, cite={}, swhid={}]{Software}{https://github.com/algo-hhu/fair-k-center-via-outliers/}

\funding{Funded by the Deutsche Forschungsgemeinschaft (DFG, German Research Foundation) - Project 456558332 through the Emmy Noether Programme}

\nolinenumbers 

\EventEditors{John Q. Open and Joan R. Access}
\EventNoEds{2}
\EventLongTitle{42nd Conference on Very Important Topics (CVIT 2016)}
\EventShortTitle{CVIT 2016}
\EventAcronym{CVIT}
\EventYear{2016}
\EventDate{December 24--27, 2016}
\EventLocation{Little Whinging, United Kingdom}
\EventLogo{}
\SeriesVolume{42}
\ArticleNo{23}

\begin{document}

\maketitle

\begin{abstract}
    We study the $k$-center clustering problem under demographic fairness constraints, where the point set is partitioned into groups, and the aim is to compute clusters that exhibit a given group proportion. Previous work in this direction assumes that the entire point set already respects the desired proportions or uses relaxed notions of fairness.

    In this work, we propose a model that facilitates the creation of clusters that exactly match given target ratios, even when the input point set does not. We combine the well-known fair clustering model initiated by Chierichetti, Kumar, Lattanzi, and Vassilvitskii \cite{chierichetti2017fair} with the notion of outliers to obtain a practical combinatorial framework that provides constant-factor approximate solutions for all proportion settings from $1:1$ for two groups to $t_1:t_2:\ldots:t_m$ for $m\geq 2$ groups, where $t_1,\ldots,t_m$ are integers.
    
    We implement and evaluate our algorithms, compare different variants, and provide evidence of the practicability of this approach.
\end{abstract}

\thispagestyle{empty}
\newpage
\setcounter{page}{1}

\section{Introduction}

Fair clustering is a very active line of research in clustering algorithm design, introduced by Chierichetti, Kumar, Lattanzi, and Vassilvitskii~\cite{chierichetti2017fair} in 2017. The key idea is that, when applying $k$-clustering, the user may specify a protected attribute such as gender, race, or another demographic variable to which the algorithm should pay particular attention. The goal is to ensure group-level fairness: each cluster shall contain the same proportion of attribute values (i.\,e., demographic groups) as observed in the overall dataset. That is, the composition of every cluster shall match the global distribution of the protected attribute.

Enforcing such fairness can be valuable in a variety of real-world settings, as it helps ensure diversity within each group. This is why the model is called \emph{fair}: It makes sure that every group is adequately represented in each cluster. Fairness is a desired or necessary condition in many applications, e.\,g., representation of people with protected characteristics in committees, creating geographic zones with demographic constraints (e.g., schools~\cite{doi:10.1137/1.9781611977912.97}), allocation of scarce resources (e.g., access to childcare, or charging times for electric vehicles~\cite{10251779}), or preventing dominance of a single actor in ads~\cite{ahmadian2019clustering}.

To further explain the model, let us consider a simple base case (the general case is defined in Section~\ref{sec:def}): Given is a set of points $P$, a metric $d : P \times P \to \mathbb{R}_{\ge 0}$, the desired number of centers $k$, and a mapping $\gamma : P \to \{\text{blue},\text{red}\}$ that satisfies $|\gamma^{-1}(\text{blue})|=|\gamma^{-1}(\text{red})|$, i.\,e., the input contains an equal number of red and blue points. Then the fair $k$-center problem is to partition $P$ into $k$ clusters where each of the clusters $C_1,\ldots, C_k$ also contains an equal number of red and blue points, i.\,e.,  $|C_i \cap \gamma^{-1}(\text{blue})|=|C_i \cap \gamma^{-1}(\text{red})|$  for all $i\in [k]$. Quite intuitively, this problem can be solved via a combination of matching and unconstrained clustering: First, pair each blue point with one red point such that the maximum distance between a point and its partner is minimized (see Figure~\ref{fig:fair-clustering:b}). Second, compute an unconstrained $k$-center clustering on $P$ and assign each pair $(r,b)$ to the center $c$ that minimizes $\max\{d(r,c),d(b,c)\}$ (see Figure~\ref{fig:fair-clustering:c}). 
This approach indeed yields a constant factor approximation, and it can be refined to achieve a $3$-approximation for the fair $k$-center problem for this special two-color case with color classes of equal cardinality.

\tikzset{
    pointset/.pic={
    \node[point=red] (r1) at (0,0.2) {};
    \node[point=blue] (b1) at (1,1) {};
    \node[point=red] (r2) at (1,0) {};
    \node[point=blue] (b2) at (1,-1) {};
    \node[point=red] (r3) at (1.4,.45) {};
    \node[point=blue] (b3) at (1.6,-.1) {};
    \node[point=red] (r4) at (3.2,0) {};
    \node[point=blue] (b4) at (2.3,.6) {};
    \node[point=red] (r5) at (5,.4) {};
    \node[point=blue] (b5) at (4.5,-.3) {};
  }
}

\begin{figure}[htbp]
    \captionsetup[subfigure]{justification=centering}
    \centering
    \begin{subfigure}[t]{0.3\textwidth}
        \scalebox{.64}{%
        \begin{tikzpicture}
            \useasboundingbox (-1,-3) rectangle (5,3);
            \pic{pointset};
            
            \draw (r2) circle (1.55);
            \draw (b5) circle (1.45);

            \node[center=red] (r2) at (r2) {};
            \node[center=blue] (b5) at (b5) {};
        \end{tikzpicture}}
        \subcaption{A clustering that is not aware of the colors.}
        \label{fig:fair-clustering:a}
    \end{subfigure}\hfill
    
    \begin{subfigure}[t]{0.3\textwidth}
    \scalebox{.63}{%
    \begin{tikzpicture}
        \useasboundingbox (-1,-3) rectangle (5,3);
        \pic{pointset};
        \draw[-] (r1) -- (b1);
        \draw[-] (r2) -- (b2);
        \draw[-] (r3) -- (b3);
        \draw[-] (r4) -- (b4);
        \draw[-] (r5) -- (b5);
    \end{tikzpicture}}
    \subcaption{A matching between red and blue points.}
    \label{fig:fair-clustering:b}
    \end{subfigure}\hfill
    \begin{subfigure}[t]{0.3\textwidth}
    \scalebox{.63}{%
    \begin{tikzpicture}
        \useasboundingbox (-1,-3) rectangle (5,3);
        \pic{pointset};
        \draw[-] (r1) -- (b1);
        \draw[-] (r2) -- (b2);
        \draw[-] (r3) -- (b3);
        \draw[-] (r4) -- (b4);
        \draw[-] (r5) -- (b5);

        \draw (r2) circle (2.3);
        \draw (b5) circle (1.0);

        \node[center=red] (r2) at (r2) {};
        \node[center=blue] (b5) at (b5) {};
    \end{tikzpicture}}
    \subcaption{A fair clustering that respects a $1:1$ color ratio in every cluster.}
    \label{fig:fair-clustering:c}
    \end{subfigure}\hfill
    \begin{subfigure}[t]{0.3\textwidth}
        \scalebox{.63}{%
    \begin{tikzpicture}
        \useasboundingbox (-1,-3) rectangle (5,3);
        \pic{pointset};
        \node[point=blue] (b6) at ($(r4)+(.7,0)$) {};
        
        \draw[-] (r1) -- (b1);
        \draw[-] (r2) -- (b2);
        \draw[-] (r3) -- (b3);
        \draw[-] (r4) -- (b6);
        \draw[-] (r5) -- (b5);

        \draw (r2) circle (1.15);
        \draw (b5) circle (1.45);

        \node[center=red] (r2) at (r2) {};
        \node[center=blue] (b5) at (b5) {};
    \end{tikzpicture}}
    \caption{A $1:1$-fair clustering with $1$ blue outlier.}
    \label{fig:fairwithoutlier}
    \end{subfigure}
    \caption{Visualizations of fair decomposition and fair clustering without/with one outlier. \label{fig:fair-clustering}}
\end{figure}
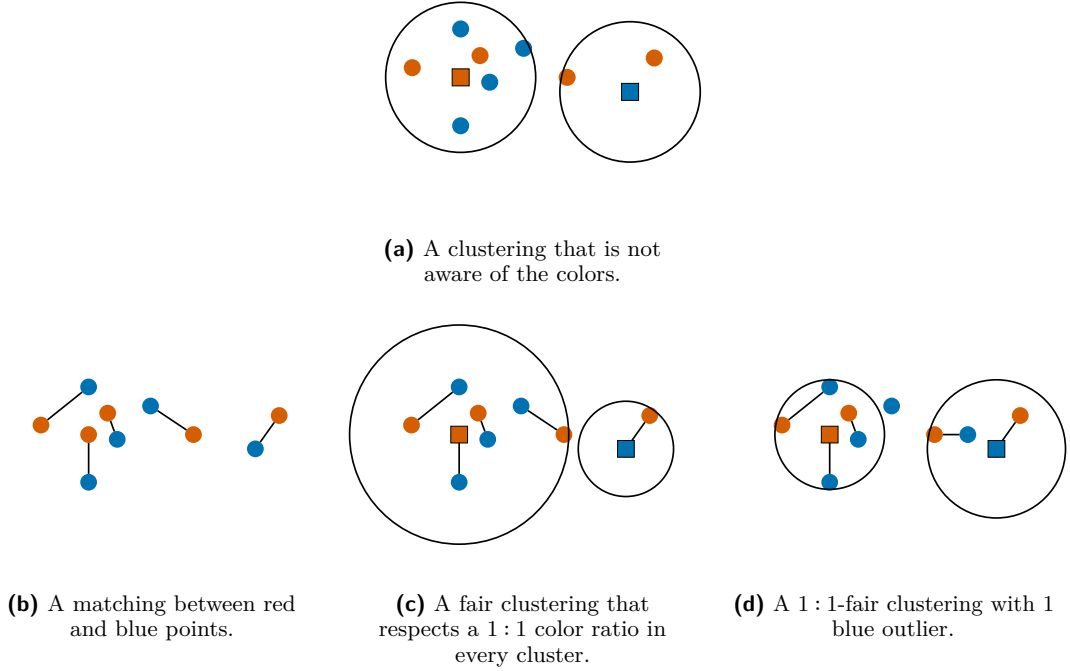

The first step is a special case of a so-called \emph{fairlet decomposition} which decomposes $P$ into micro clusters called fairlets which 1) exactly match a given ratio and 2) are inclusion-wise minimal under this property, i.\,e., no subset satisfies the ratio. 
When applied to multiple colors and arbitrary ratios, the combinatorial structure is more involved, but computing fairlet decompositions that exactly meet given ratios remains a combinatorially nice problem. For example~\cite{rosner2018privacy} describe how to compute fairlet decompositions in general.

The main hurdle now is that in real-world instances, enforcing exact proportional representation is often infeasible or undesirable. Imagine that a point set $P$ contains $999$ red points and $1001$ blue points. 
Since these numbers are coprime, the only subset of $P$ that matches the input ratio is $P$ itself, i.\,e., the only existing fairlet decomposition is a composition in just one subset, namely, $P$. 

There are different ways out of this scenario. The path pursued so far is to relax the notion of fairness and require that the input ratio is only nearly matched, e.\,g., by allowing an $\epsilon$-deviation. This idea is studied in 
\cite{bera2019fairclustering} and \cite{bercea2018bi} for multicolor inputs and arbitrary input ratios. The resulting algorithms achieve bicriteria guarantees, allowing a small additive violation in the fairness constraints, and are based on LP rounding. Combinatorial algorithms for these interval-based models are not known. The reliance on LP methods can be expected for this variant, as the introduction of approximate preservation disrupts the underlying combinatorial structure that fair clustering problems typically exhibit and that combinatorial approximation algorithms could exploit.

We go a different route. Our model is inspired by the well-studied $k$-center with outlier problem where a few points can be excluded from the clusters. 
For an initial motivation, consider the problem of finding a minimum cost maximum matching in a bipartite graph. If the two sides of the graph are of different cardinality, then we can only find a perfect matching for the smaller side. The larger side will have unmatched vertices, which can be viewed as outliers.
The number of outliers is determined by the difference in cardinality of the two sides and the outliers are chosen such that the cost of the matching is minimized. This scenario is a special case of our model, with \(m=2\) and a target ratio of \(1:1\).
We want to generalize this idea to fair clustering. In its most general form, we have an arbitrary number of classes and allow the specification of a target ratio that might be different from the input ratio. In the fair clustering, the target ratio is achieved by excluding a certain number of points from different classes as outliers such that the ratio of the remaining points matches the target ratio.

Our model, applied to the two-color case, is as follows: Say the ratio of red to blue points in the input point set is $r:b$, where these numbers are coprime and $r < b$. Then the user can specify a target ratio $1:s$ where $s \le b/r$, i.\,e., it is achievable by excluding points from the majority color. Then our fair $k$-center with outliers formulation asks for $1:s$-fair clusters and allows the respective number of points of the second color to be marked as outliers. See \Cref{fig:fairwithoutlier} for a $1:1$-fair clustering with $1$ blue outlier.

The benefit of this model is that it keeps the nice combinatorial structure of the exactly fair model (compared to the approximately fair models) but enables us to compute reasonable clusterings for any point set, not only those where the ratio of colors in the input is specifically nice. 
We extend known algorithms for fair $k$-center clustering to obtain our results. 
\Cref{tab:summary of cases} summarizes our results for fair $k$-center for different cases. The first column describes the target ratio that we want to achieve. The first row is only for two colors and the important special case that the ratio is of the form $1:t$ which allows for better approximation. Rows 2 and 3 are for arbitrary numbers of colors where row 3 allows ratios of the form $2:3:3$ while row 2 gives the easier case that is an extension of row 1. We explicitly list how expensive the fairlet decomposition is; the final clustering only increases the cost mildly.
\begin{table} [htbp] 
    \caption{Summary of our results on fair $k$-center with outliers. The final clustering is always computed by  \Cref{alg:general-case-outlier-algorithm}, see \Cref{lem:4-approximation}. $\OPT$ refers to the optimal solution for the respective fairness condition.}
    \label{tab:summary of cases}
    \begin{center}
    \begin{tabular}{ccccc} 
            \toprule
        & compute fairlets & fairlet cost & clustering cost \\
        \midrule
        $1:t$ & \Cref{alg:compute-1-t-fairlets} & $2\OPT$ (\Cref{cor:1-t-fairlet-decomposition-2OPT}) & $4\OPT$ (\Cref{thm:1-t-unbalanced}) \\ 
        $1:t_2:...:t_m$ & \Cref{alg:compute-1:t2:...:tm-fairlets} & $2\OPT$ (\Cref{cor:1-t2-...-tm-fairlet-decomposition-2OPT}) & $4\OPT$ (\Cref{thm:1-t1-...-tm-unbalanced}) \\ 
        $t_1:\ldots:t_m$ & \Cref{alg:compute-s1-s2-...-sm-fairlets} & $12\OPT$ (\Cref{lem:properties-computed-fairlets-s1:s2:...:sm}) & $14\OPT$ (\Cref{thm:s1:...:sm-14-approx})
    \end{tabular}
    \end{center}
\end{table}

One can view our algorithms also as a contribution to the field of combining different clustering constraints, in this case, fairness and outliers. However, note that the number of outliers is computed automatically from the target ratio, and that we never choose outliers from the minority color class.

\subsection{Related work}
We review work that is most related to our work and follows the seminal paper~\cite{chierichetti2017fair}. 
Various variants of fair clustering have been proposed that we do not discuss in detail, e.\;g.,~\cite{anegg2022trueapprox,jia2021faircolorful,kim2025fairclusteringalignment,KleindessnerSAM19}.
Chierichetti, Kumar, Lattanzi, and Vassilvitskii~\cite{chierichetti2017fair} introduced the following model: For a point set $P = R \cup B$ consisting of red and blue points, the \emph{balance} of a set $Y \subset P$ is the value $\min ( \frac{|Y \cap R|}{|Y \cap B|}, \frac{|Y \cap B|}{|Y \cap R|} ) \in [0, 1]$. The balance of a clustering is the smallest balance of any of its clusters. The focus of this model is to produce clusterings that are as balanced as possible, ideally reaching a balance of $1$ (i.\,e., a \(1:1\) ratio). A clustering is $t$-fair in this model if every cluster has a balance value of at least $t$.

\cite{chierichetti2017fair} also introduced the notion of fairlets.  One striking observation in the paper was that for objective functions such as $k$-center and $k$-median, the computation of a cheap \emph{fairlet decomposition}, i.\,e., a partitioning of $P$ into fairlets, often becomes a clean combinatorial optimization problem. For example, if $P$ itself is (exactly) $1$-balanced, then computing a fairlet decomposition becomes a matching problem (in a suitable bipartite graph with weights depending on the objective), and constant factor approximations for computing $1$-balanced clusterings can be deduced from this fact.
Early follow-up works \cite{huang2019coresetsfair} and \cite{schmidt2019faircoresets} provide speed-ups that increase the scalability of this approach.

\subparagraph*{Multi-color generalizations.}
If the protected attribute has more than two values, 
then a literal generalization of the balance-based model is the $\alpha$-capped $k$-center problem introduced by~\cite{ahmadian2019clustering}. There, the goal is to partition $P$ in a way such that no single color (group) exceeds an $\alpha$-fraction of any cluster. 
This generalization is motivated by real-world scenarios like online advertising, where keywords (points) are clustered, and advertisers (colors) should not dominate any single cluster to prevent manipulation. For two colors, the $\alpha$-capped $k$-center problem equals the $t$-balanced fair $k$-center problem. This model focuses on limiting the majority's share in each cluster, not on matching the cluster representation to the original distribution of group proportions.
Our work builds on the line of work that emphasizes accurate group representation and minority protection and uses a different generalization. The works~\cite{bercea2018bi,BohmFLMS21,rosner2018privacy} for the multi-color case demand that each protected group's share of each cluster must match the share in the overall population exactly. We call this model \emph{exact representation} in Section~\ref{sec:def}. If the input contains the same number of points of each color, then~\cite{BohmFLMS21} provides $O(1)$-approximations for exact representation for various $k$-clustering objectives, including $k$-median and $k$-center. For preserving other input ratios $t_1:t_2:\ldots:t_m$, 
$O(1)$-approximations are only known for the $k$-center and $k$-supplier case, see~\cite{bercea2018bi,rosner2018privacy}. The best approximation ratio for fair $k$-center with exact representation is a $5$-approximation for multiple colors~\cite{bercea2018bi}.

\subparagraph*{Approximate representation.} The works of~\cite{bercea2018bi} and~\cite{bera2019fairclustering} independently introduced a model for approximate representation where the ratio of color $i$ in each cluster has to be in an interval $[\rho_i-\ell_i,\rho_i+u_i]$ where $\rho_i$ is the input ratio and $\ell_i$ and $u_i$ can be specified in the input. Both papers provide bi-criteria approximations for this model that incur a small additive fairness violation, and both results are obtained via LP rounding. Harb and Lam~\cite{harb2020kfc} employ a randomized approach that reduces the size of the LP in practice, providing a $3$-approximation where clusters are fair only in expectation.

\subsection{Definitions}\label{sec:def}

Let $P$ be a set of $n$ points, $d\colon P\times P\rightarrow \mathbb R_{\geq 0}$ be a metric on $P$ and $k \in\{1,\ldots,n\}$ be the desired number of clusters. A  $k$-clustering is a sequence of $k$ pairwise disjoint sets $C_1,\ldots, C_k\subset P$ and points $c_i\in C_i$ for $1\leq i\leq k$ such that  $\bigcup_{i=1}^k C_i=P$. 
For a given clustering instance $(P, d, k)$ the $k$-center problem is to find such a $k$-clustering while minimizing the objective $\max_{i=1,\ldots, k}\max_{p\in C_i}d(p,c_i).$
The fair $k$-center problem is a constrained version in which every point in the set $P$ is assigned a color between $1$ and $m$, and the clustering must preserve certain color ratios. 
We first define the input setting for such a problem.

\begin{definition}
\label{def:fair_clustering_instance:copy}
    A fair clustering instance $(P,d,k,m,\gamma)$ consists of a set $P$, a metric $d$ on P, an upper bound on the number of clusters $k\in\mathbb N$, the number of colors $m$, and a color assignment $\gamma\colon P\rightarrow \{1,\ldots, m\}$. We denote by $\Gamma=[m]$ the set of colors. For every $i\in \Gamma$ we denote by $H_i=\gamma^{-1}(i)$ the points in $P$ with color $i$.  
\end{definition}
We assume that the colors are numbered such that $|H_1|\le |H_2|\le \ldots \le |H_m|$. 

\subparagraph*{Fair $k$-clustering with outliers.}
Let $t_1,\ldots,t_m \in \mathbb{N}$ be \emph{coprime} integers that indicate the desired proportion of points from different protected groups within the clusters, i.\,e., a fair cluster should contain points from the $m$ protected groups in the ratio $t_1\colon t_2\colon \ldots \colon t_m$. We define:

\begin{definition}
    A cluster $C_j$ is \emph{$t_1:t_2:\ldots:t_m$-fair} for coprime integers $t_1,\ldots,t_m \in \mathbb{N}$ if it satisfies for all $i\leq m$:
\begin{align}
\label{eq:fairness:new}
    \frac{\left|C_j \cap H_i\right|}{\left|C_j\right|} = \frac{t_i}{\sum_{\ell\leq m}t_\ell}.
\end{align}
\end{definition}
If there are no outliers, then fair clusters under this definition can only be obtained if the $t_i$ match the input ratios. 
In the \emph{exact representation model} one chooses the $t_i$ by setting $t_i = |H_i| / \gcd(|H_1|, \ldots, |H_m|)$  for all $i\leq m$. 
In fair $k$-clustering with outliers, we can also use other values for the $t_i$ by allowing the fitting number of outliers for all but the minority color. 
 
\begin{definition}
 Given an instance $(P,d,k, m,\gamma)$ as defined in \Cref{def:fair_clustering_instance:copy}, a \emph{$t_1:t_2:\ldots:t_m$-fair $k$-clustering with $z$ outliers} is given by a set of $k$ centers $C=\{c_1,\ldots,c_k\}\subseteq P$ and an assignment $\alpha\colon P \to C\cup \{\perp\}$ such that: (a) $\left|Z\right| \leq z$ where $Z = \alpha^{-1}(\perp)$, (b) every cluster $C_i \coloneqq \alpha^{-1}(c_i)$ induced by $c_i\in C$ is $t_1:t_2:\ldots:t_m$-fair.

 Equivalently, such a clustering can also directly be defined as a sequence of pairwise disjoint clusters $\C=\{C_1,\ldots, C_k\}$ from $P$ such that: (a) $Z=P\backslash\bigcup_{i\leq k}C_i$ has cardinality at most $z$, (b)  $C_i$ is $t_1:t_2:\ldots:t_m$-fair for all $1\leq i\leq k$.
\end{definition}
We will use both definitions interchangeably.
Fair $k$-clustering with outliers can be applied to any $k$-clustering problem. In this paper, we initiate the study by investigating the metric $k$-center problem where 
$r(C,\alpha)=\max_{x\in P\setminus Z} d(x,\alpha(x))$
needs to be minimized.

\subsubsection{General fairlets and fairlet decompositions}
We now give a general definition of fairlets.
A fairlet is an atomic fair cluster, i.\,e., a cluster that is fair and as small as possible.

\begin{definition}[$t_1:\ldots:t_m$-fairlet]\label{def:t1:...:tm-fairlet}
    Let $P=H_1\cupdot H_2 \cupdot \ldots \cupdot H_m$ be a set of points, and $t_1,t_2,\ldots,t_m\in \mathbb{N}$ with $\text{gcd}(t_1,\ldots, t_m)=1$.
    A subset $f\subseteq P$ is a $t_1:\ldots:t_m$-fairlet if $|f\cap H_i|=t_i$ for all $i\le m$.
\end{definition}
The concept of fairlets was introduced by~\cite{chierichetti2017fair}, and generalized in subsequent papers~\cite{bercea2018bi,BohmFLMS21,rosner2018privacy}, and our algorithms are adapted from \cite{chierichetti2017fair} and \cite{rosner2018privacy}.
The structural insight that we inherit from the exact representation model is that if every cluster in the optimal solution is $t_1:\ldots:t_m$-fair, then every optimal cluster can be decomposed into fairlets, and thus $P$ can be decomposed into a set of fairlets (see Figure~\ref{fig:fairlet-decomp-outlier} for an optimal $t_1:\ldots:t_m$-fair $k$-center clustering with outliers with $m=3$, $t_1=1$, $t_2=2$, $t_3=3$ and $k=2$ and its decomposition into fairlets), and these fairlets are cheap (their diameter cannot exceed the diameter of an optimal cluster). 
Since this decomposition into cheap fairlets exists, finding it and then computing a clustering from it becomes the core of the problem, and this can be tackled via combinatorial algorithms.

\begin{figure} 
    \centering
    \scalebox{0.5}{
\begin{tikzpicture}[x=1cm,y=1cm]
  \tikzset{
    edge/.style   ={gray!70, line width=0.8pt}
  }

  \coordinate (L1) at (4,7.3);
  \node[point=purple] at (L1) {};
  \node[point=red] (l1a) at (2.6,8.1) {};
  \node[point=purple] (l1b) at (3.5,8.2) {};
  \node[point=blue]  (l1c) at (4.9,7.7) {};
  \node[point=purple] (l1e) at (2.6,6.6) {};
  \node[point=blue] (lm1) at (2.7,7.0) {};

  \fairletarea{l1e,lm1,l1a,l1b,l1c,L1};

  \coordinate (L2) at (3.2,3.6);
  \node[point=red] at (L2) {};
  \node[point=purple] (l2a) at (2.4,3.0) {};
  \node[point=blue]  (l2b) at (4.2,3.0) {};
  \node[point=blue]  (l2c) at (3.2,1.9) {};
  \node[point=purple] (l2d) at (4.0,2.1) {};
  \node[point=purple] (l1f) at (1.9,4.0) {};
  \fairletarea{l2b,l2d,l2c,l2a,l1f,L2};

  \node[point=blue]  (m1) at (2.5,5.2) {};
  \node[point=purple] (m2) at (3.7,5.1) {};
  \node[point=purple] (m3) at (6.0,5.5) {};
  \coordinate (m5) at (6.0,6.0);
  \node[point=purple] (l1d) at (0.6,5.0) {};
  \node[point=blue] (m4) at (6.6,6.6) {};
  \node[point=red] (mred) at (6.6,5.5) {};
  \fairletarea{mred,m3,m2,m1,l1d,m5,m4}

  \coordinate (R1) at (11.9,5.2);
  \node[point=red] at (R1) {};
  \node[point=blue]  (r1a) at (11.9,6.2) {};
  \node[point=purple] (r1b) at (13.5,5.1) {};
  \node[point=purple] (r1c) at (12.85,4.55) {};
  \node[point=blue]  (r1d) at (11.1,5.2) {};
  \node[point=purple] (r1e) at (10.5,5.2) {};
  \fairletarea{r1a,r1b,r1c,R1,r1d,r1e}

  \coordinate (R2) at (11.5,4.0);
  \node[point=red] at (R2) {};
  \node[point=purple] (r2a) at (10.2,2.2) {};
  \node[point=blue]  (r2b) at (12.0,2.1) {};
  \node[point=purple] (r2c) at (12.8,3.4) {};
  \node[point=blue]  (r2d) at (10.0,3.8) {};
  \node[point=purple] (r2e) at (11.1,1.8) {};
  \fairletarea{r2a,r2d,R2,r2c,r2b,r2e}
  \draw (R2) circle (2.7cm);

  \draw (m2) circle (3.75cm);

  \node[point=blue]  at (6.0,9.0) {};
  \node[point=purple]  at (1.7,9) {};
  \node[point=blue]  at (6.5,1.1) {};
  \node[point=purple] at (1.7,0.5) {};
  \node[point=blue]  at (13.0,8.1) {};

\end{tikzpicture}
}
    \caption{An optimal clustering with two pink and three blue outliers and a decomposition of the optimal clusters into fairlets.}
    \label{fig:fairlet-decomp-outlier}
\end{figure}
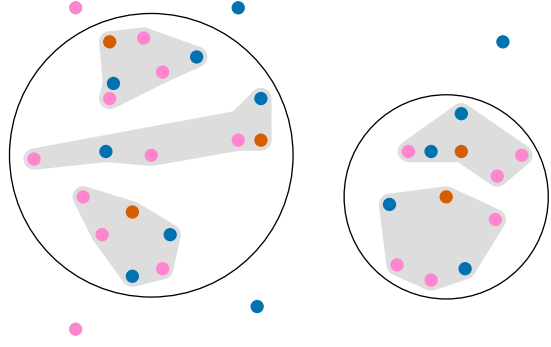

For every fairlet $f$, we fix an \emph{anchor} $\anc(f)\in P$. This is a point that will represent the fairlet. It can be a point in the fairlet, but also outside of the fairlet. 
We define the \emph{cost} of a fairlet $f$ as $\cost(f) = \max_{p\in f}d(\anc(f),p)$. 
A \emph{fairlet decomposition} of a set $P$ is a maximal set of pairwise disjoint fairlets inside $P$. Let $\F$ be a fairlet decomposition. By $P(\F) = \bigcup_{f\in \F}f$, we denote the set of points covered by the fairlet decomposition.
The \emph{cost of a fairlet decomposition} is given by the maximal cost among any of its fairlets.
We will also need the notion of agreement between a fairlet decomposition and a clustering.
\begin{definition}
\label{def:fairlet-agree-with-clustering}
	Let $\C$ be a fair $k$-center clustering with outliers and $\F$ be a fairlet decomposition of $P$. We say that $\F$ and $\C$ \emph{agree} with each other iff for all $f\in\F$ there exists a cluster $C\in\C$ such that $f\subseteq C$.
\end{definition}
We assume that the desired color ratio $t_1:\ldots:t_m$ fulfills $\gcd(t_1,\ldots,t_m)=1$. Otherwise, we can divide the $t_i$ by their greatest common divisor.
We distinguish between two cases. In the \emph{one-sided outlier case}, we fix one minority color class from which we do not remove any outliers and the numbers of outliers in the other colors are determined by the excess of the number of points in relation to the minority color class. In this case, the outliers are a means to achieve the target ratio. 
In the \emph{two-sided outlier case}, we are given a number \(z\) of outliers we may exclude instead of it being determined by the number of points in the minority class. Thereby, we can choose outliers from all color classes. In this case, the outliers are a means to achieve the target ratio as well as to reduce the clustering cost. This variant is closer to the original idea of clustering with outliers; we give an algorithm for the \(2\)-color case \(m=2\).
In the one-sided outlier case, we further assume that $t_i|H_1|\le t_1|H_i|$ for all $i\leq m$ and that $|H_1|/t_1\in \mathbb{Z}$. Under these assumptions, it is possible to find a fairlet decomposition of size $|H_1|/t_1$ that fully covers $H_1$, and the fairlet decomposition problem is to compute such a decomposition of minimal cost.

\section{The algorithms}

We extend a collection of results for fair $k$-center with exact representation to the case of fair $k$-center with outliers from~\cite{chierichetti2017fair,rosner2018privacy}. The general approach of our adaptation is as follows: Find a set of fairlets and so-called \emph{anchors} and declare the points not contained in any fairlet as outliers. Then, find a set of centers among the non-outlier points (e.g., within the set of all anchors) and assign the non-outlier points to centers such that points inside the same fairlet are assigned to the same center. Outliers are never chosen from the color class that has the smallest cardinality. 
We exclude the minimum number of points from the other classes such that the point set achieves the target ratio.
That is, the number of outliers is determined by the fairness criterion and the actual color distribution inside the given point set.
(Note that in Appendix~\ref{sec:bothsides} we also study a variant with two colors where outliers can be chosen from both color classes and give an $11$-approximation for this scenario, but the main part of our paper does not allow choosing outliers from the minority color class.)

\begin{algorithm}
	\LinesNumbered
	\caption{\textsc{Fair-\texorpdfstring{$k$}{k}-center-with-outliers}}
	\label{alg:general-case-outlier-algorithm}
	\SetKwInOut{Input}{Input}
	\SetKwInOut{Output}{Output}
	\BlankLine
	\Input{Point set $P$, distances $d\colon P\times P \to \mathbb{R}_{\ge 0}$, number of clusters $k$}
	\Output{Set of centers $C$, assignment $\alpha\colon P(\F)\to C$}
	\BlankLine
        $\mathcal{F}, \anc \gets \textsc{compute-fairlets}(P,d,k) $ \label{alg-line:call-compute-fairlets}\\
        $\Representatives \gets \{\anc(f) \mid f\in \F\}$\\
	$C \gets \textsc{compute-centers}(\Representatives,d,k)$ \label{alg-line:call-compute-centers}\\
	\For{$f\in \F$}{
		$c_f \gets \arg\min_{c\in C}d(c,\anc(f))$\\
        $\forall p\in f\colon\ \alpha(p) \gets c_f$
	}
    \Return{$C,\alpha$}
\end{algorithm}

\begin{figure}
    \centering
    \scalebox{0.5}{
\begin{tikzpicture}[x=1cm,y=1cm]
  \tikzset{
    edge/.style   ={gray!70, line width=0.8pt}
  }

  \coordinate (L1) at (4,7.3);
  \node[point=red, anchormark=black] (l1r1) at (2.6,6.6) {};
  \node[point=blue]  (l1b1) at (4.3,8) {};
  \node[point=blue] (l1b2) at (2.7,10) {};
  \draw (l1r1) -- (l1b2) node[draw=none,fill=none,font=\huge,midway,left] {$\varphi$};

  \fairletarea{l1b1,l1r1,l1b2};

  \node[point=blue]  (l2b1) at (3.5,5.2) {};
  \coordinate (L2) at (6.0,6.0);
  \node[point=blue] (l2b2) at (6.6,6.6) {};
  \node[point=red, anchormark=black] (l2r1) at (6.1,5.5) {};
  \fairletarea{l2b1,l2b2,l2r1};
 
  \coordinate (L3) at (3.2,3.6);
  \node[point=red, anchormark=black] (l3r1) at (L3) {};
  \node[point=blue]  (l3b1) at (4.2,3.0) {};
  \node[point=blue]  (l3b2) at (3.2,1.9) {};
  
  \fairletarea{l3b1,l3b2,L3};

  

  \coordinate (R1) at (11.9,5.2);
  \node[point=red, anchormark=black] (r1r1) at (R1) {};
  \node[point=blue]  (r1b1) at (11.9,6.8) {};
  \node[point=blue]  (r1b2) at (11.1,5.2) {};
  \fairletarea{r1b1,R1,r1b2};

  \coordinate (R2) at (11.5,4.0);
  \node[point=red, anchormark=black] (r2r1) at (R2) {};
  \node[point=blue]  (r2b1) at (12.0,2.1) {};
  \node[point=blue]  (r2b2) at (9.0,3.5) {};
  \fairletarea{r2b2,R2,r2b1};

  \draw[red] (R1) circle (1.3cm);

  \draw[red] (L3) circle (3.5cm);
  \draw (l3r1) -- (l2r1) node[draw=none,fill=none,font=\huge,midway,below] {$\rho$};

  \node[outlier=blue]  at (6.0,9.0) {};
  \node[outlier=blue]  at (6.5,1.1) {};
  \node[outlier=blue]  at (13.0,3.4) {};

\end{tikzpicture}
}
    \caption{Illustration of the clustering framework. Anchors are depicted by diamond frames, the fairlet decomposition is indicated by gray areas, and the outliers are depicted by unfilled circles. 
    The red circles show an unconstrained clustering of the anchors.
    The fair clustering is obtained by assigning all points of a fairlet to the center that is closest to their anchor. The resulting fair clustering costs \(\rho+\varphi\), where \(\rho\) is the cost of the unconstrained clustering of the red anchors and \(\varphi\) is the cost of the fairlet decomposition.}
    \label{fig:cluster-framework}
\end{figure}
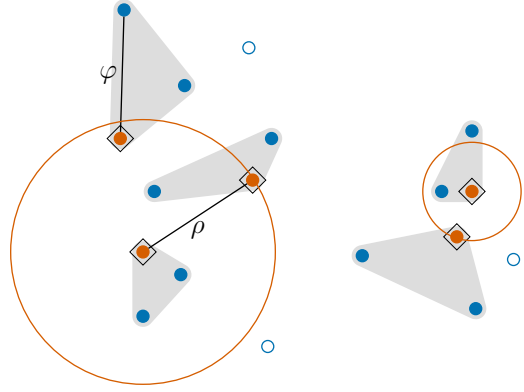

\Cref{alg:general-case-outlier-algorithm} describes the framework in detail. It uses the subroutines \textsc{compute-fairlets} and \textsc{compute-centers}. The latter is given in \Cref{alg-line:end-of-gonzalez-1:1} and consists of running the well-known farthest-first traversal on $\Representatives$. Gonzalez showed that farthest-first traversal yields a $2$-approximation for $k$-center that can be implemented to run in time $O(|\Representatives|\cdot k)$ \cite{Gonz85}. In \Cref{lem:compute-centers-2OPT-Gonzalez} (Appendix~\ref{sec:framework-lemmata}), we generalize the approximation result and show that the cost of the solution computed by farthest-first traversal on $A$ is bounded by twice the cost of any valid $k$-center solution $S$ on $P$. The difference here is that $S$ might use points as centers that are not present in $\Representatives$.
It is important that we run \textsc{compute-centers} on the set of anchors, not on $P$ (this might result in outliers being chosen as centers), and also not on the inlier points. The latter idea results in a worse approximation guarantee as shown in \Cref{sec:considerations:choosingcenters}. 
\begin{algorithm}
    \LinesNumbered
    \caption{\textsc{compute-centers} \cite{Gonz85}}
    \label{alg:compute-centers-Gonzalez}
    \SetKwInOut{Input}{Input}
    \SetKwInOut{Output}{Output}
    \BlankLine
    \Input{Set of points $\Representatives$, distances $d\colon A\times A\to \mathbb{R}_{\ge 0}$, number of clusters $k$}
    \Output{Set of centers $C$}
    \BlankLine
    pick $c_1$ arbitrarily from $\Representatives$\label{alg-line:gonzalez-pick-c1-1:1}\\
    \For{$i=2,\ldots,k$}{
        $c_i \gets \arg\max_{c\in \Representatives}d(c,\{c_1,\ldots,c_{i-1}\})$
    } \label{alg-line:end-of-gonzalez-1:1}
\end{algorithm}

The subroutine \textsc{compute-fairlets} differs depending on the case we are studying.
\Cref{fig:cluster-framework} illustrates the framework for $1:2$-fairness. We want to achieve $1:2$-fairness between the red and blue points. We use the red points as anchor points, which is convenient because every fairlet contains exactly one red point. As we discuss later in \Cref{sec:considerations:choosingcenters}, this choice is crucial for the approximation guarantee. We will give an example on which choosing the centers from the set of all inliers merely yields an \(8\)-approximation. The gray areas indicate the fairlet decomposition of cost $\varphi$, computed by \textsc{compute-fairlets}.
The circles indicate a vanilla $k$-center clustering of the red points with radius $\rho$, computed by \textsc{compute-centers}. 
The fairlet decomposition excludes three blue points as outliers.
We build a fair clustering by assigning each point in a fairlet to the center closest to its anchor. The fair clustering has radius $\varphi + \rho$.
The analysis of the framework is done in Appendix~\ref{sec:framework-lemmata}. We restate the cost result here.

\begin{restatable}{lemma}{CombinedCost}
    \label{lem:4-approximation}
    Assume that \textsc{compute-fairlets} in Line~\ref{alg-line:call-compute-fairlets} of \Cref{alg:general-case-outlier-algorithm} returns a set of fairlets $\F$ such that $\varphi \coloneqq \max_{f\in \F}\max_{p\in f}d(p,\anc(f))$. Let $A$ denote the set of anchors of the fairlets in $\F$. 
    Further, assume that $\textsc{compute-centers}$ in Line~\ref{alg-line:call-compute-centers} gives a set of centers $C\subseteq \Representatives$ such that $\rho \coloneqq \max_{a\in A}\min_{c\in C}d(a,c)$. 
    Then, \Cref{alg:general-case-outlier-algorithm} computes a fair $k$-center solution on $P$ with outliers $P\setminus P(\F)$ and cost $\varphi + \rho$ (also see \Cref{fig:cluster-framework}).
\end{restatable}

\subsection{\texorpdfstring{$1:t$}{1:t}-fairness}

The algorithm for this case follows~\cite{chierichetti2017fair,rosner2018privacy}. 
We have two colors, and the desired color ratio is $1:t$ for some integer $t\geq 1$, i.\,e., we want to achieve $|B\cap C| = t\cdot |R\cap C|$ for all clusters $C$.
Let us assume that $|B| \geq t|R|$. To guarantee a $1:t$-fair clustering, we exclude $|B|-t\cdot|R|$ blue outliers.
Fairlets consist of one red and $t$ blue points.
For the special case of $t=1$, fairlets consist of one blue and one red point each, and it suffices to compute a matching. If there are no outliers, i.\,e., if \(|B|=|R|\), choosing the anchors carefully from the set of all points yields an approximation guarantee of \(3\), which we explain in detail in \Cref{sec:appendix-1:1-fairness}.
However, for fair $k$-center with outliers, the same approach does not yield the same guarantee, but only a \(6\)-approximation. Instead, if we limit the set of potential anchors to the red points, we are able to achieve a \(4\)-approximation. As the approach for the \(1:1\)-case is a special case of the \(1:t\)-case, we directly proceed with the general $1:t$-case.
There, we have to match a red point to $t$ blue points, and the red point is the fairlet's anchor. We identify $R=A$ for the remainder of this section, and justify this decision further in \Cref{sec:considerations:choosingcenters}. The algorithm is described in \Cref{alg:compute-1-t-fairlets}.
\begin{algorithm}
    \LinesNumbered
    \caption{\textsc{compute-$1:t$-fairlets}}
    \label{alg:compute-1-t-fairlets}
    \SetKwInOut{Input}{Input}
    \SetKwInOut{Output}{Output}
    \BlankLine
    \Input{Point set $P$, $R\subseteq P$, $B\subseteq P$, distance metric $d$, integer $k$, threshold factor $\tau$}
    \Output{Set of fairlets $\F$, anchor assignment $\anc\colon \F\to R$}
    \BlankLine
    $\F \gets \emptyset$\\
    $D\gets \{d(x,y)\mid x,y\in P\}$\\
    sort $D$ increasingly\\
    \For{all $\delta \in D$}{
        construct flow network $N_{\delta}=(V, E(\tau\delta),v,w,c)$\\
        $\mathbf{f}\gets$ maximum integral flow in $N_{\delta}$ \label{alg-line:compute-max-flow-1:t}\\
        \If{$|\mathbf{f}|=t|R|$}{
            \For{all $r\in R$ \label{alg-line:start-constructing-fairlets-1:t}}{
            	$f_r \gets \{r\}$\label{alg-line:add-r-to-fairlet-1:t}\\
                $\anc(f_r) \gets r$\\
                \For{all $b$ such that $(b,r) \in E_B^R(\tau\delta)$}{
                    \If{$\mathbf{f}(b,r) = 1$}{
                        $f_r \gets f_r \cup \{b\}$ \label{alg-line:add-b-to-fairlet-1:t}
                    }
                }
                $\F \gets \F \cup \{f_r\}$ \label{alg-line:end-constructing-fairlets-1:t}
            }
            \Return{$\F,\anc$}
        }
    }
\end{algorithm}
To find a fairlet decomposition, it sets up a flow network $N=(V,E(T),v,w,c)$ as follows: 
We set $V = B \cup R \cup \{v,w\}$ and $E(T) = E_v^B \cup E_B^R(T) \cup E_R^w$, where $E_v^B = \{(v,b) \mid b \in B\}$, $E_B^R(T) = \{(b,r) \mid b\in B, r\in R, d(b,r) \leq T\}$, for some threshold $T\geq 0$, and $E_R^w = \{(r,w) \mid r\in R\}$. 
In general, we want $T=\tau\delta$ for some guess $\delta$ of the radius of an optimal solution, and some threshold factor $\tau$.
Let \(\OPT\) denote the cost of an optimal solution. As \(\OPT\) equals the pairwise distance between two points, it is \(\OPT\in D\) and the for-loop will eventually consider \(\delta = \OPT\). In this iteration the threshold becomes \(T=\tau\OPT\).
We choose \(\tau\) such that we can prove that we can construct a fairlet decomposition of cost \(\tau\OPT\), which depends on the fairness case. For the $1:t$ and $1:t_2:\ldots:t_m$ cases, we set $\tau=2$ because we can show that the cost of a fairlet decomposition that agrees with an optimal solution can be bounded by \(2\OPT\) if we use the anchor as a center, and we can find such a decomposition easily by a max flow computation. In the \(t_1:t_2:\ldots:t_m\) case, we need to set $\tau=12$ because of a more intricate procedure to compute the fairlets.
The capacities are given by 
\[ c(e) = \begin{cases}
    1 & \text{if } e\in E_v^B \cup E_B^R(T)\\
    t & \text{if } e\in E_R^w
\end{cases} \]
If an edge $(b,r)\in E_B^R(T)$ carries flow, this means that $b$ and $r$ belong to the same fairlet.

The overview of the analysis of this method is presented in Appendix~\ref{appendix:analysis-1-t}. We state the main cost result here; note that a more general result, including the running time, is given in the next section for the related multicolor case.
\begin{restatable}{theorem}{thmonetunbalanced}
\label{thm:1-t-unbalanced} 
    let \(t\in \mathbb{N}\) be an integer and let \(P=B\cup R\) be set of points consisting of blue and red points with \(|B|\ge t|R|\).
    \Cref{alg:general-case-outlier-algorithm} with \textsc{compute-centers} = \Cref{alg:compute-centers-Gonzalez} and \textsc{compute-fairlets} = \Cref{alg:compute-1-t-fairlets} computes a center-aware 4-approximation for $1:t$-fair $k$-center with $|B|-t|R|$ blue outliers.
\end{restatable}
The term \emph{center-aware} in this theorem refers to the property that a center point is assigned to itself. This natural property is sometimes violated by algorithms for fair clustering, i.\,e., it can happen that a cluster can have a center which is itself not part of the cluster. We say that the approximation algorithm is center-aware if it is ensured that this cannot happen.

\Cref{fig:comparison-4-approx} illustrates a worst-case scenario that matches the 4-approximation guarantee. It demonstrates that, even in the $1:1$ case, no better approximation can be guaranteed with this approach.
\Cref{fig:comparison-4-approx-opt} shows an optimal solution for $k=1$. \Cref{fig:comparison-4-approx-algo} illustrates a solution computed by \Cref{alg:general-case-outlier-algorithm} for the same set of points. 
In this case, each fairlet has a diameter of $2\OPT$.
The radius of the resulting clustering is four times larger than the optimal one. This example illustrates how the choice of fairlets and outliers can affect the quality of the resulting clustering solution.

\begin{figure}[htbp]
    \captionsetup[subfigure]{justification=centering}
    \centering
    \begin{subfigure}[t]{0.49\linewidth}
        \scalebox{0.5}{
\begin{tikzpicture}[x=1cm,y=1cm]
  \tikzset{
    transp/.style={circle, opacity=0, minimum size=0pt, inner sep=0pt},
    edge/.style   ={gray!70, line width=0.8pt}
  }

    \coordinate (L) at (0,0) {};
    \node[point=blue] (lb1) at (L) {};
    \node[point=blue] (lb2) at (1.2,0) {};
    \node[point=red] (lr1) at (0,1.2) {};
    \node[point=red] (lr2) at (0,-1.2) {};
    \node[outlier=blue] (lb3) at (0,2.4) {};
    \node[outlier=blue] (lb4) at (0,3.6) {};
    \node[outlier=blue] (lb5) at (0,-2.4) {};
    \node[outlier=blue] (lb6) at (0,-3.6) {};

    \draw (L) -- (lb2) node[draw=none,fill=none,font=\normalsize,midway,above] {$\OPT$};

    \draw (L) circle (1.2cm);
  


    \node[transp] at (0,6) {};
    \node[transp] at (-4.8,1.2) {};
    \node[transp] at (4.8,1.2) {};
    \node[transp] at (0,-4.2) {};

\end{tikzpicture}
}
        \subcaption{Optimal solution.}
        \label{fig:comparison-4-approx-opt}
    \end{subfigure}
    \begin{subfigure}[t]{0.49\linewidth}
        \scalebox{0.5}{
\begin{tikzpicture}[x=1cm,y=1cm]
  \tikzset{
    edge/.style   ={gray!70, line width=0.8pt}
  }

    \coordinate (L) at (0,0);
    \node[outlier=blue] (lb1) at (L) {};
    \node[outlier=blue] (lb2) at (1.2,0) {};
    \node[point=red, anchormark=black] (lr1) at (0,1.2) {};
    \node[point=red, anchormark=black] (lr2) at (0,-1.2) {};
    \node[outlier=blue] (lb3) at (0,2.4) {};
    \node[point=blue] (lb4) at (0,3.6) {};
    \node[outlier=blue] (lb5) at (0,-2.4) {};
    \node[point=blue] (lb6) at (0,-3.6) {};

    \draw (lr1) -- (4.8,1.2) node[draw=none,fill=none,font=\normalsize,midway,above] {$4\OPT$};

    \draw (lr1) circle (4.8cm);

    \fairletarea{lr1,lb4,lb4,lr1}
    \fairletarea{lr2,lb6}

\end{tikzpicture}
}
        \subcaption{A possible solution of \Cref{alg:general-case-outlier-algorithm}.}
        \label{fig:comparison-4-approx-algo}
    \end{subfigure}     
    \caption{Comparison of optimal solution and solution computed by \Cref{alg:general-case-outlier-algorithm} for $1:1$-fair $k$-center with $k=1$. The gray areas indicate fairlets, and the unfilled circles are outliers. The fairlets consist only of the two filled points contained within the gray areas.}
    \label{fig:comparison-4-approx}
\end{figure}
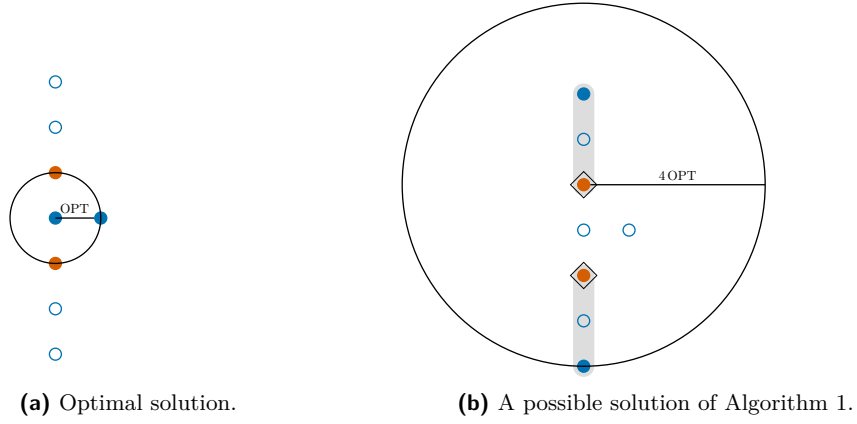

\subsection{\texorpdfstring{$1:t_2:\ldots:t_m$}{1:t2:...:tm}-fairness}

Now consider the scenario that $P$ can be partitioned into $m \geq 2$ colors: $P=H_1\cupdot \ldots\cupdot H_m$, and we aim for a color ratio of $1:t_2:\ldots:t_m$ for given values $t_2,\ldots,t_m$. We assume that $|H_i|\geq t_i|H_1|$ so that it suffices to exclude outliers from the sets $H_2,\ldots,H_m$. To achieve a fair clustering, we exclude $|H_i|-t_i|H_1|$ many points from $H_i$.  
Set $\Representatives = H_1$.
Fairlets consist of one point from $H_1$ and $t_i$ points from $H_i$ for all $2\leq i \leq m$. Similar to before, we will choose a fairlet's anchor from $H_1$. 
We can find a fairlet decomposition by finding a $1:t_i$-fairlet decomposition on $H_1\cup H_i$ for every $2\leq i \leq m$ and gluing together fairlets that share a point in $H_1$. This approach is visualized in \Cref{fig:multi-color-fairlets-gluing}, formalized in \Cref{alg:compute-1:t2:...:tm-fairlets} and its analysis can be found in Appendix~\ref{appendix:analysis-1-many-t}.

\begin{algorithm}
    \LinesNumbered
    \caption{\textsc{\texorpdfstring{compute-1:$t_2$:$\ldots$:$t_m$-fairlets}{compute-multiple-color-fairlets}}}
    \label{alg:compute-1:t2:...:tm-fairlets}
    \SetKwInOut{Input}{Input}
    \SetKwInOut{Output}{Output}
    \BlankLine
    \Input{Point set $P$, subsets $H_1,\ldots,H_m\subseteq P$, distance metric $d$, number of clusters $k$, threshold factor $\tau$}
    \Output{Set of fairlets $\F$, anchor assignment $\anc\colon \F\to H_1$}
    \BlankLine
    \For{$i=2,\ldots,m$}{
    	$\F_i \gets \textsc{compute-1:$t_i$-fairlets}(P,H_1,H_i,d,k,\tau)$ \label{alg-line:compute1:t1:...:tm-fairlets-call-compute-1:t-fairlets}
    }
    $\F=\emptyset$\\
    \For{$h_1\in H_1$}{
    	$f \gets \bigcup_{2\leq i\leq m}\{f'\mid f'\in F_i, h_1\in f'\}$ \label{alg-line:glue-small-fairlet-1:t1:...:tm}\\
        $\anc(f)\gets h_1$\\
    	$\F \gets \F \cup \{f\}$
    }
    \Return{$\F,\anc$}
\end{algorithm}

\begin{figure}
    \centering
    \scalebox{0.5}{
    \begin{tikzpicture}
        \node[point=green] (g1) at (-0.55, 2.15) {};
        \node[point=green] (g2) at (-1.50, 1.45) {};
        \node[point=green] (g3) at (-0.50, 1.55) {};
        \node[point=green] (g4) at (-0.35, 0.65) {};
        \node[point=blue] (b1) at (3.30, 1.75) {};
        \node[point=purple] (p1) at (0.95, 0.20) {};
        \node[point=purple] (p2) at (-1.0,-1.75) {};
        \node[point=red, anchormark=black] (r1) at (0,0) {};

        \fairletarea[gray, fill opacity=0.25][14pt]{r1,g2,g1}  
        \fairletarea[gray, fill opacity=0.25][14pt]{b1,r1}        
        \fairletarea[gray, fill opacity=0.25][14pt]{p1,p2,r1}        
    \end{tikzpicture}
}
    \caption{The gluing-process as described in \Cref{alg:compute-1:t2:...:tm-fairlets} for \(1:1:2:4\)-fairness. We identify the anchor color (red) and start by individually computing \(1:1\)-fairlets for the red and blue points, \(1:2\)-red-pink fairlets, and \(1:4\)-red-green-fairlets. The individual fairlets are indicated by gray areas, which overlap at one red point. The union of these fairlets forms one \(1:1:2:4\)-fairlet.}
    \label{fig:multi-color-fairlets-gluing}
\end{figure}
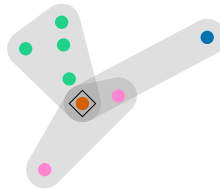

\begin{theorem}\label{thm:one-to-many}
	For $t_2,\ldots, t_m\in \mathbb N$ and a fair clustering instance satisfying $t_i|H_1|\leq |H_i|$ for all $2\leq i\leq m$, there exists a center-aware $4$-approximation for $1:t_2:\ldots: t_m$-fair $k$-center with $|H_i| - t_i|H_1|$ outliers in $H_i$ that runs in time $O(n^{2+o(1)}\log n)$.
\end{theorem}
The importance of \Cref{thm:one-to-many} lies in its practicability and broad applicability. Computing maximum flows in practice is a well-studied problem; the approximation ratio is decent, and we can apply the result to any input point set. To do so, set $t_1=1$ and $t_i = \lfloor \frac{|H_i|}{|H_1|}\rfloor$, which satisfies the precondition of the theorem. 
This can be applied to any input point set without further parameterization.

\section{General case, outliers in two colors and further considerations}\label{sec:considerations}

The general case is of lesser practical interest since it requires the computation of capacitated $k$-center solutions which is demanding in itself. We defer the algorithm and its analysis to Appendix~\ref{appendix:generalcase}. 
\begin{restatable}{theorem}{outliermulticolorgeneralcase} \label{thm:t1:...:tm-beginning}
    For $t_1, t_2,\ldots, t_m\in \mathbb N$ and a fair clustering instance satisfying $t_i|H_1|\leq t_1 |H_i|$ for all $2\leq i\leq m$ and $ \frac{|H_1|}{t_1}\in \mathbb{Z}$, 
    there exists a center-aware $14$-approximation for $t_1:t_2:\ldots:t_m$-fair $k$-center with $|H_i|-\frac{t_i}{t_1}|H_1|$ outliers from $H_i$ for all $2\le i \le m$ that runs in time $O(n^{2+o(1)}\log n)$.
\end{restatable}
The aforementioned results use outliers for colors $2$ to $m$ to achieve the designated ratios. 
We also report a result for a variant that allows outliers in color $1$. For this variant, the number of outliers $z$ is an input parameter. In the above results, $z$ is computed from the input point set to match the designated ratios. We study this concept only for the two-color case and balanced clusters. The algorithm and analysis can be found in Appendix~\ref{sec:bothsides}.

\begin{restatable}{theorem}{ThmOutBothSides}
\label{thm:outliers-both-sides}
    Let $|H_2|\geq|H_1|$. For any $z\geq |H_2|-|H_1|$ there exists an $11$-approximation for $1:1$-fair $k$-center with $z$ outliers from $H_1\cup H_2$ that runs in time $O(n^5)$.
\end{restatable}
We need the lower bound of $|H_2|-|H_1|$ on $z$ to guarantee that a $1:1$-fair solution exists. The clustering computed in \Cref{thm:outliers-both-sides} is not necessarily center-aware. In  Appendix~\ref{sec:bothsides}, we also discuss how to obtain a center-aware solution, which is a $14$-approximation.

\subsection{Choosing centers from the inlier set}\label{sec:considerations:choosingcenters}
As a final note from the analysis, let us discuss our choice of computing centers from $\Representatives$ which matches $H_1$. This means that we choose centers only from the minority color class rather than from all input points. At first glance, this seems counterintuitive: We restrict the choice of centers to a smaller set. We do this because we cannot use the complete input point set since we may choose centers that are later deemed outliers. Indeed, this is quite likely since farthest-first traversal favours points that are far away. We could use the set of inliers to compute the centers. However, our inliers do not necessarily match the inliers of an optimum solution. Quite surprisingly, choosing the centers from the inliers -- a larger set than the set $\Representatives$ -- yields a worse approximation ratio. 
\begin{figure}[htbp]
     \centering
     \captionsetup[subfigure]{justification=centering}
    \begin{subfigure}[t]{.3\textwidth}
        \centering
        \scalebox{0.4}{
\begin{tikzpicture}[baseline,x=1cm,y=1cm,rotate=-90]
  \tikzset{
    transp/.style={circle, opacity=0, minimum size=0pt, inner sep=0pt},
    edge/.style   ={gray!70, line width=0.8pt}
  }
  \draw[opacity=0] (-3,-6) rectangle (18,6); 

    \coordinate (L) at (0,0) {};
    \node[point=blue] (lb1) at (L) {};
    \node[point=blue] (lb2) at (-1.2,0) {};
    \node[point=red] (lr1) at (0,1.2) {};
    \node[point=red] (lr2) at (0,-1.2) {};
    \node[outlier=blue] (lb3) at (0,2.4) {};
    \node[outlier=blue] (lb4) at (0,3.6) {};
    \node[outlier=blue] (lb5) at (0,-2.4) {};
    \node[outlier=blue] (lb6) at (0,-3.6) {};

    \draw (L) -- (lb2) node[draw=none,fill=none,font=\normalsize,midway,left] {$\OPT$};

    \draw (L) circle (1.2cm);    

    \coordinate (M) at (7.2,0) {};
    \node[outlier=blue, label={[label distance=.1cm, font=\huge]0:$a$}] (mb1) at (M) {};
    \node[point=blue] (mb2) at ($(M)+(155:1.2cm)$) {};
    \node[point=red] (mr1) at ($(M)+(335:1.2cm)$) {};

    \draw (M) circle (1.2cm); 
    
    \coordinate (R) at (14.4,0);
    \node[point=blue] (rb1) at (R) {};
    \node[point=blue] (rb2) at ($(R)+(1.2,0)$) {};
    \node[point=red] (rr1) at ($(R)+(0,-1.2)$) {};
    \node[point=red] (rr2) at ($(R)+(0,1.2)$) {};
    \node[outlier=blue] (rb3) at ($(R)+(0,2.4)$) {};
    \node[outlier=blue] (rb4) at ($(R)+(0,3.6)$) {};
    \node[outlier=blue] (rb5) at ($(R)+(0,-2.4)$) {};
    \node[outlier=blue] (rb6) at ($(R)+(0,-3.6)$) {};

    \draw (R) circle (1.2cm);




\end{tikzpicture}
}
        \caption{Optimal solution.}
    \end{subfigure}\hfill
    \begin{subfigure}[t]{.6\textwidth}
        \centering
        \scalebox{0.4}{
\begin{tikzpicture}[baseline,x=1cm,y=1cm,rotate=-90]
  \tikzset{
    transp/.style={circle, opacity=0, minimum size=0pt, inner sep=0pt},
    edge/.style   ={gray!70, line width=0.8pt}
  }
  \clip (-3,-8) rectangle (18,8);

    \coordinate (L) at (0,0) {};
    \node[outlier=blue] (lb1) at (L) {};
    \node[outlier=blue] (lb2) at (-1.2,0) {};
    \node[point=red, anchormark=black] (lr1) at (0,1.2) {};
    \node[point=red, anchormark=black] (lr2) at (0,-1.2) {};
    \node[outlier=blue] (lb3) at (0,2.4) {};
    \node[point=blue, label={[label distance=.2cm, font=\huge]90:$c_1$}] (lb4) at (0,3.6) {};
    \node[outlier=blue] (lb5) at (0,-2.4) {};
    \node[point=blue] (lb6) at (0,-3.6) {};

    \fairletarea{lb4,lr1};   
    \fairletarea{lb6,lr2};


    \coordinate (M) at (7.2,0) {};
    \node[outlier=blue, label={[label distance=.1cm, font=\huge]-25:$a$}] (mb1) at (M) {};
    \node[point=blue] (mb2) at ($(M)+(155:1.2cm)$) {};
    \node[point=red, anchormark=black] (mr1) at ($(M)+(335:1.2cm)$) {};

    \draw (mb2) -- (mb1) node[draw=none,fill=none,font=\Large,midway, left, sloped] {$\OPT$};
    \draw (mb1) -- (mr1) node[draw=none,fill=none,font=\Large,midway, left, sloped] {$\OPT$};

    \fairletarea{mb2,mr1};

    \coordinate (R) at (14.4,0);
    \node[outlier=blue] (rb1) at (R) {};
    \node[outlier=blue] (rb2) at ($(R)+(1,0)$) {};
    \node[point=red, anchormark=black] (rr1) at ($(R)+(0,-1.2)$) {};
    \node[point=red, anchormark=black] (rr2) at ($(R)+(0,1.2)$) {};
    \node[outlier=blue] (rb3) at ($(R)+(0,2.4)$) {};
    \node[point=blue, label={[label distance=.2cm, font=\huge]270:$c_3$}] (rb4) at ($(R)+(0,3.6)$) {};
    \node[outlier=blue] (rb5) at ($(R)+(0,-2.4)$) {};
    \node[point=blue, label={[label distance=.2cm, font=\huge]270:$c_2$}] (rb6) at ($(R)+(0,-3.6)$) {};

    \fairletarea{rb6,rr1};
    \fairletarea{rb4,rr2};

    \draw (lb4) circle (9.4cm);
    \draw (lb4) -- ($(lb4)+(9.4,0)$) node[draw=none,fill=none,font=\huge,midway,above,sloped,rotate=90] {$8\OPT-\varepsilon$};

    \draw (rb4) circle (2.4cm);
    \draw (rb6) circle (2.4cm);

    \draw (rb6) -- (mr1) node[draw=none,fill=none,font=\huge,midway, sloped, above right, rotate=90] {$6\OPT-\varepsilon$};

    \draw (lb4) -- (mb2) node[draw=none,fill=none,font=\huge,midway,sloped, rotate=90, above] {$6\OPT-\varepsilon$};


\end{tikzpicture}
}
        \caption{A possible solution computed by the algorithm that runs \textsc{compute-centers} on $P(\F)$.}
    \end{subfigure}
    \caption{Comparison of optimal solution and solution computed by the algorithm that runs \textsc{compute-centers} on $P(\F)$ for 1:1-fair $k$-center with outliers for $k=3$. The unfilled circles indicate outliers.}
    \label{fig:8-approx}
\end{figure}
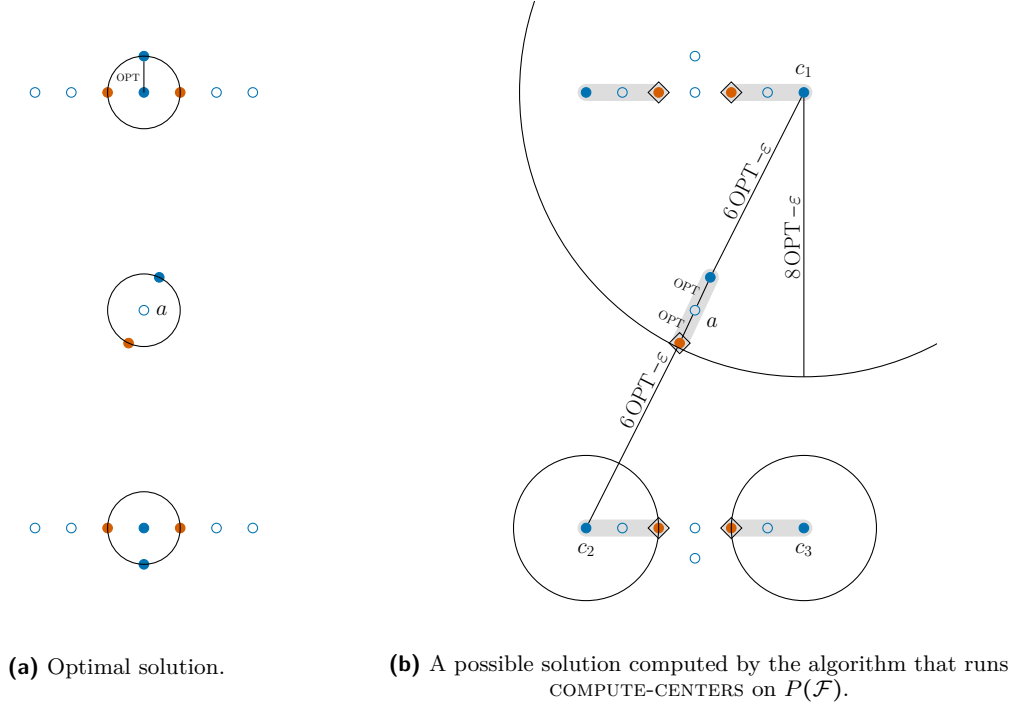
Assume that we use \Cref{alg:general-case-outlier-algorithm}, but with the slight modification that the set of centers $C$ is computed 
\[C\gets \textsc{compute-centers}(P(\F),d,k)\] 
in Line 3. \Cref{fig:8-approx} shows an instance on which the algorithm only yields an $8$-approximation.
\Cref{fig:8-approx} (a) shows an optimal solution for $k=3$ with radius $\OPT$ and \Cref{fig:8-approx} (b) shows a potential solution computed by the considered algorithm variant with radius $8\OPT$. 
The problem, as compared to the optimal solution, is two-fold.
First, the fairlets are computed. 
Here, we can only guarantee a cost of $2\OPT$. As $\OPT$ is the smallest pairwise distance between any pair of points, the network for the max-flow computation does not differentiate between point pairs of distance $\OPT$ and distance $2\OPT$. Therefore, a fairlet decomposition as depicted in \Cref{fig:8-approx} might indeed be found.
Then, the farthest-first traversal on $P(\F)$ chooses the centers $c_1,c_2$, and $c_3$, in this order. In this construction, $d(c_1,a) = d(c_3,a) = d(c_2,a)=7\OPT-\varepsilon$, so it does not matter to which center the fairlet with anchor $a$ is assigned. Assigning it to $c_1$ leads to an overall clustering cost of $8\OPT-\varepsilon$.
In \Cref{sec:appendix-choice-of-centers-8-approx}, we complement this result by showing that the approximation factor is indeed bounded from above by $8$.

\section{Experimental Evaluation}

We have argued in the previous sections that fair clustering algorithms often assume an ideal distribution of the protected attributes (e.g., a \(1:1\) male to female ratio) while realistic data often deviates from these ideal ratios.
Previous work has dealt with this problem by randomly discarding input points (see, e.g., \cite{chierichetti2017fair}) in order to establish the exact ratio in the input.
We instead hypothesized that an informed choice of discarded outliers yields better results.
In order to test the hypothesis, we compare the resulting clustering cost of~\cite{chierichetti2017fair} with random discarded outliers to the clustering cost resulting from our algorithms with an informed choice of outliers in computational experiments.
We evaluate our fair $k$-center with outlier algorithms for the case that $t_1=1$, both for two colors and for multiple colors. Our evaluation is run on four datasets, \textit{bank} \cite{misc_bank_marketing_222}, \textit{census} \cite{misc_adult_2}, \textit{diabetes} \cite{misc_diabetes_130_richtig} and \textit{income} \cite{ding2021retiring}.

\begin{table}[htbp]
    \caption{Dataset parameters. $n$: subsample size, $d$: number of features, No. of test sets: number of test sets created from this dataset, Groups: The number of different attribute values of the protected attribute, Protected attribute: the name in the original dataset.}
    \label{datasettable}
    \centering
    \begin{tabular}{lrrrrl}
    \toprule
    Dataset   & $n$ & $d$ & No. of test sets & Groups & Protected attribute \\
    \midrule
    Bank        & 1{,}000 & 7  & 45  & 2 & married \\
    Census      &   600   & 6  & 54  & 2 & sex     \\
    Diabetes    & 1{,}000 & 15 & 91  & 2 & gender  \\
    Income      & 1{,}000 & 9  & 195 & 2 & SEX     \\
    Census-Race &  600    & 6 & 54 & 5 & race \\
    Diabetes-Race & 1{,}000 & 15 & 91 & 6 & race\\
    \bottomrule
    \end{tabular}
\end{table}

Our experimental setup mirrors the setup of the reference experiments in~\cite{chierichetti2017fair} where possible: We use the same datasets Bank, Census and Diabetes and the same protected attributes. Since the reference only exists for the \(1:1\) case, we generalize accordingly where appropriate.
Like~\cite{chierichetti2017fair}, we subdivide each dataset into numerous parts in order to create a large database for our evaluation and smooth out artifacts. The total number of test datasets is $385$, see Table~\ref{datasettable} (note that two datasets appear twice since we evaluate two different protected attributes). In order to create the test datasets, we randomized the order of the datasets and then cut it into pieces. In this fashion, all data present in the dataset is used, the expected ratio of each set matches the ratio of the overall dataset, but the individual test sets are all a bit different.

Before splitting the dataset, the data is normalized: We first subtract the mean value from each attribute, then scale the values such that all attribute values lie in the interval $[-1,1]$. We do this because the datasets contain columns with very different ranges.
For each test set, we evaluate the algorithms for all $k\in\{1, 2, 5, 10, 20, 30, 40, 50, 100\}$.
Since the algorithms are not randomized, each test set is evaluated once.

\subparagraph*{Algorithms.}
We implement \Cref{alg:general-case-outlier-algorithm} for fair \(k\)-center with outliers. In the two-color $1:t$ case, we use \Cref{alg:compute-1-t-fairlets} to compute the fairlets; in the multi-color $1:t_2:\ldots:t_m$ case, we use \Cref{alg:compute-1:t2:...:tm-fairlets}.
Our algorithm implicitly identifies a set of $z$ outliers during the fairlet decomposition, see below. 
In the two-color case, we compare ourselves to the algorithm proposed by Chierichetti et al. \cite{chierichetti2017fair} for fair \(k\)-center without outliers. This algorithm assumes that the given point set satisfies the desired color ratio; hence, we first identify the number of points $z$ needed to exclude from the majority class to achieve a balanced point set and subsample the points accordingly, i.\,e., we sample $n-z$ blue points uniformly at random and discard the other points.   

\subparagraph*{Choosing the desired ratios for fair $k$-clustering with outliers automatically.} Notice that we can \emph{choose} the $t_i$ for the specific input data at hand. In the exact representation model, this is done by setting $t_i = |H_i| / \gcd(|H_1|, \ldots, |H_m|)$  for all $i\leq m$. This results in no outliers.
We want to use outliers to allow for reasonable clusterings in cases where the input $|H_i|$ are, for example, coprime. For illustration, consider a dataset with \(|H_1|=100\) and \(|H_2|=201\). In the exact representation model without outliers, we would have \(t_1=100\) and \(t_2=201\) and a fairlet decomposition would just consist of one fairlet containing all points. However, the intuitively correct ratio is \(1:2\), which would lead to much smaller fairlets of size \(3\) each and likely more meaningful clusterings. Removing only \(1\) outlier allows for such a ratio. In the following, we formalize this intuition:
Recall that we assume $|H_1|\le|H_2|\le \ldots, |H_m|$. An automatic way to choose the $t_i$ in a beneficial and feasible way is to set 
\[
t_1=1 \text{ and } t_i = \lfloor |H_i| / |H_1| \rfloor
\]
This results in a $1:t_2:\ldots:t_m$-fair $k$-clustering problem with outliers, and it can be applied to any input dataset in an automatic fashion. We always get ratios with $t_1=1$ in this fashion, thus we can apply the above mentioned special case algorithms.
In \Cref{appendixinputandoutputratios}, we report the input ratios of all test sets and the ratios to which these were rounded down.

\subparagraph{Cost comparison for two colors.}
We compare the clustering cost of our method to the method by Chierichetti et al.~\cite{chierichetti2017fair} in Figure~\ref{clusteringcost-comparison}. For each $k$, we compare the mean, min and max clustering cost over all subinstances. In each subinstance, we cluster to the rounded color ratio as described above. For the input of~\cite{chierichetti2017fair}, we randomly sample outliers from the majority color to achieve the idealized ratio. The comparison shows the clustering cost by the more informed outlier selection of our method (blue) vs.~\cite{chierichetti2017fair} (orange). The diagrams also report the fairlet cost as it is often observed that this dominates the final cost in fair clustering, an effect that can also be seen here.
As expected, our method gives lower clustering cost.

\begin{figure}
 \includegraphics[width=.49\textwidth]{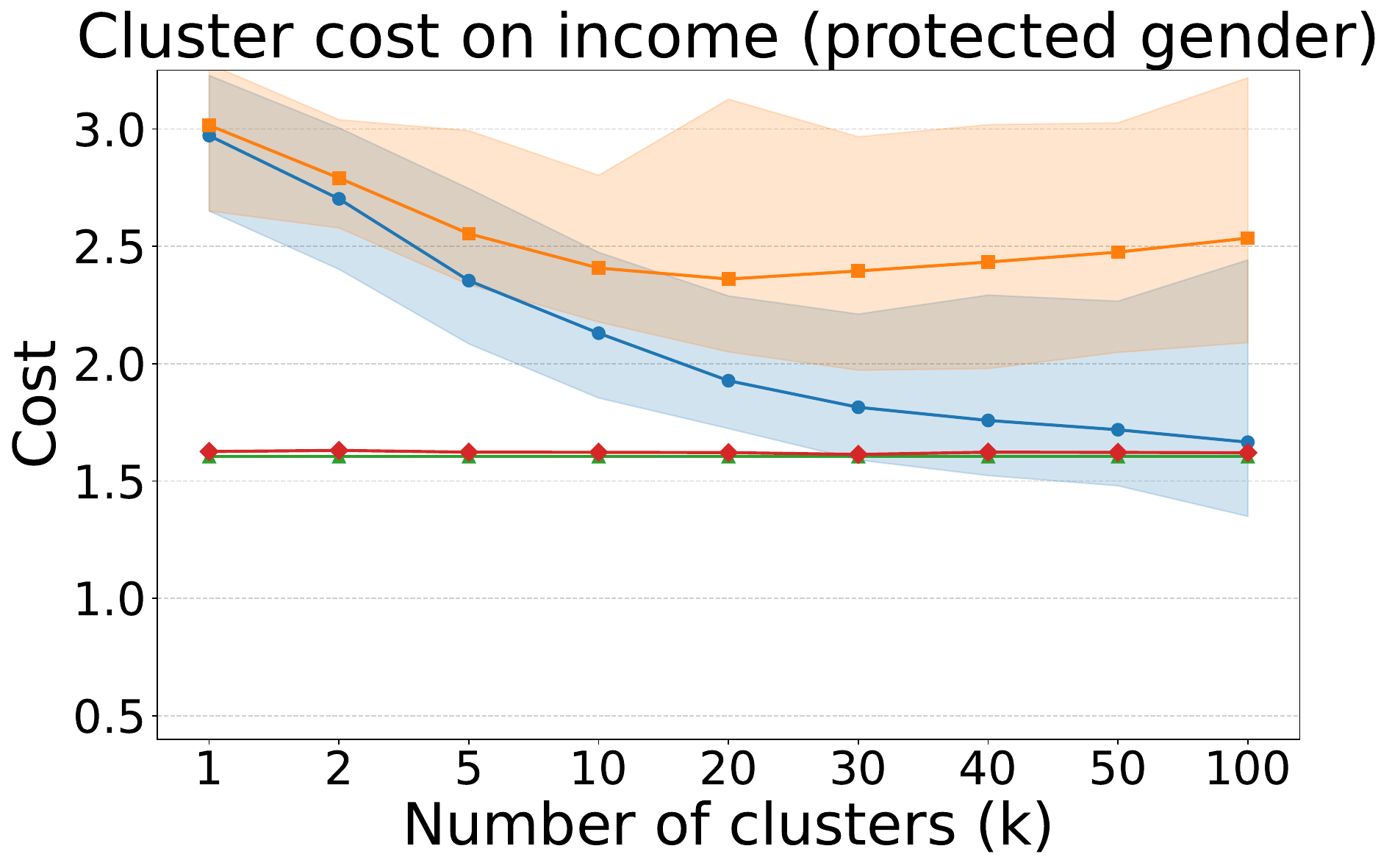}
 \hfill
  \includegraphics[width=.49\textwidth]{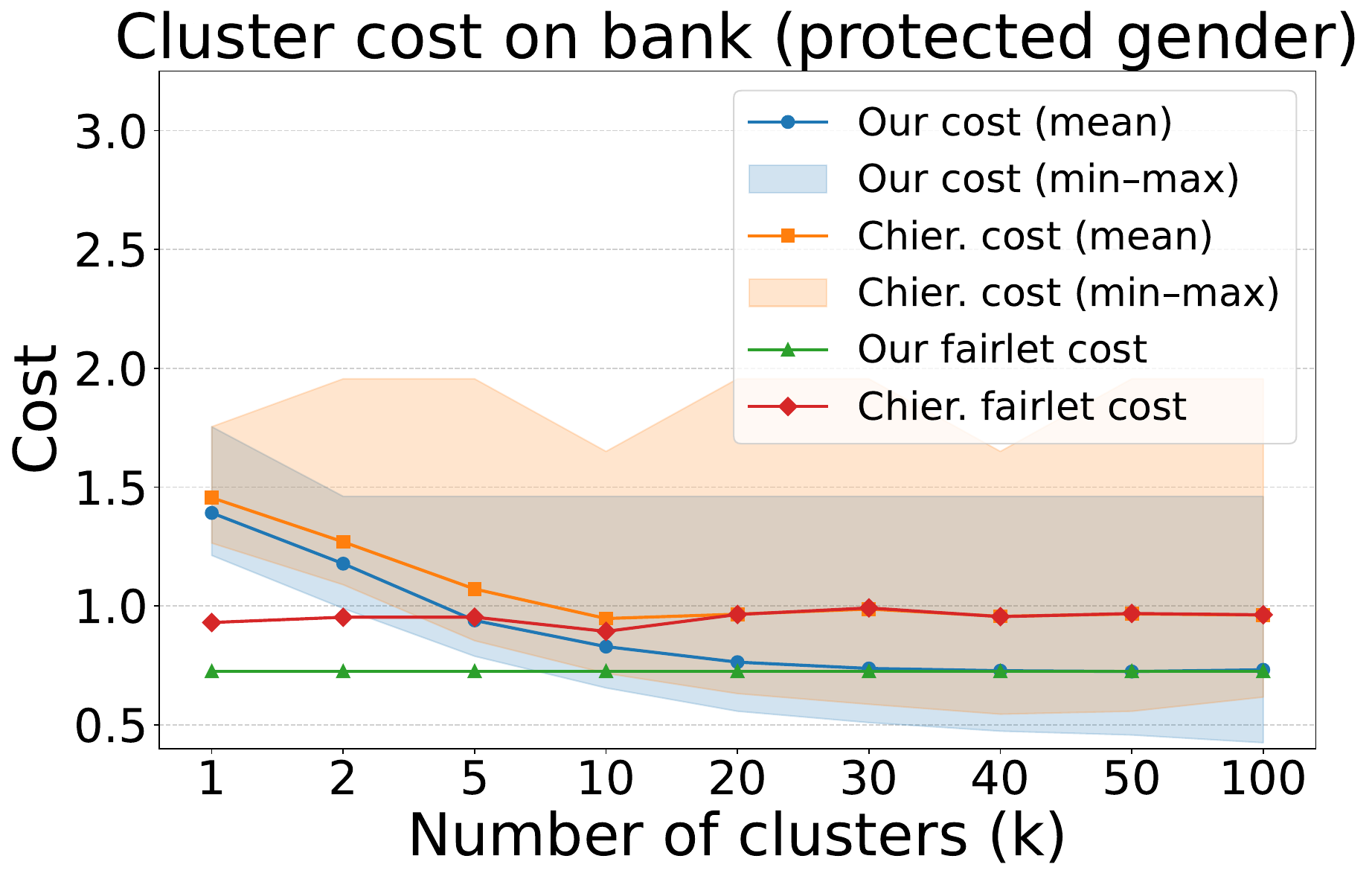}
 \hspace{0.5cm}

 \includegraphics[width=.49\textwidth]{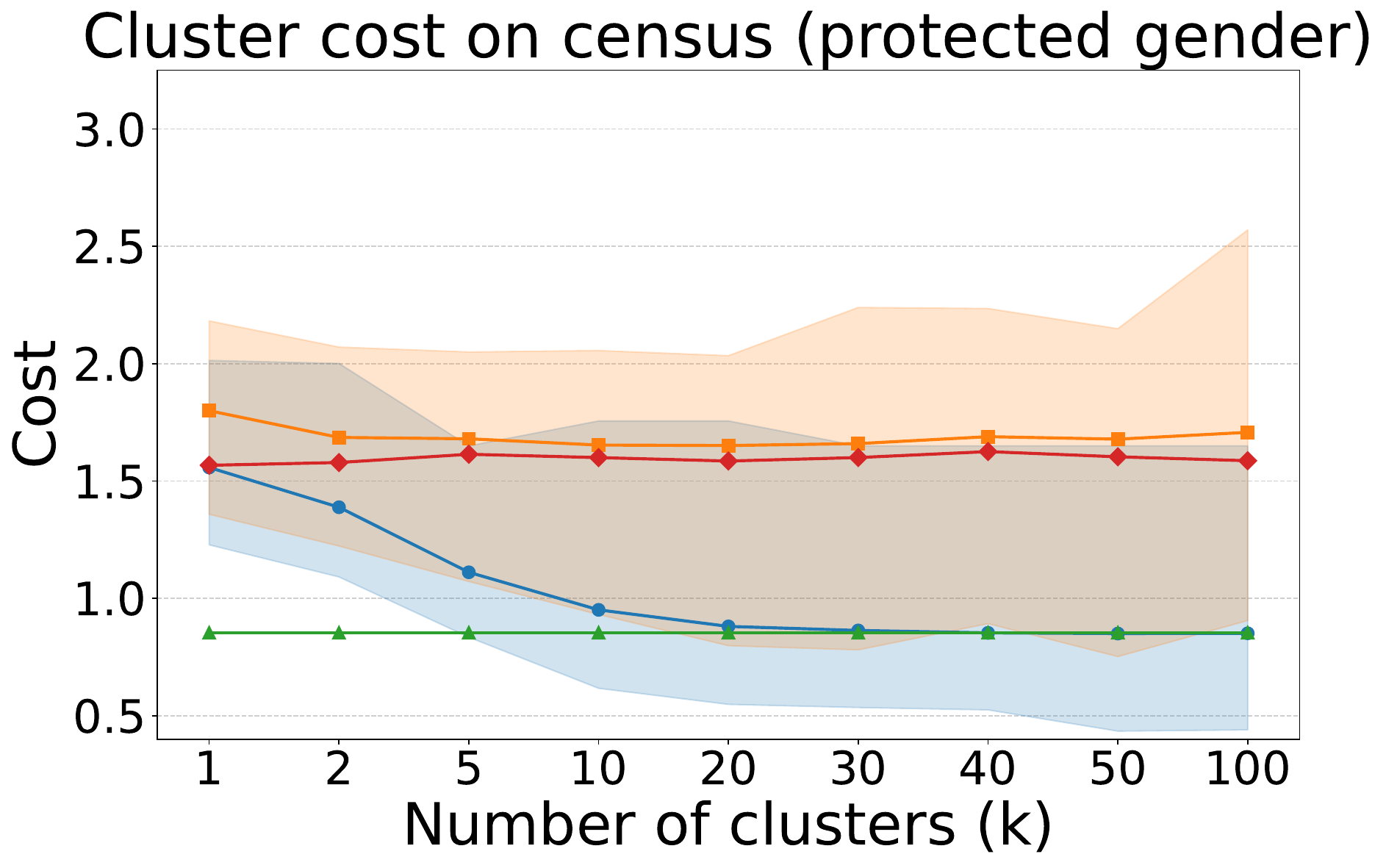}
 \hfill
 \includegraphics[width=.49\textwidth]{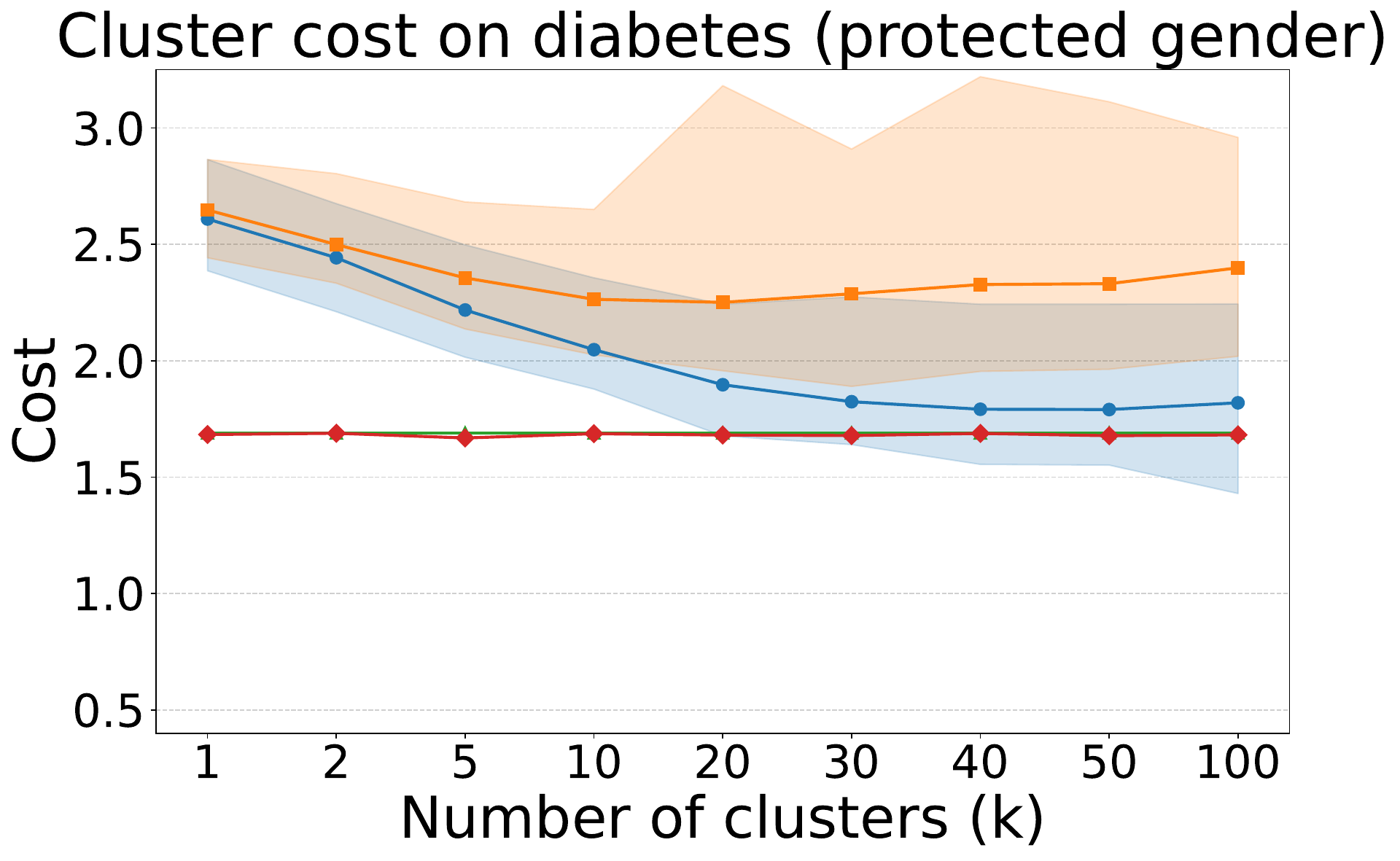}

 \caption{Comparison of the clustering cost of our method to random outlier selection.\label{clusteringcost-comparison}}
\end{figure}

\subparagraph{Effect of center choices on the clustering cost.}
In Figure~\ref{centerchoice}, we demonstrate that our findings on the center choice are not merely theoretical. 
As opposed to Figure~\ref{clusteringcost-comparison}, we now restrict~\cite{chierichetti2017fair} to \emph{only} use minority points as anchors, this forcing this algorithm to select minority points as centers. We observe a sizable decrease of the clustering cost. This is in line with our theoretical results. We observe that our method still gives better clustering cost.

\begin{figure}[hp]
 \includegraphics[width=.49\textwidth]{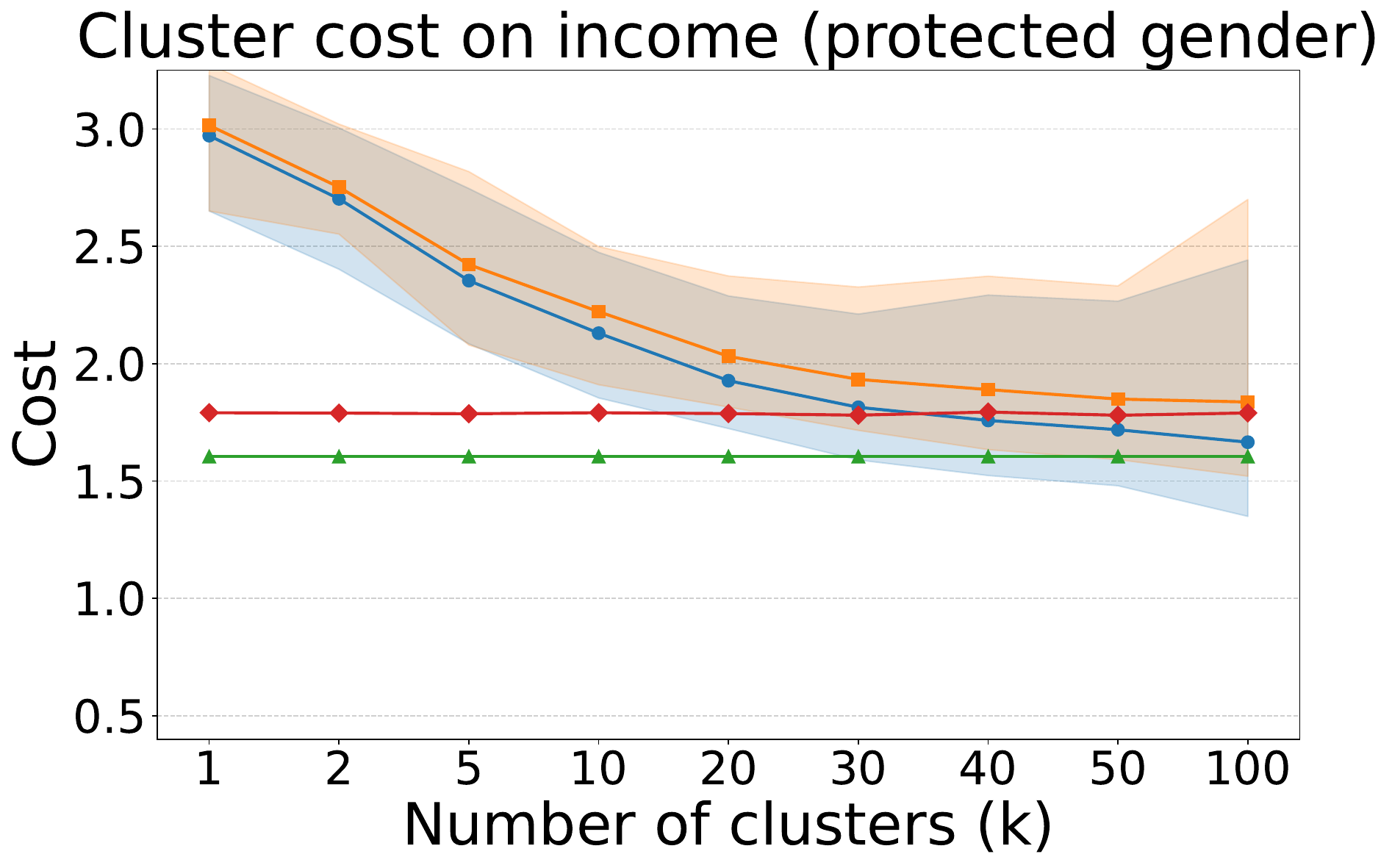}
 \hfill
  \includegraphics[width=.49\textwidth]{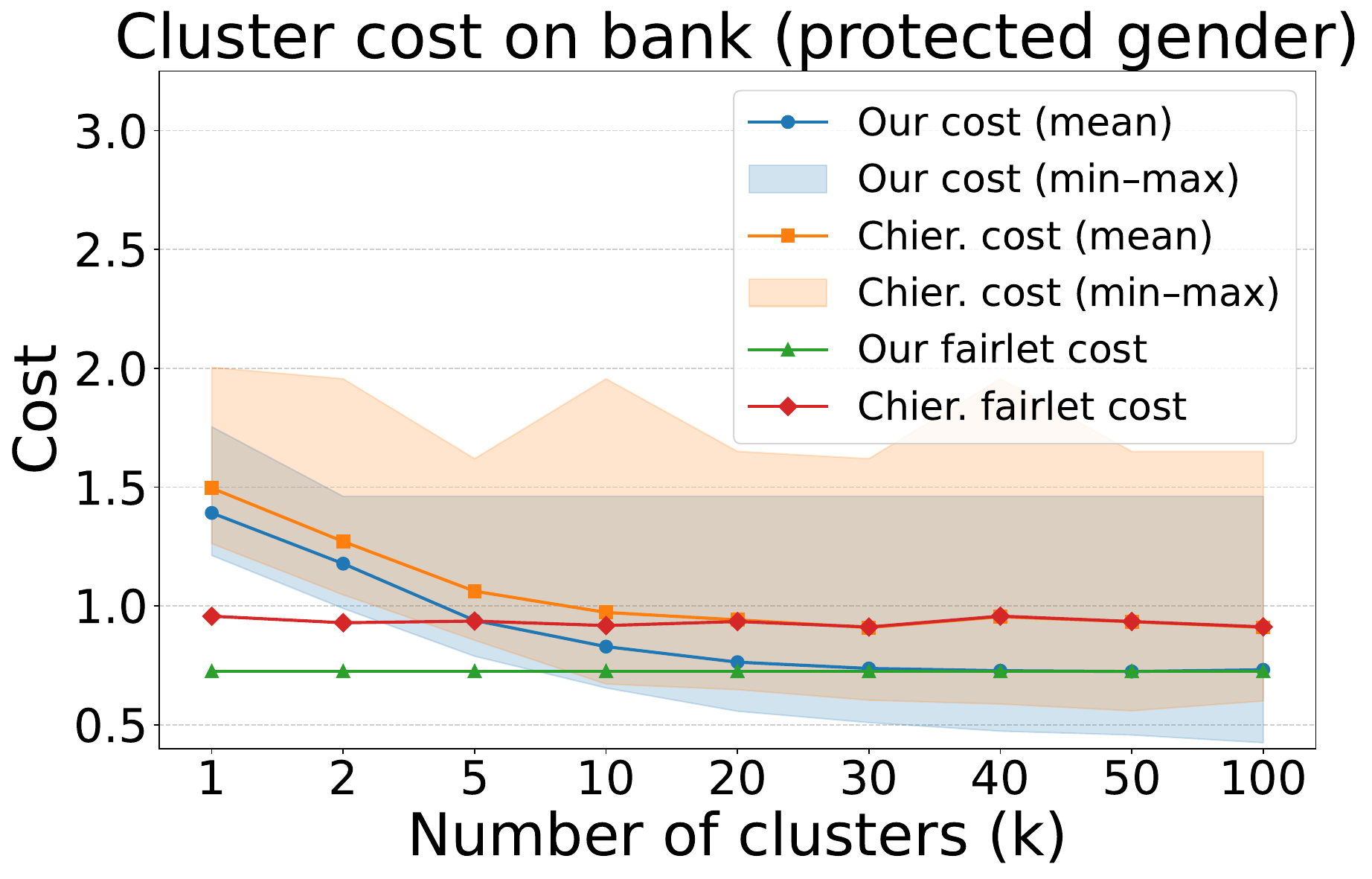}
 \hspace{0.5cm}

 \includegraphics[width=.49\textwidth]{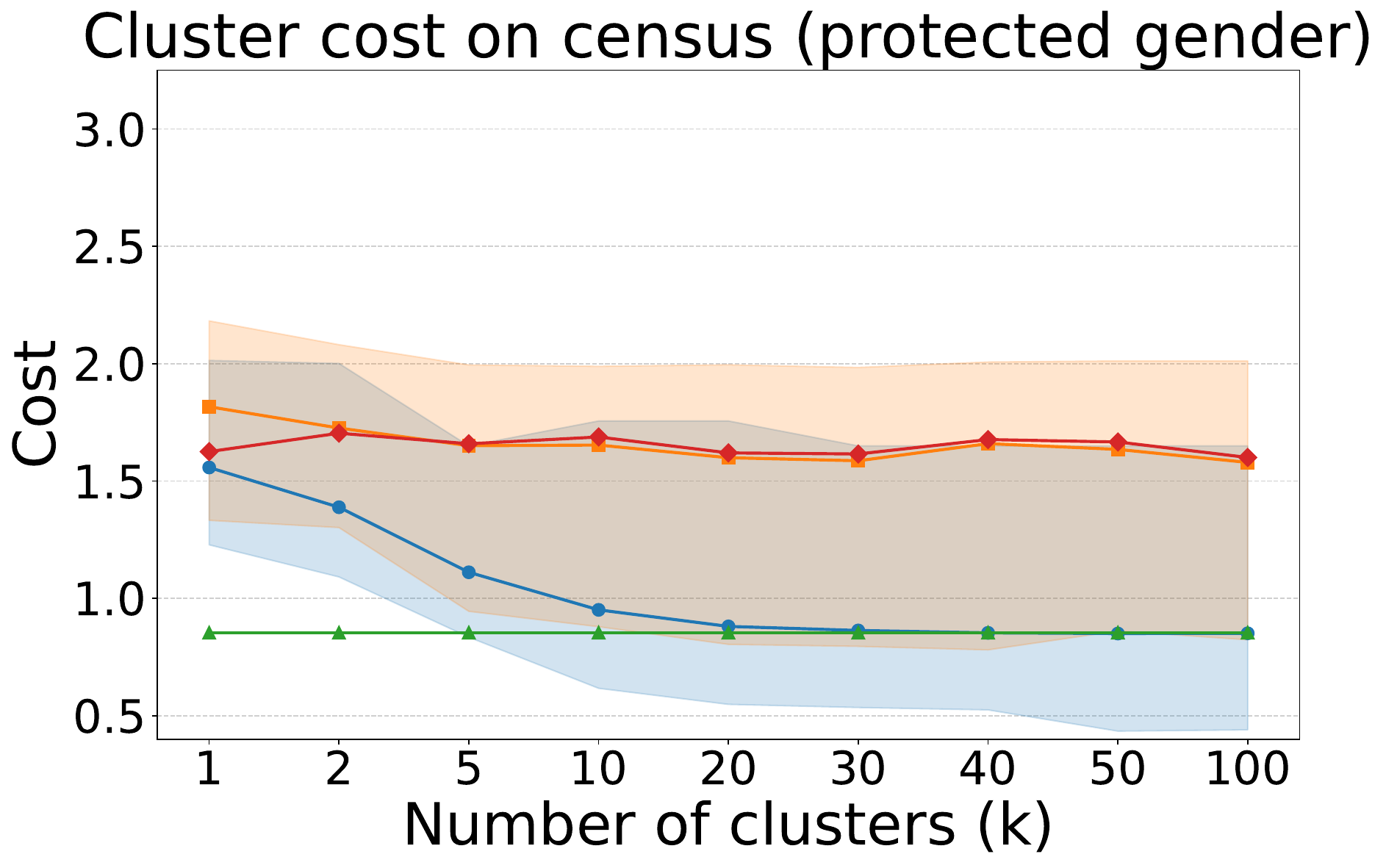}
 \hfill
 \includegraphics[width=.49\textwidth]{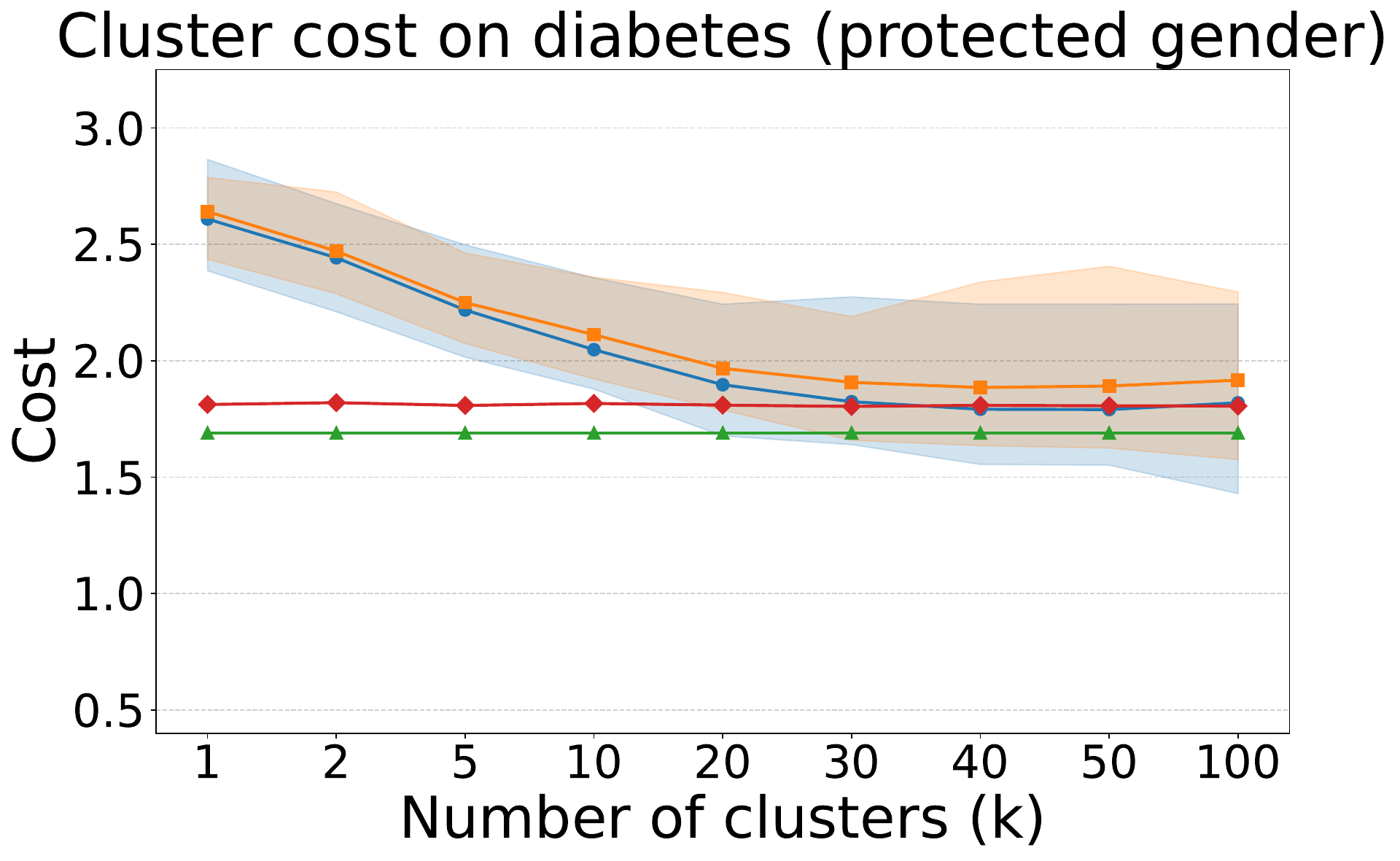}
\hspace{0.5cm}
\caption{The effect of center choices. \label{centerchoice}}
\end{figure}

\subparagraph{Multicolor results.}
\begin{figure}[hp]
  \includegraphics[width=.49\textwidth]{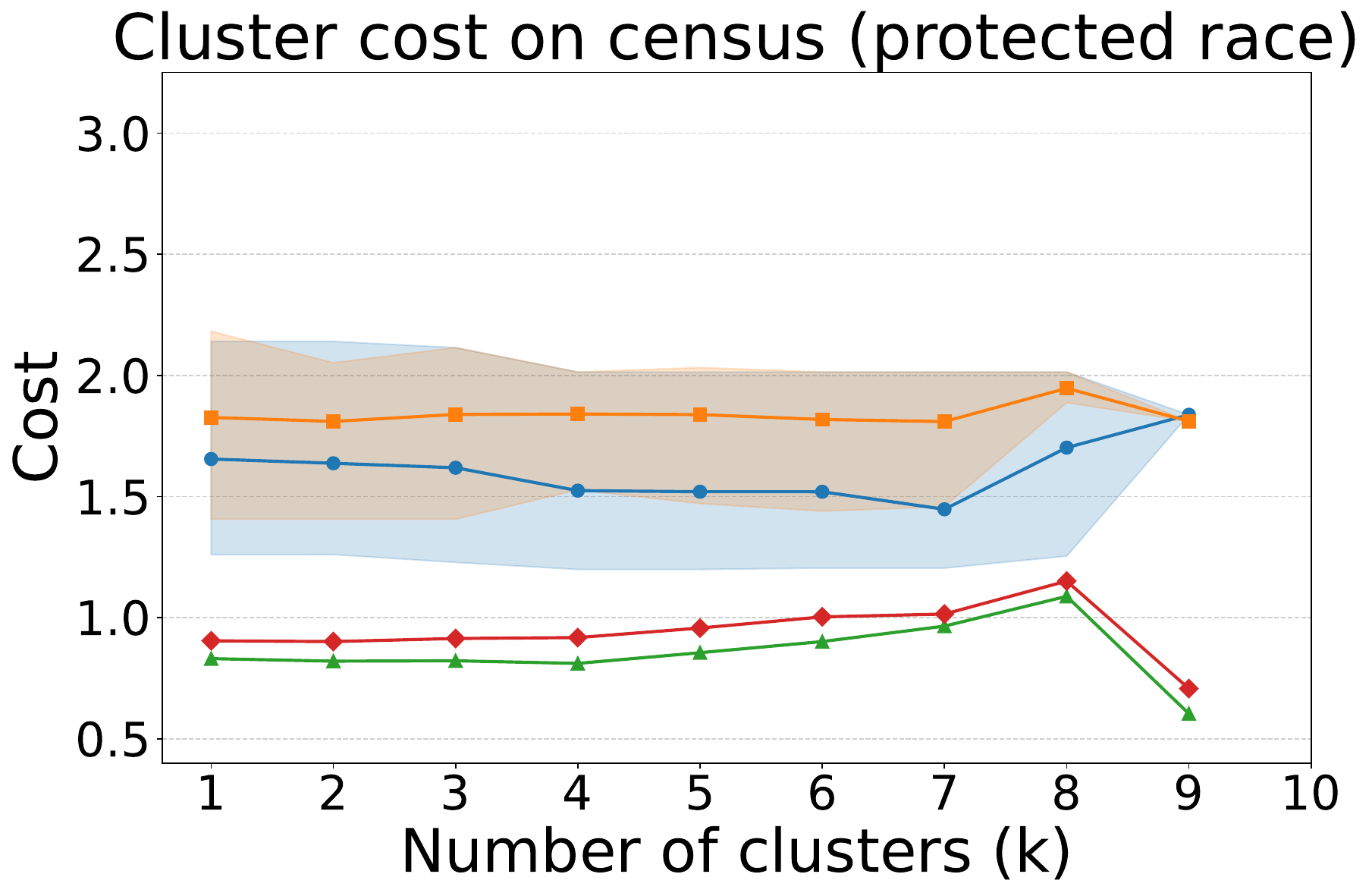}
 \hfill
 \includegraphics[width=.49\textwidth]{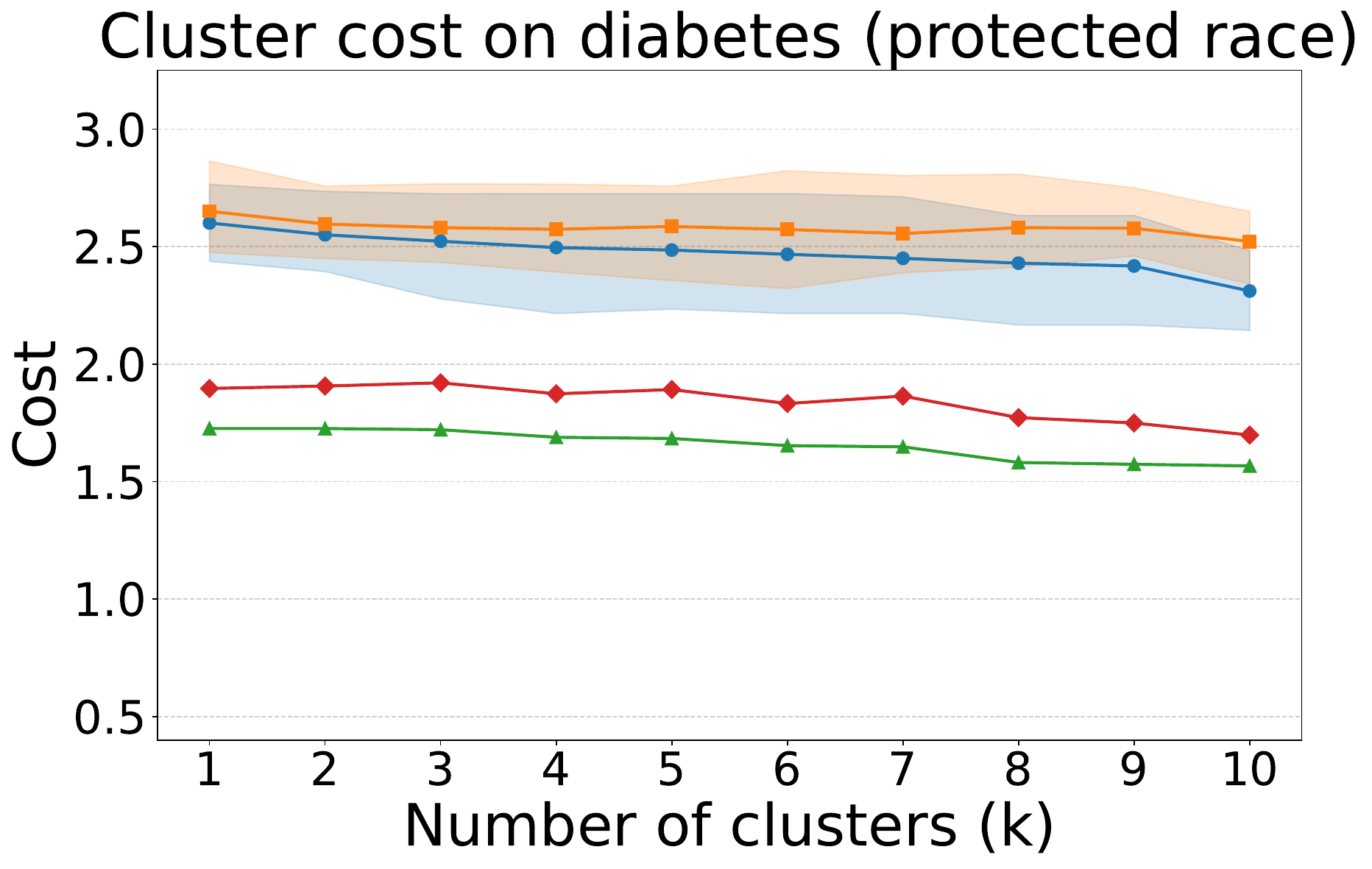}
\caption{Cluster cost comparison for multiple colors. \label{multicolor}}
\end{figure}

As in Figure~\ref{clusteringcost-comparison}, we evaluate how the clustering cost compares if we run our method and the method by~\cite{chierichetti2017fair} combined with random outlier selection. We cannot form as many clusters because the uneven distribution of the attribute values means that fairlets are large even after rounding the ratios. This is the case for \textsc{census} and \textsc{diabetes} with protected attribute \textsc{race} (also see \Cref{census-race} and \Cref{diabetes-race} in \Cref{appendixinputandoutputratios} for evidence)
where the minority group has very few points. This restricts the number of fairlets that we can build and the number of clusters that we can obtain. The result is depicted in Figure~\ref{multicolor}.
We report the raw data for the plots in Figure~\ref{clusteringcost-comparison} and \ref{centerchoice} in \Cref{appendix:costvalues}.

\clearpage
\bibliography{references}

\newpage
\appendix 
\section{Analysis of the framework}\label{sec:framework-lemmata}

\begin{restatable}[\cite{Gonz85}]{lemma}{GonzalezTwoApprox}
    \label{lem:compute-centers-2OPT-Gonzalez}
    Let $S$ be the radius of a valid $k$-center solution on $P$. Then, \textsc{compute-centers} (\Cref{alg:compute-centers-Gonzalez}) returns a set of centers $C\subseteq \Representatives$ such that $\max_{a\in \Representatives}\min_{c\in C}d(a,c)\leq 2S$.
\end{restatable}
\iffullappendix
\begin{proof}
    The centers $c'_1,\ldots,c'_k\in \Representatives$ are computed by farthest-first traversal on $\Representatives$. Let $c'_{k+1}\in \Representatives$ be a point that is farthest away from $\{c'_1,\ldots,c'_k\}$, i.e., $c'_{k+1}\in \arg\max_{p\in \Representatives}\min_{1\le i\le k}d(p,c'_i)$. 
    Consider a valid $k$-center clustering $\C$ on $P$ with value $S$. 
    By the pigeonhole principle, there exist two different points $c'_i,c'_j\in \{c'_1,\ldots,c'_{k+1}\}$ that are contained within the same cluster of $\mathcal{C}$. By triangle inequality and because the maximal radius of $\mathcal{C}$ is $S$, it follows that $d(c'_i,c'_j) \leq 2S$. Therefore, $\max_{a\in \Representatives}\min_{1\leq i \leq k}d(a,c'_i) = \min_{1\leq i \leq k}d(c'_{k+1},c'_i) \leq 2S$.
\end{proof}

\CombinedCost*
\begin{proof}
	Let $p\in P(\F)$. Then there exists a fairlet $f\in \F$ that contains $p$. Let $c_f = \arg\min_{c\in C}d(c,\anc(f))$ be the center closest to its anchor. It is $d(p,\anc(f))\leq \varphi$ by the assumption on \textsc{compute-fairlets} and $d(\anc(f),c_f) \leq \rho$ by the assumption on \textsc{compute-centers}. Hence, 
	\[ d(p,\alpha(p)) = d(p, c_f) \leq d(p,\anc(f)) + d(\anc(f),c_f) \leq \varphi + \rho.  \]
	The clustering $(C,\alpha)$ is fair as every cluster equals a union of fairlets. In the end, every fairlet is part of a cluster. Hence, by assumption, the solution excludes outliers $P\setminus P(\F)$.
\end{proof}
\fi
\begin{restatable}{lemma}{RunningTimeFramework}
    \label{lem:running-time-framework-outlier}
    Let $\F$ be the fairlet decomposition computed by $\textsc{compute-fairlets}(P,d,k)$ \linebreak in \Cref{alg-line:call-compute-fairlets}.
    Let $T_{\F}$ and $T_{C}$ denote the running times of $\textsc{compute-fairlets}(P,d,k)$ and \linebreak $\textsc{compute-centers}(\Representatives,d,k)$, respectively. Then, the running time of $\Cref{alg:general-case-outlier-algorithm}$ is in \linebreak $O(T_{\F}+T_{C}+|\F|\cdot k + n)$.
\end{restatable}
\iffullappendix
\begin{proof}
    The algorithms $\textsc{compute-fairlets}$ and $\textsc{compute-centers}$ are both called once in the beginning, incurring a running time of $O(T_{\F}+T_{C})$.
    For a fixed fairlet $f\in \F$, we iterate over $C$ in $O(k)$ to find the center $c_f$ that is closest to the anchor of $f$. This gives a total running time of $O(|\F|k)$ for finding the closest centers.
    Each point $p\in P(\F)$ is assigned once, which leads to an overall running time of $O(n)$ for the point assignments.
\end{proof}
\fi
Note that we will show that the running time of the overall algorithm amounts to $O(n^{2+o(1)}\log n)$ in all one-sided outlier cases.

\section{Analysis of the \texorpdfstring{$1:t$}{1-t}-fairlet decomposition}\label{appendix:analysis-1-t}

\iffullappendix
\thmonetunbalanced*
\begin{proof}
    Let $(C.\alpha)$ be the clustering computed by \Cref{alg:general-case-outlier-algorithm} using the fairlet decomposition $\F$.
    By \Cref{lem:4-approximation}, the cost of $(C,\alpha)$ is upper bounded by $\varphi + \rho$, where $\varphi = \max_{f\in \F}\max_{b\in f\cap B}d(b,\anc(f))$ and $\rho= \max_{r\in R}\min_{c\in C}d(r,c)$, as the set of anchors equals $R$.
    Let $\OPTfair$ denote the cost of an optimal $1:t$-fair clustering with outliers.
    \Cref{lem:properties-computed-fairlets-1:t} implies $\varphi \le 2\OPTfair$
    and \Cref{lem:compute-centers-2OPT-Gonzalez} with $S=\OPTfair$ implies $\rho \le 2\OPTfair$.
\end{proof}
\fi

\begin{restatable}{lemma}{ComputeFairletDecompOnet}
    \label{lem:properties-computed-fairlets-1:t}
    Assume that there exists a $1:t$-fairlet decomposition $\F$ such that $d(r,b) \leq \tau\delta$ for every $f\in \F$, the anchor $r\in f\cap R$ of $f$ and every blue point $b\in f\cap B$ for some $\delta\in D$.
	Then, \Cref{alg:compute-1-t-fairlets} computes a $1:t$-fairlet decomposition with this property.
\end{restatable}
\iffullappendix
\begin{proof}
	Any flow on $N_{\delta}$ can have value at most $t|R|$, as this value equals the sum of capacities of the edges incident to $w$.
	We construct a flow $\mathbf{f}^*$ with value \(t|R|\) as follows:
    For all $f = \{r,b_1,\ldots,b_t\}\in \F$, and all $i\leq t$, set $\mathbf{f}^*(v,b_i) = 1$, $\mathbf{f}^*(b_i,r) = 1$, and $\mathbf{f}^*(r,w)=t$. Set all remaining flow values to 0.
    By assumption, $d(b_i,r)\leq \tau\delta$ for all $i\leq t$, hence $N_{\delta}$ contains the edge $(b_i,r)$. 
    Consider $b\in B\cap P(\F)$. As $b$ is contained in exactly one fairlet of $\F$, there is only one outgoing edge at $b$ in $N_{\delta}$ with non-zero flow value. 
    Further, by construction, there is only one incoming edge at $b$. 
    Both these edges have flow value 1 in $\mathbf{f}^*$. 
    Similarly, we can argue that there are exactly $t$ non-zero incoming and one outgoing edge for every $r\in R$. 
    The incoming edges all have value 1, and the outgoing edge has value $t$ in the flow $\mathbf{f}^*$. 
    Hence, $\mathbf{f}^*$ is a valid flow of value $t|R|$.  
    As such a flow exists, we can find an integral flow $\mathbf{f}$ of value $t|R|$ in \Cref{alg-line:compute-max-flow-1:t}.
	
	Line~\ref{alg-line:add-r-to-fairlet-1:t} is the only place where we add red points to fairlets. This implies that every $f\in \F$ contains exactly 1 red point. 
    Let $r\in R$. For every edge $(b,r)\in E(\tau\delta)$ with $\mathbf{f}(b,r)=1$, we add $b$ to the fairlet containing $r$. As $|\mathbf{f}|=t|R|$, it must be $\mathbf{f}(r,w)=t$. 
    Hence, there must be exactly $t$ incoming edges with flow value 1. 
    This implies that every fairlet contains exactly $t$ blue points. 
	Every red point is added to a fairlet only once by \Cref{alg-line:add-r-to-fairlet-1:t}. As $\mathbf{f}(v,b)\in \{0,1\}$ for all $b\in B$, there exists at most one $r\in R$ such that $\mathbf{f}(b,r)=1$. Hence, every blue point is added to at most 1 fairlet. This implies that the fairlets are pairwise disjoint.
	$\F$ covers all red and $t|R|$ blue points. Hence, $|B|-t|R|$ blue points are not covered.
	The bound on the distances follows from the network construction.
\end{proof}
\fi

\begin{restatable}{lemma}{FindAgreeingFairletDecompOnet}
    \label{lem:find-agreeing-fairlet-decomposition-1:t}
	Let $\C$ be a $1:t$-fair $k$-center clustering with outliers. Then there exists a $1:t$-fairlet decomposition $\F$ that agrees with $\C$.
\end{restatable}
\iffullappendix
\begin{proof}
	Consider $C\in \C$.	As $\C$ is $1:t$-fair, $|C\cap B|=t|C\cap R|$. Therefore, we can arbitrarily partition $C$ into $|C\cap R|$ sets, each consisting of one red and $t$ blue points. Call the set of such sets $\F_C$. The union $\F\coloneqq\bigcup_{C\in \C}\F_C$ is a $1:t$-fairlet decomposition that agrees with $\C$. 
\end{proof}
\fi

\begin{restatable}{corollary}{CostFairletDecompOnet}
    \label{cor:1-t-fairlet-decomposition-2OPT}
    Let $\OPTfair$ be the radius of an optimal $1:t$-fair $k$-center with outliers solution.
    \Cref{alg:compute-1-t-fairlets} computes a $1:t$-fairlet decomposition with $d(r,b) \leq 2\OPTfair$ for every $f\in \F$, the anchor $r\in f\cap R$ of $f$ and every blue point $b\in f\cap B$.
\end{restatable}
\iffullappendix
\begin{proof}
    By \Cref{lem:find-agreeing-fairlet-decomposition-1:t}, there exists a fairlet decomposition with diameter $2\OPTfair$. As $\OPTfair\in D$, \Cref{lem:properties-computed-fairlets-1:t} with $\delta = \OPTfair$ and $\tau=2$ implies that \Cref{alg:compute-1-t-fairlets} finds such a fairlet decomposition.
\end{proof}
\fi

\begin{restatable}{lemma}{RunningTimeComputeFairletsOnet}
    \label{lem:running-time-compute-1:t-fairlets}
    The overall running time of \Cref{alg:compute-1-t-fairlets} is in $O(n^2\log n + \frac{n^{2+o(1)}}{t}\log n)$ if it uses binary search to find the smallest $\delta\in D$ such that a maximum flow of value $t|R|$ is found in the network $N_{\delta}$. 
\end{restatable}
\iffullappendix
\begin{proof}
    The set of pairwise distances $D$ fulfills $|D|\in O(n^2)$, hence sorting this list can be done in $O(n^2\log n)$.
    By assumption, $|B|\ge t\cdot|R|$, therefore $|R|\le \frac{n}{t}$.
    For a given $\delta$, the network $N_{\delta}$ can be constructed in $O(|R|\cdot |B|)\subseteq O(\frac{n^2}{t})$. 
    The subsequent maximum flow computation can be done in $O(|E(\tau\delta)|^{1+o(1)}) \subseteq O(\frac{n^{2+o(1)}}{t})$ using the algorithm by~\cite{brand2023maxflowalmostlinear}.
    The construction of the fairlets, performed in Lines~\ref{alg-line:start-constructing-fairlets-1:t} to \ref{alg-line:end-constructing-fairlets-1:t}, is only done once if the condition $|\mathbf{f}|=t|R|$ is fulfilled, and incurs a total running time of $O(|R|\cdot |B|)\subseteq O(\frac{n^2}{t})$, which is dominated by the running time for constructing a network and computing the maximum flow. Finding the smallest $\delta\in D$ such that $N_{\delta}$ carries flow of value $t|R|$ can be done in $O(\frac{n^{2+o(1)}}{t}\log n)$ when performing a binary search over $D$.
\end{proof}
\fi

\section{Analysis of the \texorpdfstring{$1:t_2:\ldots:t_m$}{1:t2:...:tm}-fair decomposition}\label{appendix:analysis-1-many-t}

\begin{lemma}\label{lem:properties-computed-fairlets-1:...:tm}
    Let $D=\{d(p,q) \mid p,q\in P\}$ and let $\delta\in D$ such that for all $2\le i \le m$, there exists a $1:t_i$-fairlet decomposition $\F_i$ of $H_1\cup H_i$ and $d(h_1,h_i) \le \tau\delta$ for $h_1\in f\cap H_1$ and $h_i \in f\cap H_i$ for all $f\in \F$. 
    Then, \Cref{alg:compute-1:t2:...:tm-fairlets} computes a $1:t_2:\ldots:t_m$-fairlet decomposition $\F$ of $H_1\cup\ldots\cup H_m$ excluding $|H_i|-t_i|H_1|$ outliers from $H_i$ for all $2\le i \le m$ such that $\max_{p\in f}d(p,h_1) \le \tau\delta$, where $h_1\in f\cap H_1$. 
\end{lemma}
\iffullappendix
\begin{proof}
    The algorithm computes $\F = \{\bigcup_{2\le i \le m} f_i^{h_1} \mid h_1\in H_1\}$, where $f_i^{h_1}$ is a fairlet in $\F_i$ such that $h_1\in F_i$ for a given $h_1\in H_1$. As $\F_i$ is a valid fairlet decomposition, there is exactly one such $f_i^{h_1}$ per $i$ and $h_1$. Hence, for every $h_1\in H_1$, there exists exactly one $f\in \F$ such that $h_1\in f$, i.\,e., $H_1$ is completely covered by $\F$. 
    By construction, every $f\in \F$ contains exactly one $h_1\in H_1$ and $t_i$ points from $H_i$ for all $2\le i \le m$.

    Let $f\in \F$. For all $2\le i\le m$ and $f_i\in \F_i$ such that $f_i\subseteq f$ it is $f\cap H_i = f_i \cap H_i$, which implies that $\F$ excludes exactly $|H_i|-t_i|H_1|$ outliers from $H_i$ for all $2\le i \le m$.
    The fairlets in $\F$ need to be pairwise disjoint, as otherwise, one of the $\F_i$ would violate the disjointness property.
    For the radius property, consider $f\in \F$ and its unique point $h_1\in f\cap H_1$ from $H_1$. Then, $\max_{p\in f}d(p,h_1) = \max_{2\le i \le m}\max_{p\in f_i^{h_1}}d(p,h_1) \leq \tau\delta$, where $f_i^{h_1}$ is defined as above. 
\end{proof}
\fi
In the final algorithm, we only need to change the way the fairlets are computed; everything else is done as in the $1:t$-case.

\begin{corollary} \label{cor:1-t2-...-tm-fairlet-decomposition-2OPT}
    Let $\OPTfair$ be the radius of an optimal $1:t_2:\ldots:t_m$-fair $k$-center with outliers solution.
    \Cref{alg:compute-1:t2:...:tm-fairlets} computes a $1:t_2:\ldots:t_m$-fairlet decomposition with $d(h_1,p) \leq 2\OPTfair$ for every $f\in \F$, the anchor $h_1\in f\cap H_1$ of $f$ and every blue point $p\in f\cap (H_2\cup\ldots\cup H_m)$.
\end{corollary}
\iffullappendix
\begin{proof}
    Let $2\le i \le m$.
    As $|H_i\cap C| = t_i|H_1\cap C|$ for all $C\in \C^*$, $\C^*_{i} \coloneqq \{C\cap (H_1\cup H_i) \mid C\in \C^*\}$ is a $1:t_i$-fair $k$-center clustering excluding $|H_i|-t_i|H_1|$ outliers from $H_i$. By \Cref{lem:find-agreeing-fairlet-decomposition-1:t}, there exists a $1:t_i$-fairlet decomposition $\F_i$ that agrees with $\C^*$.
    This fairlet decomposition $\F_i$ has diameter at most $2\OPT$.
    As $\OPT\in D$, \Cref{lem:properties-computed-fairlets-1:...:tm} with $\delta = \OPT$ and $\tau = 2$ implies that \Cref{alg:compute-1:t2:...:tm-fairlets} yields a $1:t_2:\ldots:t_m$-fairlet decomposition $\F$ with $\max_{p\in f}d(p,h_1) \le 2\OPT$ for all $f\in \F$ and $h_1\in f\cap H_1$.
\end{proof}
\fi

\begin{lemma} \label{lem:running-time-1:t2:...:tm-fairlets}
    \Cref{alg:compute-1:t2:...:tm-fairlets} (\textsc{compute-$1:t_2:\ldots:t_m$-fairlets}) runs in $O(n^{2+o(1)}\log n)$.
\end{lemma}
\iffullappendix
\begin{proof}
    For every $i=2,\ldots,m$, we compute a $1:t_i$-fairlet decomposition $\F_i$ of $H_1\cup H_i$ using the $\textsc{compute-$1:t_i$-fairlets}$ algorithm. By \Cref{lem:running-time-compute-1:t-fairlets}, one fairlet decomposition can be computed in $O(\frac{n_i^{2+o(1)}}{t_i}\log n_i)$, where $n_i = |H_1|+|H_i|$. It is $\sum_{i=2}^m n_i = (m-1)|H_1|+n-|H_1| = (m-2)|H_1| + n \le 2n$ because of $|H_1|\le \frac{n}{m}$. For all fairlet decompositions, it follows a running time of $O(\sum_{i=2}^m \frac{n_i^{2+o(1)}}{t_i}\log(n_i)) \subseteq O(\log(n) \sum_{i=2}^m n^{2+o(1)}) \subseteq O(\log n (\sum_{i=2}^m n_i)^{2+o(1)}) \subseteq O(n^{2+o(1)}\log n)$.
    Then, we construct the $1:t_2:\ldots:t_m$-fairlet decomposition by iterating over $H_1$ and combining all $1:t_i$-fairlets that share a point in $H_1$. If, for every fairlet $f$, we store a pointer to its anchor, which corresponds to the point in $f\cap H_1$, then we get a running time of
    $O(|H_1|\sum_{i=2}^m |F_i|) = O(|H_1|\sum_{i=2}^m |H_1|) \subseteq O(\frac{n^2}{m})$ for the construction of the fairlets, which is dominated by the running time of the first loop.
\end{proof}
\fi

\begin{theorem} \label{thm:1-t1-...-tm-unbalanced} 
	\Cref{alg:general-case-outlier-algorithm} with \textsc{compute-fairlets} = \textsc{compute-$1:t_2:\ldots:t_m$-fairlets} (\Cref{alg:compute-1:t2:...:tm-fairlets}) computes a 4-approximation for $1:t_2:\ldots:t_m$-fair $k$-center with $|H_i| - t_i|H_1|$ outliers in $H_i$ in time $O(n^{2+o(1)}\log n)$. 
\end{theorem}
\iffullappendix
\begin{proof}
    Let $(C,\alpha)$ be the clustering computed by \Cref{alg:general-case-outlier-algorithm} using the fairlet decomposition $\F$.
    By \Cref{lem:4-approximation}, the cost of $(C,\alpha)$ is upper bounded by $\varphi + \rho$, where $\varphi = \max_{f\in \F}\max_{p\in f}d(p,\anc(f))$ and $\rho= \max_{a\in H_1}\min_{c\in C}d(a,c)$, as the set of anchors equals $H_1$.
    Let $\OPTfair$ denote the cost of an optimal $1:t_2:\ldots:t_m$-fair clustering with outliers.
    \Cref{lem:properties-computed-fairlets-1:...:tm} implies $\varphi \le 2\OPTfair$
    and \Cref{lem:compute-centers-2OPT-Gonzalez} with $S=\OPTfair$ implies $\rho \le 2\OPTfair$.
    By \Cref{lem:running-time-framework-outlier}, \Cref{alg:general-case-outlier-algorithm} has a running time of $O(T_{\F}+T_C + |\F|\cdot k + n)$, where $T_{\F}$ is the running time of \textsc{compute-fairlets} and $T_C$ is the running time of \textsc{compute-centers}. Here, $T_{\F} \in O(n^{2+o(1)}\log n)$ by \Cref{lem:running-time-1:t2:...:tm-fairlets} and $T_C\in O(|\Representatives|\cdot k) = O(|H_1|\cdot k) \subseteq O(n^2)$ \cite{Gonz85}. The size of the fairlet decomposition $\F$ can be bounded by $|H_1|\leq n$.
\end{proof}
\fi

\section{Algorithm and analysis for \texorpdfstring{$t_1:t_2:\ldots:t_m$}{t1:t2:...:tm}-fair decompositions}\label{appendix:generalcase}

Consider the most general case that there are $m$ colors $H_1,H_2,\ldots H_m$ in the ratio $t_1:t_2:\ldots:t_m$ for positive integers $t_1,t_2,\ldots,t_m$. Assume that $|H_1|>0$ is a multiple of $t_1$. We want to find $C_1,...,C_k$ such that $t_i\cdot|H_j\cap C_{\ell}| = t_j\cdot |H_i\cap C_{\ell}|$ for all $\ell\le k$ and all $1\le i,j\le m$. Assume that $|H_i| \geq \frac{t_i}{t_1}|H_1|$ for all $2\le i \le m$. A fair solution excludes $|H_i|-\frac{t_i}{t_1}|H_1|$ outliers from $H_i$. 

\begin{lemma}\label{lem:find-agreeing-fairlet-decomposition-s1:s2:...:sm}
	Let $\C$ be a $t_1:t_2:\ldots:t_m$-fair $k$-center clustering with outliers. Then there exists a $t_1:t_2:\ldots:t_m$-fairlet decomposition $\F$ that agrees with $\C$.
\end{lemma}
\iffullappendix
\begin{proof}
	Consider $C\in \C$.	As $\C$ is $t_1:t_2:\ldots:t_m$-fair, $t_1|C\cap H_i|=t_i|C\cap H_1|$. Therefore, we can arbitrarily partition $C$ into $\frac{|C\cap H_1|}{t_1}$ sets, each consisting of $t_i$ points from $H_i$ for all $1\le i \le m$. Call the set of such sets $\F_C$. The union $\F\coloneqq\bigcup_{C\in \C}\F_C$ is a $t_1:t_2:\ldots:t_m$-fairlet decomposition that agrees with $\C$. 
\end{proof}
\fi
To compute $t_1:t_2:\ldots:t_m$-fair fairlets, we first find the anchors of the fairlets by solving an instance of the \emph{capacitated multi-$\ell$-center problem} (also called capacitated $k$-center problem with soft capacities). 

\begin{definition}[The capacitated multi-$\ell$-center problem]
     Given a set of points $P$ and numbers $\ell\in \mathbb{N}$ and $u\in \mathbb{N}$. The goal is to find a multiset of at most $\ell$ centers $C^{\ell}$ and an assignment $\phi\colon P\to C^{\ell}$ from points to centers such that the $\ell$-center objective is minimized, while $|\phi^{-1}(c_i)|\leq u$ for all $i\leq \ell$.
\end{definition}

\begin{lemma}[\cite{khuller2000capacitated}]\label{lem:cap-k-center-5-approx-khuller}
    Let $\mathcal{I} = (P,\ell,u)$ be an instance of the capacitated multi-$\ell$-center problem.
    There is an algorithm that computes a 5-approximation for $\mathcal{I}$ in time $O(n^{2+o(1)}\log n)$. 
\end{lemma}
\Cref{lem:cap-k-center-5-approx-khuller} gives an upper bound on the radii of the capacitated multi-$\ell$-center clustering in terms of the cost of an optimal capacitated multi-$\ell$-center solution. We need the following result to get a bound in terms of $\OPTfair$.

\begin{corollary}\label{cor:cost-of-opt-cap-sol-optfair}
    Let $(A^*,\alpha^*)$ be an optimal solution to the capacitated multi-$\ell$-center problem on point set $H_1$ with capacity $t_1$. Then, $\max_{p\in H_1}d(p,\alpha^*(p)) \leq 2\OPTfair$. 
\end{corollary}
\iffullappendix
\begin{proof}
    We call $\max_{p\in H_1}d(p,\alpha(p))$ the cost of solution $(A,\alpha)$. 
    The capacitated multi-$\ell$-center problem is a relaxation of the version in which each point can serve as a center at most once. 
    The cost of an optimal solution to this stricter version is upper bounded by the cost of a fairlet decomposition $\max_{p\in f}d(p,\anc(f))$ that agrees with $\C^*$ as every cluster in $\C^*$ contains a multiple of $t_1$ many points from $H_1$. Such a fairlet decomposition exists by \Cref{lem:find-agreeing-fairlet-decomposition-s1:s2:...:sm}.
    The cost of a fairlet decomposition that agrees with $\C^*$ is upper bounded by $2\OPTfair$.
\end{proof}
\fi
By assumption on $|H_1|$, a $t_1:t_2:\ldots:t_m$-fairlet decomposition covers $H_1$ completely. The idea is to compute a capacitated $\ell$-center clustering with $\ell = \frac{|H_1|}{t_1}$ and capacity $t_1$. 
Let $(A,\ell)$ be such a capacitated $\ell$-center solution. The set $A$ might contain multiple copies of points from $H_1$, i.\,e., $A$ is a multiset of points from $H_1$. 
The points in $A$ will serve as the anchors for the fairlets that we will construct as follows.
We can reduce the construction of a $t_1:\ldots:t_m$-fairlet decomposition of $P$ to the construction of a $1:t_1:t_2:\ldots:t_m$-fairlet decomposition on $P'= A \cupdot H_1 \cupdot H_2 \cupdot \ldots \cupdot H_m$ where the anchors $A$ are added to form a new color set. Notice that every $p'\in P'$ is (the copy of) a point $p\in P$.
Let $d'\colon P'\times P'\to \mathbb{R} $ be the distance function induced by $d\colon P\times P\to \mathbb{R}$, where the distance between two points equals the distance between their original points.
Note that $d'$ is a pseudometric. 
However, this does not cause any problems, as all results so far hold in this setting as well. 
\Cref{alg:compute-s1-s2-...-sm-fairlets} formalizes the fairlet computation.

\begin{algorithm}
    \LinesNumbered
    \caption{\textsc{compute-$t_1:t_2:\ldots:t_m$-fairlets}}
    \label{alg:compute-s1-s2-...-sm-fairlets}
    \SetKwInOut{Input}{Input}
    \SetKwInOut{Output}{Output}
    \BlankLine
    \Input{Point set $P=H_1\cupdot H_2 \cupdot \ldots \cupdot H_m$, distance metric $d$, number of clusters $k$}
    \Output{Set of fairlets $\F$, anchor assignment $\anc\colon \F\to A$}
    \BlankLine
    $(A,\kappa) \gets$ 5-approx. sol. for capacitated multi-$\frac{|H_1|}{t_1}$-center on $H_1$ with capacity $t_1$ \label{alg-line:5-approx-cap}\\
$\F',\anc' \gets \hyperlink{alg:compute-1:t2:...:tm-fairlets}{\textsc{compute-$1:t_1:t_2:\ldots:t_m$-fairlets}}(P',A,H_1,H_2,\ldots,H_m,d',k,12)$\\
    $\F \gets \emptyset$\\
    $\forall f'\in \F'\colon f\gets f'\setminus A,\ \anc(f) \gets \anc'(f'),\ \F\gets \F\cup \{f\}$\\
    \Return{$\F,\anc$}
\end{algorithm}

\begin{lemma} \label{lem:properties-computed-fairlets-s1:s2:...:sm} 
    \Cref{alg:compute-s1-s2-...-sm-fairlets} computes a fairlet decomposition $\F$ with $|H_i|-\frac{t_i}{t_1}|H_1|$ outliers from $H_i$ for all $2\le i \le m$ such that for all $f \in \F$, $\max_{p\in f}d(\anc(f),p) \leq 12\OPT_{\text{fair}}$, where $\anc(f)$ is the anchor of fairlet $f$.
\end{lemma}
\iffullappendix
\begin{proof}
    By \Cref{lem:find-agreeing-fairlet-decomposition-s1:s2:...:sm}, there exists a fairlet decomposition $\F^*$ that agrees with an optimal $t_1:t_2:\ldots:t_m$-fair $k$-center clustering. Fix such an $\F^*$. Let $\C^{\ell} = \{\kappa^{-1}(a)\mid a\in A\}$ be the clustering induced by $(A,\kappa)$ computed in \Cref{alg-line:5-approx-cap} of \Cref{alg:compute-s1-s2-...-sm-fairlets}.
    By \Cref{lem:cap-k-center-5-approx-khuller}, $\C^{\ell}$ is a 5-approximation to an optimal capacitated multi-$\ell$-center clustering on $H_1$, i.\,e., $\max_{p\in H_1}d(p,\kappa(p)) \leq 5\OPT_{\ell}$, where $\OPT_{\ell}$ is the radius of an optimal multi-$\ell$-center clustering with capacity $t_1$.
    By \Cref{cor:cost-of-opt-cap-sol-optfair}, $\OPT_{\ell} \le 2\OPTfair$.
    The clustering $ \{\{a\} \cup \kappa^{-1}(a) \mid a\in A\}$  is a $1:t_1$-fairlet decomposition of $A \cup H_1$ with $\max_{p\in H_1}d(p,\kappa(p)) \leq 10\OPTfair$.

    It remains to construct $1:t_i$-fairlets of $A\cup H_i$ with bounded cost for all $2\le i \le m$.
    To do this, we show that we can find a perfect matching between the sets in $\F^*$ and $\C^{\ell}$ such that the matched pairs of sets share a point in their intersection. From this, we can construct a $1:t_i$-fairlet decomposition consisting of the anchor of the capacitated cluster and the points of color $H_i$ inside the fairlet.
    Let $2\le i \le m$. To find a $1:t_i$-fairlet decomposition of $A \cup H_i$, construct a bipartite graph $G=(\F^*\cup \C^{\ell},E)$, where $E=\{(f^*,C)| f^*\in \F^*, C\in \C^{\ell} \colon f^* \cap C \neq \emptyset\}$.
    To show that there exists a perfect matching, we use Hall's Marriage Theorem \cite{hall1935marriage}.
    Let $W\subseteq \F^*$. The total number of points from $H_1$ contained in all fairlets in $W$ is $t_1|W|$. 
    By the choice $\ell = \frac{|H_1|}{t_1}$ and the capacity $t_1$, each cluster $C\in N(W) \coloneqq \{C'\in \C^{\ell} \mid \exists f\in W\colon \{f,C'\}\in E\}$ contains exactly $t_1$ points from $H_1$. 
    To cover all points of $H_1 \cap W$, there must be at least $\frac{|W|t_1}{t_1}=|W|$ clusters in $N(W)$. 
    That is, $|N(W)|\geq |W|$. 
    Applying Hall's Marriage Theorem \cite{hall1935marriage} implies that there exists a perfect matching $M$ in $G$. 

    For each edge $(f^*,C)\in M$, pick a $p\in f^*\cap C$ and set it as the anchor $\anc(f^*)$ of the fairlet $f^*$. 
    Then, each cluster $C\in \C^{\ell}$ contains exactly one unique anchor. Consider the set $\F_{H_i} \coloneqq \{\{a\} \cup \{h^i\in H_i\mid h^i\in f^*\in \F^*, \kappa(\anc(f^*)) = a\} \mid a \in A\}$. 
    The matching guarantees that every point in $A$ is used exactly once.
    Hence, $\F_{H_i}$ is a $1:t_i$-fairlet decomposition of $A\cup H_i$.
    Let $f^*\in \F^*$ and $h^i\in f^*\cap H_i$. It is $d(h^i,\anc(f^*)) \leq 2\OPTfair$, and $d(\anc(f^*),\kappa(\anc(f^*)))\leq 10\OPTfair$. Hence, $d(h^i, \kappa(\anc(f^*))) \leq 12\OPTfair$ for all $h^i\in f\in \F_{H_i}$.
    Applying \Cref{lem:properties-computed-fairlets-1:...:tm} with fairlet decompositions $\C^{\ell}$ and $\F_{H_i}$, $2\le i \le m$, $\tau=12$ and $\delta = \OPTfair$ yields the claim.
\end{proof}
\fi

\begin{lemma} \label{lem:running-time-t1:...:tm-fairlets}
    \Cref{alg:compute-s1-s2-...-sm-fairlets} (\textsc{compute-$t_1:\ldots:t_m$-fairlets}) runs in time $O(n^{2+o(1)}\log n)$.
\end{lemma}
\iffullappendix
\begin{proof}
    By \Cref{lem:cap-k-center-5-approx-khuller}, the 5-approximate solution for capacitated multi-$\ell$-center can be computed in $O(|H_1|^{2+o(1)}\log |H_1|) \subseteq O(n^{2+o(1)}\log n)$.
    By \Cref{lem:running-time-1:t2:...:tm-fairlets}, the set $\F'$ can be computed in $O(|P'|^{2+o(1)}\log |P'|)$, where $|P'|=|A|\cupdot|H_1|\cupdot\ldots\cupdot|H_m| \leq \frac{|H_1|}{t_1}+n \in O(n) $. The remaining steps are dominated by this running time.
\end{proof}
\fi

In the complete approach, \Cref{alg:general-case-outlier-algorithm} uses \Cref{alg:compute-centers-Gonzalez} on the multiset $A$ and sets the final set of centers to be the underlying points from which the copies in $A$ originated. As we only identify centers that have zero distance, the bound on the cost does not change.

\begin{theorem}[Reformulation of \Cref{thm:t1:...:tm-beginning}] \label{thm:s1:...:sm-14-approx}
    For $t_1,\ldots,t_m\in \mathbb{N}$ and a fair clustering instance satisfying $t_i|H_1|\le t_1|H_i|$ for all $2\le i \le m$ and $\frac{|H_1|}{t_1}\in\mathbb{Z}$, 
    \Cref{alg:general-case-outlier-algorithm} with \textsc{compute-fairlets} = \textsc{compute-$t_1:\ldots:t_m$-fairlets} (\Cref{alg:compute-s1-s2-...-sm-fairlets}) computes a center-aware 14-approximation for $t_1:t_2:\ldots:t_m$-fair $k$-center with $|H_i|-\frac{t_i}{t_1}|H_1|$ outliers from $H_i$ for all $2\le i \le m$ in time $O(n^{2+o(1)}\log n)$.
\end{theorem}
\iffullappendix
\begin{proof}
    Let $\varphi \coloneqq \max_{f\in \F}\max_{p\in f}d(p,\anc(f))$ and $\rho \coloneqq \max_{a\in A}\min_{c\in C} d(a,c) $.
    By \Cref{lem:4-approximation}, the cost of the final clustering is bounded by $\varphi + \rho$. 
    Let $\OPTfair$ denote the cost of an optimal $t_1:\ldots:t_m$-fair $k$-center clustering with outliers.
    The set of centers $C$ is computed via farthest-first traversal on $A$. By \Cref{lem:compute-centers-2OPT-Gonzalez}, $\rho \le 2\OPTfair$.
    By \Cref{lem:properties-computed-fairlets-s1:s2:...:sm}, $\varphi \le 12\OPTfair.$

    For the running time, \Cref{lem:running-time-framework-outlier} implies that \Cref{alg:general-case-outlier-algorithm} runs in time $O(T_{\F}+T_{C}+|\F|\cdot k + n)$, where $T_{\F}$ is the running time of \textsc{compute-fairlets} and $T_C$ is the running time of \textsc{compute-centers}. Here, $T_{\F} \in O(n^{2+o(1)}\log n)$ by \Cref{lem:running-time-t1:...:tm-fairlets} and $T_C\in O(|\Representatives|\cdot k)\subseteq O(nk)$ \cite{Gonz85}.
    The fairlet decomposition $\F$ consists of $\frac{|H_1|}{t_1} \in O(n)$ fairlets.
\end{proof}
\fi

\subsection{\texorpdfstring{$1:1$}{1:1} fairness with outliers on both sides}\label{sec:bothsides}
So far, we used outliers to establish the desired fairness ratio in clusters. 
However, since $k$-center is prone to outliers, we may also want to combine fairness and outliers to satisfy fairness and simultaneously exclude points that are far away. In previous sections, the number of outliers from group $H_i$ is always a fixed number computed from the desired color ratio and the numbers $|H_1|$ and $|H_i|$. We now want to allow more outliers so that it is possible to discard points that are far away. We consider the case of two colors, i.\,e., $P=R\cupdot B$ with $|R|\leq |B|$ and a given number $z\geq |B|-|R|$ of outliers. 
If a solution has $z_r$ outliers in $R$, it has exactly $z_b=z_r+|B|-|R|$ outliers in $B$ since otherwise the above fairness condition is violated.

We now develop an algorithm for $1:1$-fair $k$-center with outliers. Similarly to the previous algorithms, we start by constructing a bipartite graph $G=(V, E)$ with $V=R\cup B$ and $E=\{\{r,b\}\mid r\in R, b\in B \text{ and }d(r,b) \leq \delta\}$, i.\,e., edges between any pair of a red and a blue point are added if their distance is upper bounded by $\delta$. Let $\F_{\delta}$ denote a maximum matching in this graph. 
If $\F_\delta$ has cardinality at least $\frac{|P|-z}{2}$, we store $F_\delta$, and otherwise, we discard it. Let $\delta_{\min}$ denote the smallest threshold such that $\F_{\delta_{\min}}$ has cardinality at least $\frac{|P|-z}{2}$.
We can determine $\delta_{\min}$ via a binary search and then compute a maximum matching for all $\delta\geq \delta_{\min}$. 
Let $\mathscr F=\{\F_{\delta}\mid \delta\geq \delta_{\min}\}$ be the set of all such matchings. 

\begin{algorithm}
    \LinesNumbered
    \caption{\textsc{set-of-possible-\(1:1\)-fairlets}}
    \label{alg:outliers_both_sides_fairlets}
    \SetKwInOut{Input}{Input}
    \SetKwInOut{Output}{Output}
    \BlankLine
    \Input{Point set $P=R\cupdot B$, distance metric $d$, integer $k$, number of outliers $z$}
    \Output{Set of possible fairlet decompositions $\mathscr F$}
    \BlankLine
    $D\gets \{d(x,y)\mid x,y\in P\}$\\
    $\mathscr F\gets \emptyset$\\
    \For{all $\delta \in D$}{
        construct bipartite graph $G=(R\cup B,\ \{\{r,b\} \mid r\in R, b\in B, d(r,b) \leq \delta\})$ \label{alg-line:graph-construction}\\
        $\F_{\delta}\gets$ maximum matching in $G$\\
        \If{$|\F_{\delta}|\geq \frac{|P|-z}{2}$}{
            $\mathscr F\gets\mathscr F\cup \{\mathcal F_{\delta}\}$
        }
    }
    \Return{$\mathscr F$}
\end{algorithm}

\begin{lemma}
\label{lem:delta_min_bound}
     Let $\OPT$ denote the optimal cost of a $1:1$-fair $k$-center solution with outliers. Then
     we have $2\OPT\geq \delta_{\min}$.
\end{lemma}
\iffullappendix
\begin{proof}
    Let $(C^*,\alpha^*)$ denote an optimal solution and let $C^*_1,\ldots, C^*_k$ be the resulting clusters and $Z^*$ the outliers. Since all clusters are $1:1$-fair, we can find a bijection between red and blue points belonging to the same cluster. Let $g$ be such a bijection and consider the matching  $M=\{\{r,g(r)\}\mid r\in R\backslash Z^*\}$ in the graph $G_{2\OPT}=(R\cup B,\ \{\{r,b\} \mid r\in R, b\in B, d(r,b) \leq 2\OPT\})$. This matching has cardinality $\frac{|P|-|Z^*|}{2}$. Furthermore, notice that for $\delta=\max\{d(x,y)\mid x,y\in P, d(x,y)\leq 2\OPT\}$ the graph with respect to $\delta$ in \Cref{alg-line:graph-construction} of \Cref{alg:outliers_both_sides_fairlets} equals $G_{2\OPT}$, so we have that $\mathcal F_{\delta}\in \mathscr F$ and thus $\delta_{\min}\leq\delta\leq2\OPT$.
\end{proof}
\fi
The proof also shows that there exists a fairlet decomposition $\mathcal F_{\delta}\in\mathscr F$ which is a maximum matching in the graph $G_{2\OPT}=(R\cup B,\ \{\{r,b\} \mid r\in R, b\in B, d(r,b) \leq 2\OPT\})$. We denote this matching by $\mathcal F_{2\OPT}$ for convenience.
For every $\delta$ we denote by $Z_{\delta}=P\backslash P(\mathcal F_{\delta})$ the points which are not part of a fairlet in $\F_{\delta}$. Notice that if $|\F_{\delta}|\geq \frac{|P|-z}{2}$ we have $|Z_{\delta}|\leq z$. Since points from both colors can be outliers, both points of a fairlet may be outliers in the optimal solution. Since the cost of an optimal solution is not known, it is not sufficient to compute the fairlet decomposition for $\delta_{\min}$ as we see in \Cref{fig:example-why-we-need-multiple-fairlet-decompositions}.

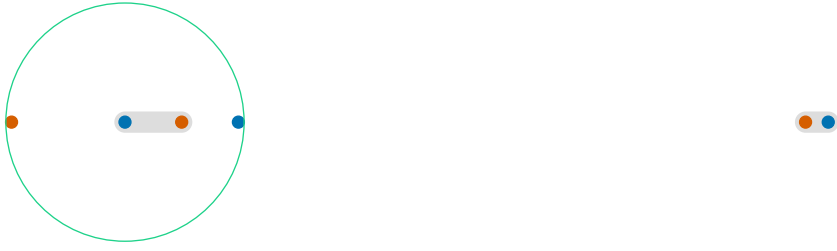
\begin{figure}[htbp]
    \centering
    \scalebox{0.5}{
  \begin{tikzpicture}  
    \def \pointwidth  {1.2pt}
    \def \linewidth {1.5}

    \node[point=red] (r1) at (0,0) {};
    \node[point=blue] (b1) at ($ (r1) + (3,0) $) {};
    \node[point=red] (r2) at ($ (b1) + (1.5,0) $) {};
    \node[point=blue] (b2) at ($ (r2) + (1.5,0) $) {};
    \node[point=red] (r3) at ($ (b2) + (15,0) $) {};
    \node[point=blue] (b3) at ($ (r3) + (0.6,0) $) {};

    \draw[draw=green] (b1) circle (3.15) {};
    \fairletarea{b1,r2};
    \fairletarea{r3,b3};
  \end{tikzpicture}
}
    \caption{Optimal clustering for $k=1$, $z=2$ shown in green. However, if we set $\delta$ too small, then we might only find the fairlets shown in gray. A clustering based on these fairlets would have to place both fairlets in the same cluster. As we could place the two points to the right arbitrarily far away, the cost of such a clustering could be arbitrarily high.}
    \label{fig:example-why-we-need-multiple-fairlet-decompositions}
\end{figure}

\begin{lemma} \label{lem:running-time-set-of-possibe-1:1-fairlets}
    \textsc{set-of-possible-$1:1$-fairlets} (\Cref{alg:outliers_both_sides_fairlets}) runs in time $O(n^{4+o(1)})$.
\end{lemma}
\begin{proof}
    The set of pairwise distances $D$ has cardinality $O(n^2)$. 
    For every $\delta\in D$, we construct a bipartite graph in $O(n^2)$ and compute a maximum matching in this graph in $O(n^{2+o(1)})$ by computing a maximum flow in the corresponding network using the algorithm by~\cite{brand2023maxflowalmostlinear}.
    In total, this yields a running time of $O(n^{4+o(1)})$ and dominates the overall running time.
\end{proof}
\begin{algorithm}
    \LinesNumbered
    \caption{\textsc{$1:1$-fair-$k$-center-with-outliers-on-both-sides}}
    \label{alg:outliers_both_sides_supplier}
    \SetKwInOut{Input}{Input}
    \SetKwInOut{Output}{Output}
    \BlankLine
    \Input{Set of fairlet decompositions $\mathscr F$}
    \Output{$1:1$-fair $k$-center solution with $\leq z$ outliers}
    \BlankLine
    \For{all $\mathcal F_{\delta}\in \mathscr F$}{
        $z_\delta\gets\frac{z-|Z_\delta|}{2}$, $R'\gets R\backslash Z_\delta$\\
        let $(C_R,\alpha_R)$ be a $k$-supplier with outliers solution on $(R', P, d, z_{\delta},k)$ \label{alg-line:k-sup-sol}\\
        $C_\delta\gets C_R$\\
        \For{all $\{r,b\}\in \mathcal F_{\delta}$}{
            $\alpha_\delta(r)=\alpha_R(r)$\\
            $\alpha_\delta(b)=\alpha_R(r)$\\
        }
        \For{all $p\in Z_\delta$}{
            $\alpha_\delta(p)=\perp$\\
        }
    }
    \Return{$\textup{argmin}_{\delta\geq\delta_{\min}}\cost(C_\delta,\alpha_\delta)$}
\end{algorithm}
Similarly to \Cref{def:fairlet-agree-with-clustering}, we define the notion of a fairlet decomposition and a clustering that agree with each other. Notice that this definition is slightly different, as fairlets can now also be outliers. 
\begin{definition}
	Let $\C$ be a clustering with outliers $Z$ and $\F_{\delta}\in \mathscr F$ a fairlet decomposition with $Z_{\delta}=P\backslash P(\mathcal F_{\delta})$. We say that $\F_{\delta}$ and $\C$ \emph{agree with each other} if and only if $Z_{\delta}\subset Z$ and for all $f\in\F_{\delta}$ either $f\subset Z$ or there exists a cluster $C\in\C$ such that $f\subseteq C$.
\end{definition}
Suppose that a clustering $\C$ with outliers $Z$ agrees with $\F_{\delta}\in\mathscr F$, then the number of fairlets from $\F_{\delta}$ which are contained in $Z$ is upper bounded by $z_{\delta}=\frac{z-|Z_{\delta}|}{2}$. Furthermore, $\C$ satisfies $1:1$ fairness since every cluster in $\C$ is a union of fairlets from $\mathcal F_{\delta}$.
We can find a clustering $\mathcal{C}$ that agrees with $\mathcal{F}_{\delta}$ by reduction to $k$-supplier with outliers. 

\begin{definition}
    An instance of $k$-supplier with outliers is given by $(X, F, d, z, k)$, where $(X\cup F,d)$ is a metric space, and $z,k\in \mathbb N$. The goal is to choose a set $C\subset F$ of cardinality at most $k$ and an assignment $\alpha\colon X\rightarrow C\cup\{\perp\}$ such that the set of outliers $Z=\alpha^{-1}(\perp)$ has cardinality at most $z$ and the cost $\max_{p\in X\backslash Z}d(p,\alpha(p))$ is minimized. 
\end{definition}
We compute a $k$-supplier with outliers solution on points $R\backslash Z_{\delta}$ and centers from $P$ with $z_{\delta}$ outliers. Given such a solution $(C_R,\alpha_R)$ we set $C_{\delta}=C_R$ and for every fairlet $\{r,b\}\in \mathcal F_{\delta}$ we set $\alpha_\delta(b)=\alpha_\delta(r)=\alpha_R(r)$ and $\alpha_\delta(p)=\perp$ for all $p\in Z_{\delta}$. Then at most $2z_{\delta}$ points from $P(\mathcal F_{\delta})$ are marked as outliers and additionally $|Z_{\delta}|$ points. In conclusion, the number of outliers does not exceed $z$, and the clustering agrees with $\mathcal F_\delta$. By definition of $\alpha_\delta$, the resulting clustering $(C_\delta,\alpha_\delta)$ is also $1:1$-fair.
To bound the cost of the final clustering, we bound the cost of a $k$-supplier solution.
\begin{lemma}
\label{lem:cost_bound_supplier}
    Let $\OPT$ denote the optimal cost of a $1:1$-fair $k$-center solution with outliers. There exists a $k$-supplier with outliers solution $(C,\alpha)$ on $(R\backslash Z_{2\OPT}, P, d, z_{2\OPT}, k)$ with cost at most $3\OPT$. 
\end{lemma}
\iffullappendix
\begin{proof}
    Let $(C^*,\alpha^*)$ be an optimal solution to $1:1$-fair $k$-center with outliers and let $Z^*$ denote the set of outliers. We show how to obtain a solution from $(C^*,\alpha^*)$ for $k$-supplier with outliers. Notice that if $(C^*,\alpha^*)$ agrees with the fairlet decomposition $\mathcal F_{2\OPT}$ then this is straightforward. We set $C=C^*$ and $\alpha(r)=\alpha^*(r)$ for all $r\in R\backslash Z_{2\OPT}$. We have to verify that at most $z_{2\OPT}$ points are declared as outliers by this solution. Since  $\mathcal F_{2\OPT}$ agrees with $(C^*,\alpha^*)$ we have  $Z_{2\OPT}\subset Z^*$ and furthermore every fairlet is either completely contained in $Z^*$ or disjoint to $Z^*$, thus the number of points from $R\backslash Z_{2\OPT}$ which are contained in $Z^*$ is upper bounded by $\frac{z-|Z_{2\OPT}|}{2}$ which is $z_{2\OPT}$ by definition.
    We conclude that it suffices to construct a solution to $1:1$-fair $k$-center with outliers which agrees with $\mathcal F_{2\OPT}$ and costs at most $3\OPT$.

    Consider the following graph $G=(V,E)$ with $V=P$ with $E=\{\{r,b\}\mid r\in R,b\in B \text{ and } d(r,b)\leq 2\OPT\}$. Notice that the graph $G$ is bipartite as the endpoints of edges have different colors. Thus, the graph does not contain cycles of odd length.
    Let $g\colon R\backslash Z^*\rightarrow B\backslash Z^*$ be a bijection such that $r$ and $g(r)$ belong to the same cluster in $(C^*,\alpha^*)$ for all $r\in R\backslash Z^*$. Let $E_1=\{\{r,g(r)\}\mid r \in R\backslash Z^*\}$ and let $E_2=\mathcal F_{2\OPT}$. Notice that we have $E_i\subset E$  and $E_i$ is a matching in $G$ for $i=1,2$. Indeed, $E_2$ is a maximum matching in $G$ by definition of $\mathcal F_{2\OPT}$. Thus, we have $|E_1|\leq |E_2|$. We construct a multi-graph $H=(V,E_1\cupdot E_2)$ where every edge in $E_1\cap E_2$ is added twice. Notice that every point in the graph has degree at most $2$ and thus $H$ can be partitioned into vertex-disjoint cycles and paths. Every path and cycle alternates between edges from $E_1$ and $E_2$. For every cycle, the number of edges from $E_1$ equals the number of edges from $E_2$, and every path contains at least as many edges from $E_2$ as from $E_1$ since $E_2$ is a maximum matching.
    We define a solution to $1:1$-fair $k$-center with outliers that agrees with $\mathcal F_{2\OPT}$ as follows. 
    Let $C'=C^*$ and for every $\{r,b\}\in \mathcal F_{2\OPT}$ we set 
    \[
        \alpha'(x) = 
        \begin{cases}
        \alpha^*(r) & \text{if } \alpha^*(r)\neq \perp \\
        \alpha^*(b) & \text{if } \alpha^*(r)= \perp \text{ and }\alpha^*(b)\neq \perp \\
        \perp & \text{else}
        \end{cases}
    \]
    for $x\in \{r,b\}$. Furthermore, we set $\alpha'(x)=\perp$ for all $x\in Z_{2\OPT}$. By definition $(C',\alpha')$ agrees with $\mathcal F_{2\OPT}$.
    We claim that $(C',\alpha')$ has at most $z$ outliers. For every path and cycle in the decomposition of $H$, we claim that the number of outliers from its vertices with respect to $(C^*,\alpha^*)$ is greater than or equal to the number of outliers with respect to $(C',\alpha')$. 
    
    Let $C$ by a cycle in the decomposition, since the degree of every point $p$ on the cycle is $2$ we know that $p$ is part of a fairlet in $\mathcal F_{2\OPT}$ and not an outlier in $(C^*,\alpha^*)$, by definition of $\alpha'$ it is not an outlier in $(C',\alpha')$. 
    
    Let $Q=(v_1,\ldots, v_t)$ be a path in the decomposition of length at least $2$. For all points $v_i$ with $1<i<t$ we can argue as before that $\alpha'(p)\neq\perp$, since $v_i$ has degree $2$.
    Remember that every path contains at least as many edges from $E_2$ as from $E_1$, so at least one edge incident to $v_1$ or $v_t$ is in $E_2$. We can assume w.l.o.g. that this edge is $\{v_1,v_2\}$. Since $v_1$ is part of a fairlet and $v_2$ has degree at least $2$ we obtain by definition of $\alpha'$ that $\alpha'(v_1)=\alpha'(v_2)=\alpha^*(v_2)\neq \perp$.
    If $\{v_{t-1}, v_t\}\in E_2$ then we can make the same argument to conclude that $\alpha'(v_t)\neq \perp$ and conclude that no points on $Q$ are outliers in $(C',\alpha')$. If $\{v_{t-1},v_t\}\in E_1$ then we know that $v_t\in Z_{2\OPT}$ so $\alpha'(v_t)=\perp$ so $Q$ contains exactly $v_t$ as the only outlier from solution $(C',\alpha')$. However, $v_1$ is an outlier in $(C^*,\alpha^*)$ since it has no incident edges in $E_1$ and thus the number of outliers stays the same.

    It is left to consider paths of length 1 and isolated vertices. Since paths of length 1 always consist of an edge from $E_2$, we know that both endpoints of this edge are outliers in $(C^*,\alpha^*)$. Moreover, isolated vertices are also outliers in $(C^*,\alpha^*)$, so in both cases, there is nothing to show. We conclude that $(C',\alpha')$ has at most $z$ outliers. 
    
    Finally, we bound the cost of this solution. Given a fairlet $\{r,b\}\in \mathcal F_{2\OPT}$ with $\alpha'(r)=\alpha'(b)\neq \perp$. By definition this can only happen when $\alpha^*(r)\neq \perp$ or $\alpha^*(b)\neq \perp$. We assume w.l.o.g. that $\alpha^*(r)\neq \perp$. Then $d(r,\alpha'(r))=d(r,\alpha^*(r))\leq \OPT$ and $d(b,\alpha'(b))\leq d(b,r)+d(r,\alpha^*(r))\leq 3\OPT$. For $x\in Z_{2\OPT}$, there is nothing to show, since we always have $\alpha'(x)=\perp$ by definition. This proves the lemma.   
\end{proof}
\fi

\begin{lemma}
\label{lem:cost_bound_supplier_reduction}
    Let $(C_\delta,\alpha_\delta)$ be the solution computed by \Cref{alg:outliers_both_sides_supplier} on $\mathcal F_\delta$ and  let $r$ be the radius of the solution $(C_R,\alpha_R)$ computed in \cref{alg-line:k-sup-sol}. Then $(C_\delta, \alpha_\delta)$ is $1:1$-fair, has at most $z$ outliers and its cost is upper bounded by $r+\delta$.
\end{lemma}
\iffullappendix
\begin{proof}
     It is easy to see that $(C_\delta,\alpha_\delta)$ agrees with $\mathcal F_\delta$ and thus satisfies $1:1$-fairness.
     Let $Z$ be the set of outliers, we have $Z=Z_{\delta}\cup \{r,b\mid \{r,b\}\in \mathcal F_{\delta}, \alpha_R(r)=\perp\}$. Thus, the number of outliers is at most $|Z_\delta|+2z_\delta=z$. 
     It is left to bound the cost for all points $p\in P\backslash Z$. If $p\in R$ then we have $d(p,\alpha_\delta(r))=d(p,\alpha_R(r))\leq r$. If $p\in B$ then there exists $r\in R$ with $\{r,b\}\in \mathcal F_\delta$ and $d(b,\alpha_\delta(b))=d(b,\alpha_R(r))\leq d(b,r)+d(r,\alpha_R(r))\leq \delta+r$. This proves the lemma.  
\end{proof}
\fi
There exists a $3$-approximation algorithm for $k$-supplier with outliers~\cite{charikar01supplier}. Using this to compute a \(k\)-supplier with outliers solution \((C_R,\alpha_R)\) on \((R',P,d,z_{\delta},k)\) in \Cref{alg:outliers_both_sides_supplier}, we obtain the following result.

\iffullappendix
\ThmOutBothSides*
\begin{proof}
    Let $\OPT$ be the cost of an optimal $1:1$-fair $k$-center solution with outliers.
    We know by \Cref{lem:delta_min_bound} that $\delta_{\min}\leq 2\OPT$. Let $(C_{2\OPT},\alpha_{2\OPT})$ be the solution that we obtain when we execute \Cref{alg:outliers_both_sides_supplier} for $\mathcal F_{2\OPT}$. To compute the $k$-supplier solution in \Cref{alg-line:k-sup-sol} we use the $3$-approximation for $k$-supplier with outliers~\cite{charikar01supplier}. Let $r$ be the radius of the solution $(C_R,\alpha_R)$ computed in \cref{alg-line:k-sup-sol}. We obtain by \Cref{lem:cost_bound_supplier} that $r\leq 9\OPT$. 
    We conclude by \Cref{lem:cost_bound_supplier_reduction} that the cost of $(C_{2\OPT},\alpha_{2\OPT})$ is upper bounded by $r+2\OPT\leq 11\OPT$.
    Now let $(C,\alpha)$ be the solution returned by \Cref{alg:outliers_both_sides_supplier}. We get that 
    \[\cost(C,\alpha)\leq \cost(C_{2\OPT},\alpha_{2\OPT})\leq 11\OPT.\] Furthermore, we know by \Cref{lem:cost_bound_supplier_reduction} that $(C,\alpha)$ is $1:1$-fair and has at most $z$ outliers. 

    Regarding the running time of \Cref{alg:outliers_both_sides_supplier}, there can be $O(n^2)$ possible fairlet decompositions.
    For every such fairlet decomposition, we compute a solution for $k$-supplier with outliers. This can be done in $O(n^3)$ \cite{charikar01supplier}.
    The subsequent assignments can be computed in $O(n)$.
\end{proof}
\fi

\begin{observation}
Notice that the clustering computed in \Cref{alg:outliers_both_sides_supplier} is not necessarily center-aware, since centers can be outliers. For the center-aware version, we can similarly prove a $14$-approximation in the following way:
Suppose we can ensure that $C_R\subset R\backslash Z_{\delta}$ in \Cref{alg-line:k-sup-sol}. Let $r\in C_R$ and $\{r,b\}\in \mathcal F_\delta$ be the fairlet containing $r$. We set $\alpha_\delta(r)=\alpha_\delta(b)=r$ to obtain a center-aware solution. 
By \Cref{lem:cost_bound_supplier}, there exists a solution $(C,\alpha)$ for $k$-supplier with outliers on $(R\backslash Z_\delta, P\backslash Z_\delta,d, z_\delta, k)$ with cost $3\OPT$. For $c\in C$ and its cluster $D$, we can pick an arbitrary point from $D\cap (R\backslash Z_{\delta})$ as a new center. This increases the cost by $3\OPT$. Thus, there exists a solution with cost $6\OPT$ for $k$-center with $z_\delta$ outliers on $R\backslash Z_\delta$. If we replace \Cref{alg-line:k-sup-sol} by a $2$-approximation for $k$-center with outliers \cite{chak20non-uniform}, this yields a $14$-approximation in total.
\end{observation}

\iffullappendix
\section{\texorpdfstring{$1:1$}{1:1} fairness without outliers if \texorpdfstring{$|B|=|R|$}{there are equally many red and blue points}} \label{sec:appendix-1:1-fairness}

For the case that $|R|=|B|$ in the underlying point set, there exists a fairlet decomposition that covers $P$ completely. Chierichetti et al.\ construct fairlets with cost at most $\OPT$ and use these to achieve a 3-approximation for $1:1$-fair $k$-center \cite{chierichetti2017fair}. In the following, we give a slightly modified version of their algorithm and the proof.

The crucial difference from the fairlet construction used in the main body is that here, anchors can come from the whole point set $P$ rather than $A$. Initially, anchors for all point pairs are computed by setting $\anc(r,b) \coloneqq \arg\min_{x\in P} \max\{d(r,x),d(b,x)\}$ for all $(r,b)\in R\times B$. 
We denote the actual distance to be minimized by $w(r,b) \coloneqq \max\{d(r,\anc(r,b)), d(b,\anc(r,b))\}$. 
Using the anchors, a bipartite graph $G=(V, E)$ is constructed, with $V=R\cup B$ and $E=\{(r,b)\in R\times B \mid w(r,b) \leq \delta\}$, i.e., edges between any pair of a red and a blue point are added if the distance to their anchor is bounded by some threshold $\delta$. For $\delta$ large enough (at least for $\delta \ge \OPT$), there exists a perfect matching $M$ in $G$. The smallest such $\delta$ can be found by binary search on the sorted list of pairwise distances. This is formalized in \Cref{alg:compute-1-1-fairlets}.

\begin{algorithm}
    \LinesNumbered
    \caption{\textsc{compute-\(1:1\)-fairlets-OPT}}
    \label{alg:compute-1-1-fairlets}
    \SetKwInOut{Input}{Input}
    \SetKwInOut{Output}{Output}
    \BlankLine
    \Input{Point set $P$, distance metric $d$, integer $k$}
    \Output{Set of fairlets $\F$, anchor assignment $\anc\colon R\times B\to P$}
    \BlankLine
    \For{all $r\in R,\ b\in B$}{
        $\anc(r,b) \gets \arg\min_{x\in P}\max\{d(r,x), d(b,x)\}$\\
        $w(r,b) \gets \max\{d(r,\anc(r,b)), d(b,\anc(r,b))\}$
    }
    $D\gets \{d(x,y)\mid x,y\in P\}$\\
    sort $D$ increasingly\\
    \For{all $\delta \in D$}{
        construct bipartite graph $G_{\delta}=(R\cup B,\ \{(r,b)\in R\times B \mid w(r,b) \leq \delta\})$\\
        $\F\gets$ maximum matching in $G_{\delta}$\\
        \If{$|\F|=|R|$ \label{alg-line:enough-fairlets-found-1:1}}{
            \Return{$\F, \anc$}
        }
    }
\end{algorithm}

\begin{lemma} \label{lem:properties-computed-fairlets-1:1}
	\Cref{alg:compute-1-1-fairlets} computes a fairlet decomposition $\F$ and an anchor assignment $\anc\colon R\times B \to P$ such that
	$\cost(f) = \max_{x\in f}d(x,\anc(f)) \leq \OPT_\text{fair}$ for all $f\in \F$. 
\end{lemma}
\iffullappendix
\begin{proof}
	$\F$ is computed as a maximum matching inside a bipartite graph with edges only between red and blue vertices. Hence, every element of $\F$ is a set $\{r,b\}$ consisting of a red point $r$ and a blue point $b$. Further, all these sets are pairwise disjoint due to the definition of a matching.
	The algorithm stops if a maximum matching of size $|R|$ is found. Such a matching has to cover all red points. This directly implies that $R\subseteq P(\F)$. As $|B|=|R|$, it follows that $P(\F)=P$.
	
	It remains to argue that the condition $|\F| = |R|$ in \Cref{alg-line:enough-fairlets-found-1:1} eventually holds true in one of the loop iterations. 
	Let $\C^*$ be an optimal 1:1-fair clustering excluding $z$ blue outliers with value $\OPT_\text{fair}$.
	Consider the iteration in which $\delta = \OPT_\text{fair}$. Such an iteration exists as the optimal radius can be described by the distance between two points in $P$, and $D$ contains all pairwise distances. 
	Every cluster $C^*\in \C^*$ fulfills $|C^*\cap R|=|C^*\cap B|$. Let $c^*$ be the optimal center of $C^*$. For all $p\in C^*$, $d(p,c^*)\leq \OPT_\text{fair}$. Therefore, for every pair of red and blue points $r\in C^*\cap R,\ b\in C^*\cap B$, an edge $\{r,b\}$ exists inside the bipartite graph $G$. Hence, we can find a matching of size $|C^*\cap R|$ between red and blue points within $C^*$. As $\C^*$ covers $R$ completely, we can find a matching of size $|R|$. 
	The bound on the distances follows directly as $\delta=\OPT_\text{fair}$.
\end{proof}
\fi

\begin{corollary} \label{cor:diameter-of-fairlet}
	Let $\F$ be the fairlet decomposition computed by \Cref{alg:compute-1-1-fairlets}. Let $f = (r,b) \in \F$. Then
	\[ d(r,b) \leq 2\OPT_{\text{fair}}. \]
\end{corollary}
The algorithm can be implemented to run in time $\Theta(n^{3})$. The bottleneck is the computation of the anchors of all possible point pairs.

\begin{lemma} \label{lem:running-time-compute-1-1-fairlets}
    The overall running time of \Cref{alg:compute-1-1-fairlets} is in $\Theta(n^{3})$ if it uses binary search to find the smallest $\delta$ such that a maximum matching on the bipartite graph has size $|R|$.
\end{lemma}
\iffullappendix
\begin{proof}
    There are $\Theta(n^2)$ possible pairs consisting of one red and one blue point. Computing the anchor of any such pair takes $\Theta(n)$ time. This amounts to a total running time of $\Theta(n^3)$ for the computation of all anchors in the first for-loop.
    We show that this part dominates the algorithm's running time.

    The set $D$ consists of $O(n^2)$ elements and can be sorted in $O(n^2\log(n))$.
    For a given $\delta \in D$, we construct a bipartite graph $G_{\delta}$ in $O(n^2)$. Computing a maximum matching can be done in $O(|E|^{1+o(1)}) = O(n^{2+o(1)})$ by finding the maximum flow in the corresponding source-target network \cite{brand2023maxflowalmostlinear}.
    We find the smallest $\delta \in D$ such that $G_{\delta}$ contains a perfect matching by performing a binary search on $D$.
\end{proof}
\fi
The centers are computed by performing farthest-first traversal on $P$ for $k$ iterations and assigning fairlets to the center that is closest to their anchor point.
This approach already yields a valid 3-approximation. Note that it can be made center-aware as follows: As soon as a point $c$ is added to the set of centers, we assign all points inside its containing fairlet $f_c$ to $c$ and remove them from $P$. That way, they cannot become centers later on. The whole algorithm is stated in \Cref{alg:1:1-balanced-3-approx}.

\begin{algorithm}
    \LinesNumbered
    \caption{\textsc{\texorpdfstring{$1:1$}{1:1}-fair-\texorpdfstring{$k$}{k}-center-without-outliers}}
    \label{alg:1:1-balanced-3-approx}
    \SetKwInOut{Input}{Input}
    \SetKwInOut{Output}{Output}
    \BlankLine
    \Input{Point set $P$, distance metric $d$, integer $k$}
    \Output{Fair $k$-center clustering given by a set of centers $C$ and a fair assignment $\alpha\colon P \to C$}
    \BlankLine
    $\mathcal{F}, \anc \gets \textsc{compute-1:1-fairlets}(P,d,k) $ \label{alg-line:call-compute-1:1-fairlets}\\
    choose $c_1\in P$ arbitrarily\\
    $P \gets P\setminus f_{c_1}$\\
    \For{$i=2,\ldots, k$}{
        $c_i \gets \arg\max_{c\in P}\min_{j<i}d(c,c_j)$\\
        $P\gets P\setminus f_{c_i}$\\
    } \label{alg-line:end-of-farthest-first-traversal-1:1-balanced}
    \For{$i=1,\ldots,k$}{
        \For{all $x\in f_{c_i}$}{
            $\alpha(x) \gets c_i$ \label{alg-line:assign-fairlets-to-centers-within}
        }
    }
    \For{$f\in \F\setminus \{f_{c_i}\mid i\leq k\}$}{
        $c_f \gets \arg\min_{c\in C} d(c,\anc(f))$\\
        \For{all $x\in f$}{
            $\alpha(x) \gets c_f$ \label{alg-line:assign-remaining-fairlets-1:1-balanced}
        }
    }
    \Return{$c_1, \ldots, c_k, \alpha$}
\end{algorithm}

\begin{theorem}[$1:1$ fairness, $|R|=|B|$]
    Given a set of points $P=R\cupdot B$ with $|R|=|B|$,
    \Cref{alg:1:1-balanced-3-approx} computes a center-aware 3-approximation for 1:1-fair $k$-center.
\end{theorem}
\iffullappendix
\begin{proof}
    \Cref{alg:compute-1-1-fairlets} called in \Cref{alg-line:call-compute-1:1-fairlets} outputs a set of fairlets $\F$ and their corresponding anchors.
    Let $C = \{c_1,\ldots,c_k\}$ be the set of centers computed with farthest-first traversal until \Cref{alg-line:end-of-farthest-first-traversal-1:1-balanced}. Note that this set does not change throughout the algorithm.
    Let $\F_C = \{f_c \mid c\in C\} \subseteq \F$ denote the set of fairlets containing centers.

    For $x\in P$ let $f_x\in \F$ denote the fairlet with $x\in f_x$.
    For $x\in P$, we distinguish the following two cases.
    Either $x\in P(\F_C)$. Then, $\alpha(x) \in f_x$ by \Cref{alg-line:assign-fairlets-to-centers-within}. Therefore, $d(x,\alpha(x)) \le 2\OPT_{\text{fair}}$ by \Cref{lem:properties-computed-fairlets-1:1} and triangle inequality.

    Otherwise, if $x\in P\setminus P(\F_C)$, then $x$ is assigned to the center closest to $\anc(f_x)$.
    By \Cref{lem:properties-computed-fairlets-1:1}, $d(x,\anc(f_x)) \leq \OPT_{fair}$.  
    It remains to show that $d(\anc(f),c_f) \leq 2\OPT_{\text{fair}}$ for all $f\in \F$ and $c_f \coloneqq \arg\min_{c\in C}d(c,\anc(f))$. Then, using triangle inequality and \Cref{lem:properties-computed-fairlets-1:1}, we can conclude that \[d(x,\alpha(x)) = d(x,c_f) \leq d(x,\anc(f_x)) + d(\anc(f_x), c_f) \leq 3\OPT_{\text{fair}}.\]

    We make a case distinction. 
    \begin{itemize}
        \item If $\anc(f)\in P(\F_C)$, then there exists a center $c'\in C \cap f_{\anc(f)}$. By \Cref{cor:diameter-of-fairlet},  $d(\anc(f),c_f) \leq d(\anc(f),c') \leq 2\OPT_{\text{fair}}$.
        \item Otherwise, $\anc(f)\in P\setminus P(\F_C)$. 
        Let $c_{k+1} \coloneqq \arg\max_{x\in P\setminus P(\F_C)}\min_{j\leq k}d(x,c_j) $ and\linebreak $\Delta \coloneqq \min_{j\leq k}d(c_{k+1},c_j) $. 
        By construction, $c_1,\ldots,c_k,c_{k+1}$ form $k+1$ points with pairwise distance of at least $\Delta$. In any (optimal) solution, at least two of these points need to be contained in the same cluster. This implies that the diameter of any optimal clustering is at least $\Delta$, and therefore, $\OPT_{\text{fair}} \geq \frac{\Delta}{2} $.
        By construction of $c_{k+1}$, 
        \[d(c_f,\anc(f)) \leq \min_{j\leq k} d(c_{k+1},c_j) = \Delta \leq 2\OPT_{\text{fair}}.\] 
    \end{itemize}
\end{proof}
\fi
\fi

\iffullappendix
\section{$8$-approximation guarantee when choosing centers from the inlier point set}
\label{sec:appendix-choice-of-centers-8-approx}

For completeness, we complement our lower bound from Section~\ref{sec:considerations:choosingcenters} by showing that this algorithm variant guarantees an \(8\)-approximation.
A difficulty in the analysis is that the set of covered points might not coincide with the set of points covered by some optimal solution. 
Consequently, our solution might choose points as centers that are outliers in the optimal solution, which makes it harder to compare with. 

Let $(C^*,\OPT)$ be some optimal solution.
We will first show that $(C^*,3\OPT)$ is a feasible $k$-center solution on $P(\F)$. We can use this observation to show that $(C,6\OPT)$ is a feasible $k$-center solution on $P(\F)$ as well, where $C$ are the centers computed by $\textsc{compute-centers}(P(\F),d,k)$. Through the assignment of fairlets, we incur an extra cost of $2\OPT$ as the anchors might be outliers in our solution. 

\begin{corollary} \label{cor:old-approach-our-solution-can-be-covered-by-3OPT}
Consider an optimal solution with centers $c_1^*,...,c_k^*$ and optimal radius OPT. 
Then, the set $P\setminus Z$ can be covered with centers $c_1^*,...,c_k^*$ and radius $3\OPT$.
\end{corollary}
\iffullappendix
\begin{proof}
    Let $p\in P\backslash Z$. Then, there exists a fairlet $f\in \F$ such that $p\in f$. Either $p$ is a red or a blue point inside $f=\{b,r\}$.
    If $p=r$, then $d(p,c_i^*) \le \OPT$ for some $i\le k$ because red points are never outliers.

    Otherwise, let $p=b$. As argued above, there exists a center $c_i^*$ such that $d(r,c_i^*)\le \OPT$ for the red point $r$ in $f$. The distance from any point to its anchor point is at most $\OPT$ by construction. 
    By, triangle inequality
    \[ d(c^*_i,b)\leq d(c^*_i,r)+d(r,\anc(f))+d(\anc(f),b)\leq \OPT + \OPT + \OPT = 3\OPT. \]
\end{proof}
\fi

\begin{theorem}
    \Cref{alg:general-case-outlier-algorithm} that uses \textsc{compute-centers}($P(\F),d,k$) to compute the centers yields an $8$-approximation for $1:1$-fair $k$-center with outliers.
\end{theorem}

\iffullappendix
\begin{proof}
    Let $c_1,...,c_k$ be the set of centers computed by farthest-first traversal on $P(\F)$ and $c_1^*,\ldots,c_k^*$ be the set of centers of some fixed optimal $1:1$-fair $k$-center with outliers solution. Let $c_{k+1}$ be the point that is farthest away from $c_1,...,c_k$. By \Cref{cor:old-approach-our-solution-can-be-covered-by-3OPT} and the pigeonhole principle, there exist two distinct points $c_i,c_j$ in $c_1,...,c_{k+1}$ and an optimal center $c^*$ such that $d(c_i,c^*)\le 3\OPT$ and $d(c_j,c^*)\le 3\OPT$. By triangle inequality, $d(c_i,c_j) \le 6\OPT$. 
    The distance $d(c_{k+1},\{c_1,...,c_k\})$ corresponds to the radius of the solution that assigns every point in $P(\F)$ to its closest center. It must be less than or equal to $d(c_i,c_j)$, since otherwise $c_{k+1}$ would be selected as a center in iteration $i$ or $j$: 
    \[d(c_{k+1},\{c_1,...,c_k\})\leq d(c_i,c_j)\leq 6\OPT.\] 
    Hence, for all $p\in P\backslash Z$, there exists a center $c_i$ with $i\in\{1,...,k\}$ such that $d(p,c_i)\leq 6\OPT$. 
    However, such an assignment is not necessarily fair.
    
    The algorithm assigns the points of a fairlet to the center that is closest to their anchor.
    Let $f=\{p_1,p_2\}\in \F$ and $\anc(f)$ be the anchor.
    It might happen that $\anc(f)$ is an outlier.
    By construction, we have $d(p_1,\anc(f))\leq \OPT$ and $d(\anc(f),p_2)\leq \OPT$. The algorithm assigns $p_1$ and $p_2$ to the center that is closest to $\anc(f)$. There exists a center $c$ such that $d(c,p_1)\leq 6\OPT$ and therefore,
    \[d(c,\anc(f))\leq d(c,p_1)+d(p_1,\anc(f))\leq 6\OPT + \OPT = 7\OPT.\]
    Since $\alpha(p_1)$ minimizes the distance to $\anc(f)$, we have: $d(\alpha(p_1),\anc(f)) \le d(c,\anc(f))\le 7\OPT$. It follows that 
    \[d(\alpha(p_1),p_1) \leq d(\alpha(p_1),\anc(f)) + d(\anc(f),p_1)\leq 7\OPT+\OPT=8\OPT\]
    and analogously $d(\alpha(p_2),p_2)\le 8\OPT$.
\end{proof}
\fi
\fi

\iffullappendix
\clearpage
\section{Numerical results for clustering costs}\label{appendix:costvalues}

\begin{minipage}[t]{0.48\textwidth}
Data Set \textbf{bank}\\[4pt]
{\small
\begin{tblr}{lrrrr}
\toprule
$k$ & \SetCell[c=2]{c} cost & & \SetCell[c=2]{c} fairlet cost & \\
 & ours & Chier & ours & Chier \\
\midrule
1 & \textbf{1.391} & 1.496 & \textbf{0.726} & 0.957 \\
2 & \textbf{1.178} & 1.271 & \textbf{0.726} & 0.929 \\
3 & \textbf{1.078} & 1.183 & \textbf{0.726} & 0.929 \\
4 & \textbf{1.014} & 1.123 & \textbf{0.726} & 0.949 \\
5 & \textbf{0.939} & 1.062 & \textbf{0.726} & 0.936 \\
6 & \textbf{0.920} & 1.058 & \textbf{0.726} & 0.955 \\
7 & \textbf{0.901} & 1.049 & \textbf{0.726} & 0.970 \\
8 & \textbf{0.874} & 1.005 & \textbf{0.726} & 0.939 \\
9 & \textbf{0.854} & 0.998 & \textbf{0.726} & 0.945 \\
10 & \textbf{0.829} & 0.973 & \textbf{0.726} & 0.918 \\
15 & \textbf{0.787} & 0.948 & \textbf{0.726} & 0.929 \\
20 & \textbf{0.764} & 0.941 & \textbf{0.726} & 0.935 \\
25 & \textbf{0.751} & 0.917 & \textbf{0.726} & 0.914 \\
30 & \textbf{0.737} & 0.909 & \textbf{0.726} & 0.911 \\
40 & \textbf{0.728} & 0.955 & \textbf{0.726} & 0.957 \\
50 & \textbf{0.724} & 0.934 & \textbf{0.726} & 0.935 \\
100 & \textbf{0.731} & 0.910 & \textbf{0.726} & 0.912 \\
\bottomrule
\end{tblr}}
\end{minipage}
\hfill
\begin{minipage}[t]{0.48\textwidth}
Data Set \textbf{census}\\[4pt]
{\small
\begin{tblr}{lrrrr}
\toprule
$k$ & \SetCell[c=2]{c} cost & & \SetCell[c=2]{c} fairlet cost & \\
 & ours & Chier & ours & Chier \\
\midrule
1 & \textbf{1.558} & 1.817 & \textbf{0.854} & 1.625 \\
2 & \textbf{1.388} & 1.726 & \textbf{0.854} & 1.704 \\
3 & \textbf{1.217} & 1.688 & \textbf{0.854} & 1.675 \\
4 & \textbf{1.158} & 1.650 & \textbf{0.854} & 1.655 \\
5 & \textbf{1.111} & 1.650 & \textbf{0.854} & 1.659 \\
6 & \textbf{1.046} & 1.641 & \textbf{0.854} & 1.662 \\
7 & \textbf{1.019} & 1.643 & \textbf{0.854} & 1.648 \\
8 & \textbf{0.984} & 1.617 & \textbf{0.854} & 1.638 \\
9 & \textbf{0.960} & 1.657 & \textbf{0.854} & 1.679 \\
10 & \textbf{0.951} & 1.653 & \textbf{0.854} & 1.688 \\
15 & \textbf{0.902} & 1.626 & \textbf{0.854} & 1.653 \\
20 & \textbf{0.881} & 1.599 & \textbf{0.854} & 1.620 \\
25 & \textbf{0.879} & 1.576 & \textbf{0.854} & 1.613 \\
30 & \textbf{0.863} & 1.586 & \textbf{0.854} & 1.615 \\
40 & \textbf{0.853} & 1.659 & \textbf{0.854} & 1.677 \\
50 & \textbf{0.850} & 1.635 & \textbf{0.854} & 1.666 \\
100 & \textbf{0.851} & 1.579 & \textbf{0.854} & 1.600 \\
\bottomrule
\end{tblr}}
\end{minipage}

\begin{minipage}[t]{0.48\textwidth}
Data Set \textbf{census} (pr. attr. \textbf{race})\\[4pt]
{\small
\begin{tblr}{lrrrr}
\toprule
$k$ & \SetCell[c=2]{c} cost & & \SetCell[c=2]{c} fairlet cost & \\
 & ours & Chier & ours & Chier \\
\midrule
1 & \textbf{1.655} & 1.827 & \textbf{0.831} & 0.904 \\
2 & \textbf{1.637} & 1.810 & \textbf{0.821} & 0.901 \\
3 & \textbf{1.619} & 1.839 & \textbf{0.822} & 0.914 \\
4 & \textbf{1.524} & 1.841 & \textbf{0.811} & 0.918 \\
5 & \textbf{1.520} & 1.838 & \textbf{0.855} & 0.957 \\
6 & \textbf{1.520} & 1.818 & \textbf{0.901} & 1.003 \\
7 & \textbf{1.447} & 1.809 & \textbf{0.965} & 1.015 \\
8 & \textbf{1.702} & 1.947 & \textbf{1.088} & 1.151 \\
9 & 1.838 & \textbf{1.811} & \textbf{0.604} & 0.707 \\
\bottomrule
\end{tblr}}
\end{minipage}
\hfill
\begin{minipage}[t]{0.48\textwidth}
Data set \textbf{diabetes}\\[4pt]
{\small
\begin{tblr}{lrrrr}
\toprule
$k$ & \SetCell[c=2]{c} cost & & \SetCell[c=2]{c} fairlet cost & \\
 & ours & Chier & ours & Chier \\
\midrule
1 & \textbf{2.609} & 2.641 & \textbf{1.690} & 1.812 \\
2 & \textbf{2.442} & 2.471 & \textbf{1.690} & 1.820 \\
3 & \textbf{2.321} & 2.378 & \textbf{1.690} & 1.813 \\
4 & \textbf{2.266} & 2.325 & \textbf{1.690} & 1.805 \\
5 & \textbf{2.218} & 2.250 & \textbf{1.690} & 1.808 \\
6 & \textbf{2.166} & 2.212 & \textbf{1.690} & 1.797 \\
7 & \textbf{2.141} & 2.170 & \textbf{1.690} & 1.808 \\
8 & \textbf{2.099} & 2.168 & \textbf{1.690} & 1.811 \\
9 & \textbf{2.080} & 2.111 & \textbf{1.690} & 1.816 \\
10 & \textbf{2.048} & 2.112 & \textbf{1.690} & 1.817 \\
15 & \textbf{1.953} & 2.030 & \textbf{1.690} & 1.822 \\
20 & \textbf{1.897} & 1.967 & \textbf{1.690} & 1.809 \\
25 & \textbf{1.858} & 1.940 & \textbf{1.690} & 1.815 \\
30 & \textbf{1.824} & 1.907 & \textbf{1.690} & 1.804 \\
40 & \textbf{1.792} & 1.885 & \textbf{1.690} & 1.808 \\
50 & \textbf{1.790} & 1.892 & \textbf{1.690} & 1.806 \\
100 & \textbf{1.819} & 1.917 & \textbf{1.690} & 1.805 \\
\bottomrule
\end{tblr}}
\end{minipage}

\begin{minipage}[t]{0.48\textwidth}
Data set \textbf{diabetes} (pr. attr. \textbf{race})\\[4pt]
{\small
\begin{tblr}{lrrrr}
\toprule
$k$ & \SetCell[c=2]{c} cost & & \SetCell[c=2]{c} fairlet cost & \\
 & ours & Chier & ours & Chier \\
\midrule
1 & \textbf{2.601} & 2.651 & \textbf{1.726} & 1.896 \\
2 & \textbf{2.550} & 2.596 & \textbf{1.726} & 1.907 \\
3 & \textbf{2.523} & 2.581 & \textbf{1.721} & 1.920 \\
4 & \textbf{2.496} & 2.574 & \textbf{1.688} & 1.874 \\
5 & \textbf{2.485} & 2.586 & \textbf{1.683} & 1.892 \\
6 & \textbf{2.467} & 2.573 & \textbf{1.653} & 1.832 \\
7 & \textbf{2.450} & 2.555 & \textbf{1.648} & 1.864 \\
8 & \textbf{2.430} & 2.581 & \textbf{1.581} & 1.772 \\
9 & \textbf{2.418} & 2.578 & \textbf{1.574} & 1.749 \\
10 & \textbf{2.311} & 2.522 & \textbf{1.567} & 1.698 \\
\bottomrule
\end{tblr}}
\end{minipage}
\hfill
\begin{minipage}[t]{0.48\textwidth}
Data set \textbf{income}\\[4pt]
{\small
\begin{tblr}{lrrrr}
\toprule
$k$ & \SetCell[c=2]{c} cost & & \SetCell[c=2]{c} fairlet cost & \\
 & ours & Chier & ours & Chier \\
\midrule
1 & \textbf{2.972} & 3.015 & \textbf{1.604} & 1.791 \\
2 & \textbf{2.702} & 2.752 & \textbf{1.604} & 1.789 \\
3 & \textbf{2.552} & 2.598 & \textbf{1.604} & 1.788 \\
4 & \textbf{2.430} & 2.496 & \textbf{1.604} & 1.784 \\
5 & \textbf{2.354} & 2.422 & \textbf{1.604} & 1.787 \\
6 & \textbf{2.292} & 2.377 & \textbf{1.604} & 1.787 \\
7 & \textbf{2.239} & 2.336 & \textbf{1.604} & 1.790 \\
8 & \textbf{2.200} & 2.287 & \textbf{1.604} & 1.789 \\
9 & \textbf{2.162} & 2.239 & \textbf{1.604} & 1.790 \\
10 & \textbf{2.130} & 2.222 & \textbf{1.604} & 1.791 \\
15 & \textbf{2.006} & 2.086 & \textbf{1.604} & 1.782 \\
20 & \textbf{1.927} & 2.031 & \textbf{1.604} & 1.787 \\
25 & \textbf{1.865} & 1.969 & \textbf{1.604} & 1.778 \\
30 & \textbf{1.814} & 1.933 & \textbf{1.604} & 1.781 \\
40 & \textbf{1.758} & 1.890 & \textbf{1.604} & 1.794 \\
50 & \textbf{1.719} & 1.849 & \textbf{1.604} & 1.780 \\
100 & \textbf{1.665} & 1.837 & \textbf{1.604} & 1.790 \\
\bottomrule
\end{tblr}}
\end{minipage}
\clearpage

\twocolumn[]
\section{Input ratios and output ratios}\label{appendixinputandoutputratios}

{\small
\begin{longtblr}[
  caption = {bank},
]{
  colspec = {lrrrr},
  rowhead = 1,
}
\toprule
Inst. & G1 & G2 & Ratio & Ideal \\
\midrule
00000 & \textbf{421} & 579 & 1.00:1.38 & 1:1 \\
00001 & \textbf{376} & 624 & 1.00:1.66 & 1:1 \\
00002 & \textbf{425} & 575 & 1.00:1.35 & 1:1 \\
00003 & \textbf{415} & 585 & 1.00:1.41 & 1:1 \\
00004 & \textbf{380} & 620 & 1.00:1.63 & 1:1 \\
00005 & \textbf{395} & 605 & 1.00:1.53 & 1:1 \\
00006 & \textbf{410} & 590 & 1.00:1.44 & 1:1 \\
00007 & \textbf{393} & 607 & 1.00:1.54 & 1:1 \\
00008 & \textbf{419} & 581 & 1.00:1.39 & 1:1 \\
00009 & \textbf{417} & 583 & 1.00:1.40 & 1:1 \\
00010 & \textbf{406} & 594 & 1.00:1.46 & 1:1 \\
00011 & \textbf{388} & 612 & 1.00:1.58 & 1:1 \\
00012 & \textbf{418} & 582 & 1.00:1.39 & 1:1 \\
00013 & \textbf{422} & 578 & 1.00:1.37 & 1:1 \\
00014 & \textbf{397} & 603 & 1.00:1.52 & 1:1 \\
00015 & \textbf{406} & 594 & 1.00:1.46 & 1:1 \\
00016 & \textbf{380} & 620 & 1.00:1.63 & 1:1 \\
00017 & \textbf{401} & 599 & 1.00:1.49 & 1:1 \\
00018 & \textbf{421} & 579 & 1.00:1.38 & 1:1 \\
00019 & \textbf{359} & 641 & 1.00:1.79 & 1:1 \\
00020 & \textbf{376} & 624 & 1.00:1.66 & 1:1 \\
00021 & \textbf{387} & 613 & 1.00:1.58 & 1:1 \\
00022 & \textbf{407} & 593 & 1.00:1.46 & 1:1 \\
00023 & \textbf{401} & 599 & 1.00:1.49 & 1:1 \\
00024 & \textbf{381} & 619 & 1.00:1.62 & 1:1 \\
00025 & \textbf{396} & 604 & 1.00:1.53 & 1:1 \\
00026 & \textbf{402} & 598 & 1.00:1.49 & 1:1 \\
00027 & \textbf{412} & 588 & 1.00:1.43 & 1:1 \\
00028 & \textbf{377} & 623 & 1.00:1.65 & 1:1 \\
00029 & \textbf{397} & 603 & 1.00:1.52 & 1:1 \\
00030 & \textbf{393} & 607 & 1.00:1.54 & 1:1 \\
00031 & \textbf{386} & 614 & 1.00:1.59 & 1:1 \\
00032 & \textbf{390} & 610 & 1.00:1.56 & 1:1 \\
00033 & \textbf{388} & 612 & 1.00:1.58 & 1:1 \\
00034 & \textbf{388} & 612 & 1.00:1.58 & 1:1 \\
00035 & \textbf{412} & 588 & 1.00:1.43 & 1:1 \\
00036 & \textbf{390} & 610 & 1.00:1.56 & 1:1 \\
00037 & \textbf{405} & 595 & 1.00:1.47 & 1:1 \\
00038 & \textbf{391} & 609 & 1.00:1.56 & 1:1 \\
00039 & \textbf{387} & 613 & 1.00:1.58 & 1:1 \\
00040 & \textbf{407} & 593 & 1.00:1.46 & 1:1 \\
00041 & \textbf{396} & 604 & 1.00:1.53 & 1:1 \\
00042 & \textbf{425} & 575 & 1.00:1.35 & 1:1 \\
00043 & \textbf{373} & 627 & 1.00:1.68 & 1:1 \\
00044 & \textbf{402} & 598 & 1.00:1.49 & 1:1 \\
\bottomrule
\end{longtblr}}
\onecolumn

\twocolumn[]
{\small
\begin{longtblr}[
  caption = {census},
]{
  colspec = {lrrrr},
  rowhead = 1,
}
\toprule
Inst. & G1 & G2 & Ratio & Ideal \\
\midrule
00000 & 388 & \textbf{212} & 1.83:1.00 & 1:1 \\
00001 & 379 & \textbf{221} & 1.71:1.00 & 1:1 \\
00002 & 368 & \textbf{232} & 1.59:1.00 & 1:1 \\
00003 & 406 & \textbf{194} & 2.09:1.00 & 2:1 \\
00004 & 416 & \textbf{184} & 2.26:1.00 & 2:1 \\
00005 & 415 & \textbf{185} & 2.24:1.00 & 2:1 \\
00006 & 411 & \textbf{189} & 2.17:1.00 & 2:1 \\
00007 & 398 & \textbf{202} & 1.97:1.00 & 1:1 \\
00008 & 408 & \textbf{192} & 2.12:1.00 & 2:1 \\
00009 & 411 & \textbf{189} & 2.17:1.00 & 2:1 \\
00010 & 399 & \textbf{201} & 1.99:1.00 & 1:1 \\
00011 & 404 & \textbf{196} & 2.06:1.00 & 2:1 \\
00012 & 402 & \textbf{198} & 2.03:1.00 & 2:1 \\
00013 & 398 & \textbf{202} & 1.97:1.00 & 1:1 \\
00014 & 423 & \textbf{177} & 2.39:1.00 & 2:1 \\
00015 & 406 & \textbf{194} & 2.09:1.00 & 2:1 \\
00016 & 390 & \textbf{210} & 1.86:1.00 & 1:1 \\
00017 & 364 & \textbf{236} & 1.54:1.00 & 1:1 \\
00018 & 390 & \textbf{210} & 1.86:1.00 & 1:1 \\
00019 & 413 & \textbf{187} & 2.21:1.00 & 2:1 \\
00020 & 405 & \textbf{195} & 2.08:1.00 & 2:1 \\
00021 & 386 & \textbf{214} & 1.80:1.00 & 1:1 \\
00022 & 404 & \textbf{196} & 2.06:1.00 & 2:1 \\
00023 & 401 & \textbf{199} & 2.02:1.00 & 2:1 \\
00024 & 405 & \textbf{195} & 2.08:1.00 & 2:1 \\
00025 & 417 & \textbf{183} & 2.28:1.00 & 2:1 \\
00026 & 398 & \textbf{202} & 1.97:1.00 & 1:1 \\
00027 & 396 & \textbf{204} & 1.94:1.00 & 1:1 \\
00028 & 423 & \textbf{177} & 2.39:1.00 & 2:1 \\
00029 & 407 & \textbf{193} & 2.11:1.00 & 2:1 \\
00030 & 391 & \textbf{209} & 1.87:1.00 & 1:1 \\
00031 & 398 & \textbf{202} & 1.97:1.00 & 1:1 \\
00032 & 401 & \textbf{199} & 2.02:1.00 & 2:1 \\
00033 & 402 & \textbf{198} & 2.03:1.00 & 2:1 \\
00034 & 395 & \textbf{205} & 1.93:1.00 & 1:1 \\
00035 & 410 & \textbf{190} & 2.16:1.00 & 2:1 \\
00036 & 390 & \textbf{210} & 1.86:1.00 & 1:1 \\
00037 & 382 & \textbf{218} & 1.75:1.00 & 1:1 \\
00038 & 413 & \textbf{187} & 2.21:1.00 & 2:1 \\
00039 & 418 & \textbf{182} & 2.30:1.00 & 2:1 \\
00040 & 416 & \textbf{184} & 2.26:1.00 & 2:1 \\
00041 & 410 & \textbf{190} & 2.16:1.00 & 2:1 \\
00042 & 377 & \textbf{223} & 1.69:1.00 & 1:1 \\
00043 & 412 & \textbf{188} & 2.19:1.00 & 2:1 \\
00044 & 407 & \textbf{193} & 2.11:1.00 & 2:1 \\
00045 & 403 & \textbf{197} & 2.05:1.00 & 2:1 \\
00046 & 398 & \textbf{202} & 1.97:1.00 & 1:1 \\
00047 & 402 & \textbf{198} & 2.03:1.00 & 2:1 \\
00048 & 404 & \textbf{196} & 2.06:1.00 & 2:1 \\
00049 & 417 & \textbf{183} & 2.28:1.00 & 2:1 \\
00050 & 411 & \textbf{189} & 2.17:1.00 & 2:1 \\
00051 & 395 & \textbf{205} & 1.93:1.00 & 1:1 \\
00052 & 402 & \textbf{198} & 2.03:1.00 & 2:1 \\
00053 & 393 & \textbf{207} & 1.90:1.00 & 1:1 \\
\bottomrule
\end{longtblr}}
\onecolumn

{\small
\begin{longtblr}[
  caption = {census (protected attribute: race)},
  label = {census-race},
]{
  colspec = {lrrrrrrr},
  rowhead = 1,
}
\toprule
Inst. & G1 & G2 & G3 & G4 & G5 & Ratio & Ideal \\
\midrule
00000 & 4 & \textbf{3} & 18 & 53 & 522 & 1.33:1.00:6.00:17.67:174.00 & 1:1:6:17:174 \\
00001 & \textbf{4} & 11 & 18 & 54 & 513 & 1.00:2.75:4.50:13.50:128.25 & 1:2:4:13:128 \\
00002 & \textbf{5} & 6 & 29 & 54 & 506 & 1.00:1.20:5.80:10.80:101.20 & 1:1:5:10:101 \\
00003 & \textbf{4} & 5 & 16 & 70 & 505 & 1.00:1.25:4.00:17.50:126.25 & 1:1:4:17:126 \\
00004 & \textbf{3} & 9 & 14 & 54 & 520 & 1.00:3.00:4.67:18.00:173.33 & 1:3:4:18:173 \\
00005 & 8 & \textbf{5} & 15 & 49 & 523 & 1.60:1.00:3.00:9.80:104.60 & 1:1:3:9:104 \\
00006 & 5 & \textbf{4} & 14 & 48 & 529 & 1.25:1.00:3.50:12.00:132.25 & 1:1:3:12:132 \\
00007 & \textbf{3} & 5 & 15 & 61 & 516 & 1.00:1.67:5.00:20.33:172.00 & 1:1:5:20:172 \\
00008 & \textbf{3} & 5 & 24 & 64 & 504 & 1.00:1.67:8.00:21.33:168.00 & 1:1:8:21:168 \\
00009 & 5 & \textbf{2} & 16 & 55 & 522 & 2.50:1.00:8.00:27.50:261.00 & 2:1:8:27:261 \\
00010 & 8 & \textbf{7} & 15 & 57 & 513 & 1.14:1.00:2.14:8.14:73.29 & 1:1:2:8:73 \\
00011 & \textbf{8} & \textbf{8} & 19 & 57 & 508 & 1.00:1.00:2.38:7.12:63.50 & 1:1:2:7:63 \\
00012 & \textbf{7} & \textbf{7} & 24 & 66 & 496 & 1.00:1.00:3.43:9.43:70.86 & 1:1:3:9:70 \\
00013 & \textbf{5} & 7 & 12 & 59 & 517 & 1.00:1.40:2.40:11.80:103.40 & 1:1:2:11:103 \\
00014 & 7 & \textbf{5} & 14 & 50 & 524 & 1.40:1.00:2.80:10.00:104.80 & 1:1:2:10:104 \\
00015 & \textbf{7} & 8 & 20 & 59 & 506 & 1.00:1.14:2.86:8.43:72.29 & 1:1:2:8:72 \\
00016 & \textbf{5} & 7 & 24 & 55 & 509 & 1.00:1.40:4.80:11.00:101.80 & 1:1:4:11:101 \\
00017 & 3 & \textbf{2} & 20 & 66 & 509 & 1.50:1.00:10.00:33.00:254.50 & 1:1:10:33:254 \\
00018 & 6 & \textbf{3} & 12 & 56 & 523 & 2.00:1.00:4.00:18.67:174.33 & 2:1:4:18:174 \\
00019 & 9 & \textbf{3} & 12 & 64 & 512 & 3.00:1.00:4.00:21.33:170.67 & 3:1:4:21:170 \\
00020 & 8 & \textbf{3} & 16 & 62 & 511 & 2.67:1.00:5.33:20.67:170.33 & 2:1:5:20:170 \\
00021 & 9 & \textbf{5} & 28 & 60 & 498 & 1.80:1.00:5.60:12.00:99.60 & 1:1:5:12:99 \\
00022 & \textbf{1} & 5 & 21 & 60 & 513 & 1.00:5.00:21.00:60.00:513.00 & 1:5:21:60:513 \\
00023 & \textbf{4} & 6 & 15 & 61 & 514 & 1.00:1.50:3.75:15.25:128.50 & 1:1:3:15:128 \\
00024 & \textbf{4} & \textbf{4} & 16 & 66 & 510 & 1.00:1.00:4.00:16.50:127.50 & 1:1:4:16:127 \\
00025 & 9 & \textbf{7} & 11 & 50 & 523 & 1.29:1.00:1.57:7.14:74.71 & 1:1:1:7:74 \\
00026 & \textbf{2} & 3 & 25 & 77 & 493 & 1.00:1.50:12.50:38.50:246.50 & 1:1:12:38:246 \\
00027 & 6 & \textbf{5} & 18 & 50 & 521 & 1.20:1.00:3.60:10.00:104.20 & 1:1:3:10:104 \\
00028 & 4 & \textbf{3} & 24 & 56 & 513 & 1.33:1.00:8.00:18.67:171.00 & 1:1:8:18:171 \\
00029 & \textbf{3} & \textbf{3} & 20 & 51 & 523 & 1.00:1.00:6.67:17.00:174.33 & 1:1:6:17:174 \\
00030 & \textbf{3} & 6 & 27 & 57 & 507 & 1.00:2.00:9.00:19.00:169.00 & 1:2:9:19:169 \\
00031 & \textbf{5} & 9 & 14 & 63 & 509 & 1.00:1.80:2.80:12.60:101.80 & 1:1:2:12:101 \\
00032 & 6 & \textbf{5} & 17 & 61 & 511 & 1.20:1.00:3.40:12.20:102.20 & 1:1:3:12:102 \\
00033 & \textbf{3} & 7 & 21 & 51 & 518 & 1.00:2.33:7.00:17.00:172.67 & 1:2:7:17:172 \\
00034 & 7 & \textbf{6} & 21 & 55 & 511 & 1.17:1.00:3.50:9.17:85.17 & 1:1:3:9:85 \\
00035 & \textbf{6} & \textbf{6} & 18 & 41 & 529 & 1.00:1.00:3.00:6.83:88.17 & 1:1:3:6:88 \\
00036 & \textbf{3} & 7 & 21 & 63 & 506 & 1.00:2.33:7.00:21.00:168.67 & 1:2:7:21:168 \\
00037 & \textbf{5} & 6 & 23 & 65 & 501 & 1.00:1.20:4.60:13.00:100.20 & 1:1:4:13:100 \\
00038 & \textbf{1} & 6 & 22 & 49 & 522 & 1.00:6.00:22.00:49.00:522.00 & 1:6:22:49:522 \\
00039 & 4 & \textbf{3} & 27 & 68 & 498 & 1.33:1.00:9.00:22.67:166.00 & 1:1:9:22:166 \\
00040 & 6 & \textbf{2} & 34 & 62 & 496 & 3.00:1.00:17.00:31.00:248.00 & 3:1:17:31:248 \\
00041 & \textbf{5} & 7 & 13 & 35 & 540 & 1.00:1.40:2.60:7.00:108.00 & 1:1:2:7:108 \\
00042 & \textbf{5} & 6 & 23 & 44 & 522 & 1.00:1.20:4.60:8.80:104.40 & 1:1:4:8:104 \\
00043 & 6 & \textbf{5} & 22 & 73 & 494 & 1.20:1.00:4.40:14.60:98.80 & 1:1:4:14:98 \\
00044 & \textbf{2} & \textbf{2} & 20 & 62 & 514 & 1.00:1.00:10.00:31.00:257.00 & 1:1:10:31:257 \\
00045 & \textbf{5} & 7 & 14 & 46 & 528 & 1.00:1.40:2.80:9.20:105.60 & 1:1:2:9:105 \\
00046 & \textbf{4} & 5 & 19 & 62 & 510 & 1.00:1.25:4.75:15.50:127.50 & 1:1:4:15:127 \\
00047 & \textbf{4} & 6 & 21 & 54 & 515 & 1.00:1.50:5.25:13.50:128.75 & 1:1:5:13:128 \\
00048 & \textbf{5} & 7 & 19 & 56 & 513 & 1.00:1.40:3.80:11.20:102.60 & 1:1:3:11:102 \\
00049 & \textbf{8} & 9 & 21 & 63 & 499 & 1.00:1.12:2.62:7.88:62.38 & 1:1:2:7:62 \\
00050 & \textbf{3} & 12 & 15 & 66 & 504 & 1.00:4.00:5.00:22.00:168.00 & 1:4:5:22:168 \\
00051 & \textbf{9} & 12 & 22 & 53 & 504 & 1.00:1.33:2.44:5.89:56.00 & 1:1:2:5:56 \\
00052 & \textbf{5} & 6 & 20 & 68 & 501 & 1.00:1.20:4.00:13.60:100.20 & 1:1:4:13:100 \\
00053 & \textbf{2} & 6 & 16 & 44 & 532 & 1.00:3.00:8.00:22.00:266.00 & 1:3:8:22:266 \\
\bottomrule
\end{longtblr}}
\onecolumn

\twocolumn[]
{\small
\begin{longtblr}[
  caption = {diabetes},
]{
  colspec = {lrrrr},
  rowhead = 1,
}
\toprule
Inst. & G1 & G2 & Ratio & Ideal \\
\midrule
00000 & \textbf{469} & 531 & 1.00:1.13 & 1:1 \\
00001 & \textbf{437} & 563 & 1.00:1.29 & 1:1 \\
00002 & \textbf{466} & 534 & 1.00:1.15 & 1:1 \\
00003 & \textbf{468} & 532 & 1.00:1.14 & 1:1 \\
00004 & \textbf{445} & 555 & 1.00:1.25 & 1:1 \\
00005 & \textbf{453} & 547 & 1.00:1.21 & 1:1 \\
00006 & \textbf{465} & 535 & 1.00:1.15 & 1:1 \\
00007 & \textbf{450} & 550 & 1.00:1.22 & 1:1 \\
00008 & \textbf{471} & 529 & 1.00:1.12 & 1:1 \\
00009 & \textbf{444} & 556 & 1.00:1.25 & 1:1 \\
00010 & \textbf{461} & 539 & 1.00:1.17 & 1:1 \\
00011 & \textbf{463} & 537 & 1.00:1.16 & 1:1 \\
00012 & \textbf{464} & 536 & 1.00:1.16 & 1:1 \\
00013 & \textbf{464} & 536 & 1.00:1.16 & 1:1 \\
00014 & \textbf{470} & 530 & 1.00:1.13 & 1:1 \\
00015 & \textbf{436} & 564 & 1.00:1.29 & 1:1 \\
00016 & \textbf{442} & 558 & 1.00:1.26 & 1:1 \\
00017 & \textbf{495} & 505 & 1.00:1.02 & 1:1 \\
00018 & \textbf{430} & 570 & 1.00:1.33 & 1:1 \\
00019 & \textbf{443} & 557 & 1.00:1.26 & 1:1 \\
00020 & \textbf{458} & 542 & 1.00:1.18 & 1:1 \\
00021 & \textbf{470} & 530 & 1.00:1.13 & 1:1 \\
00022 & \textbf{482} & 518 & 1.00:1.07 & 1:1 \\
00023 & \textbf{462} & 538 & 1.00:1.16 & 1:1 \\
00024 & \textbf{457} & 543 & 1.00:1.19 & 1:1 \\
00025 & \textbf{476} & 524 & 1.00:1.10 & 1:1 \\
00026 & \textbf{468} & 532 & 1.00:1.14 & 1:1 \\
00027 & \textbf{473} & 527 & 1.00:1.11 & 1:1 \\
00028 & \textbf{453} & 547 & 1.00:1.21 & 1:1 \\
00029 & \textbf{456} & 544 & 1.00:1.19 & 1:1 \\
00030 & \textbf{452} & 548 & 1.00:1.21 & 1:1 \\
00031 & \textbf{468} & 532 & 1.00:1.14 & 1:1 \\
00032 & \textbf{445} & 555 & 1.00:1.25 & 1:1 \\
00033 & \textbf{463} & 537 & 1.00:1.16 & 1:1 \\
00034 & \textbf{468} & 532 & 1.00:1.14 & 1:1 \\
00035 & \textbf{441} & 559 & 1.00:1.27 & 1:1 \\
00036 & \textbf{449} & 551 & 1.00:1.23 & 1:1 \\
00037 & \textbf{417} & 583 & 1.00:1.40 & 1:1 \\
00038 & \textbf{463} & 537 & 1.00:1.16 & 1:1 \\
00039 & \textbf{436} & 564 & 1.00:1.29 & 1:1 \\
00040 & \textbf{469} & 531 & 1.00:1.13 & 1:1 \\
00041 & \textbf{479} & 521 & 1.00:1.09 & 1:1 \\
00042 & \textbf{463} & 537 & 1.00:1.16 & 1:1 \\
00043 & \textbf{465} & 535 & 1.00:1.15 & 1:1 \\
00044 & \textbf{461} & 539 & 1.00:1.17 & 1:1 \\
00045 & \textbf{476} & 524 & 1.00:1.10 & 1:1 \\
00046 & \textbf{443} & 557 & 1.00:1.26 & 1:1 \\
00047 & \textbf{452} & 548 & 1.00:1.21 & 1:1 \\
00048 & \textbf{462} & 538 & 1.00:1.16 & 1:1 \\
00049 & \textbf{496} & 504 & 1.00:1.02 & 1:1 \\
00050 & \textbf{480} & 520 & 1.00:1.08 & 1:1 \\
00051 & \textbf{492} & 508 & 1.00:1.03 & 1:1 \\
00052 & \textbf{470} & 530 & 1.00:1.13 & 1:1 \\
00053 & \textbf{485} & 515 & 1.00:1.06 & 1:1 \\
00054 & \textbf{469} & 531 & 1.00:1.13 & 1:1 \\
00055 & \textbf{460} & 540 & 1.00:1.17 & 1:1 \\
00056 & \textbf{456} & 544 & 1.00:1.19 & 1:1 \\
00057 & \textbf{483} & 517 & 1.00:1.07 & 1:1 \\
00058 & \textbf{470} & 530 & 1.00:1.13 & 1:1 \\
00059 & \textbf{464} & 536 & 1.00:1.16 & 1:1 \\
00060 & \textbf{436} & 564 & 1.00:1.29 & 1:1 \\
00061 & \textbf{445} & 555 & 1.00:1.25 & 1:1 \\
00062 & \textbf{466} & 534 & 1.00:1.15 & 1:1 \\
00063 & \textbf{481} & 519 & 1.00:1.08 & 1:1 \\
00064 & \textbf{442} & 558 & 1.00:1.26 & 1:1 \\
00065 & \textbf{482} & 518 & 1.00:1.07 & 1:1 \\
00066 & \textbf{468} & 532 & 1.00:1.14 & 1:1 \\
00067 & \textbf{467} & 533 & 1.00:1.14 & 1:1 \\
00068 & \textbf{455} & 545 & 1.00:1.20 & 1:1 \\
00069 & \textbf{458} & 542 & 1.00:1.18 & 1:1 \\
00070 & \textbf{471} & 529 & 1.00:1.12 & 1:1 \\
00071 & \textbf{436} & 564 & 1.00:1.29 & 1:1 \\
00072 & \textbf{455} & 545 & 1.00:1.20 & 1:1 \\
00073 & \textbf{491} & 509 & 1.00:1.04 & 1:1 \\
00074 & \textbf{479} & 521 & 1.00:1.09 & 1:1 \\
00075 & \textbf{459} & 541 & 1.00:1.18 & 1:1 \\
00076 & \textbf{441} & 559 & 1.00:1.27 & 1:1 \\
00077 & \textbf{454} & 546 & 1.00:1.20 & 1:1 \\
00078 & \textbf{479} & 521 & 1.00:1.09 & 1:1 \\
00079 & \textbf{451} & 549 & 1.00:1.22 & 1:1 \\
00080 & \textbf{489} & 511 & 1.00:1.04 & 1:1 \\
00081 & \textbf{459} & 541 & 1.00:1.18 & 1:1 \\
00082 & \textbf{468} & 532 & 1.00:1.14 & 1:1 \\
00083 & \textbf{489} & 511 & 1.00:1.04 & 1:1 \\
00084 & \textbf{474} & 526 & 1.00:1.11 & 1:1 \\
00085 & \textbf{434} & 566 & 1.00:1.30 & 1:1 \\
00086 & \textbf{475} & 525 & 1.00:1.11 & 1:1 \\
00087 & \textbf{437} & 563 & 1.00:1.29 & 1:1 \\
00088 & \textbf{456} & 544 & 1.00:1.19 & 1:1 \\
00089 & \textbf{442} & 558 & 1.00:1.26 & 1:1 \\
00090 & \textbf{457} & 543 & 1.00:1.19 & 1:1 \\
\bottomrule
\end{longtblr}}
\onecolumn

{\small
\begin{longtblr}[
  caption = {diabetes (protected attribute: race)},
  label = {diabetes-race},
]{
  colspec = {lrrrrrrrr},
  rowhead = 1,
}
\toprule
Inst. & G1 & G2 & G3 & G4 & G5 & G6 & Ratio & Ideal \\
\midrule
00000 & \textbf{8} & 11 & 29 & 13 & 196 & 743 & 1.00:1.38:3.62:1.62:24.50:92.88 & 1:1:3:1:24:92 \\
00001 & \textbf{9} & 13 & 26 & 16 & 190 & 746 & 1.00:1.44:2.89:1.78:21.11:82.89 & 1:1:2:1:21:82 \\
00002 & \textbf{9} & 17 & 25 & 31 & 190 & 728 & 1.00:1.89:2.78:3.44:21.11:80.89 & 1:1:2:3:21:80 \\
00003 & \textbf{9} & 15 & 24 & 13 & 163 & 776 & 1.00:1.67:2.67:1.44:18.11:86.22 & 1:1:2:1:18:86 \\
00004 & \textbf{4} & 10 & 23 & 25 & 185 & 753 & 1.00:2.50:5.75:6.25:46.25:188.25 & 1:2:5:6:46:188 \\
00005 & \textbf{10} & 15 & 23 & 16 & 180 & 756 & 1.00:1.50:2.30:1.60:18.00:75.60 & 1:1:2:1:18:75 \\
00006 & \textbf{3} & 17 & 22 & 11 & 174 & 773 & 1.00:5.67:7.33:3.67:58.00:257.67 & 1:5:7:3:58:257 \\
00007 & \textbf{8} & 15 & 20 & 13 & 208 & 736 & 1.00:1.88:2.50:1.62:26.00:92.00 & 1:1:2:1:26:92 \\
00008 & \textbf{6} & 20 & 27 & 19 & 208 & 720 & 1.00:3.33:4.50:3.17:34.67:120.00 & 1:3:4:3:34:120 \\
00009 & \textbf{2} & 15 & 24 & 21 & 178 & 760 & 1.00:7.50:12.00:10.50:89.00:380.00 & 1:7:12:10:89:380 \\
00010 & \textbf{9} & 12 & 17 & 26 & 187 & 749 & 1.00:1.33:1.89:2.89:20.78:83.22 & 1:1:1:2:20:83 \\
00011 & \textbf{8} & 12 & 25 & 19 & 181 & 755 & 1.00:1.50:3.12:2.38:22.62:94.38 & 1:1:3:2:22:94 \\
00012 & \textbf{6} & 21 & 23 & 19 & 189 & 742 & 1.00:3.50:3.83:3.17:31.50:123.67 & 1:3:3:3:31:123 \\
00013 & 7 & \textbf{5} & 21 & 24 & 185 & 758 & 1.40:1.00:4.20:4.80:37.00:151.60 & 1:1:4:4:37:151 \\
00014 & \textbf{8} & 13 & 17 & 22 & 187 & 753 & 1.00:1.62:2.12:2.75:23.38:94.12 & 1:1:2:2:23:94 \\
00015 & \textbf{8} & 21 & 27 & 18 & 187 & 739 & 1.00:2.62:3.38:2.25:23.38:92.38 & 1:2:3:2:23:92 \\
00016 & \textbf{5} & 12 & 18 & 20 & 177 & 768 & 1.00:2.40:3.60:4.00:35.40:153.60 & 1:2:3:4:35:153 \\
00017 & \textbf{4} & 24 & 27 & 20 & 211 & 714 & 1.00:6.00:6.75:5.00:52.75:178.50 & 1:6:6:5:52:178 \\
00018 & \textbf{3} & 14 & 17 & 21 & 172 & 773 & 1.00:4.67:5.67:7.00:57.33:257.67 & 1:4:5:7:57:257 \\
00019 & \textbf{8} & 22 & 20 & 18 & 195 & 737 & 1.00:2.75:2.50:2.25:24.38:92.12 & 1:2:2:2:24:92 \\
00020 & \textbf{5} & 15 & 30 & 19 & 190 & 741 & 1.00:3.00:6.00:3.80:38.00:148.20 & 1:3:6:3:38:148 \\
00021 & \textbf{6} & 14 & 25 & 21 & 195 & 739 & 1.00:2.33:4.17:3.50:32.50:123.17 & 1:2:4:3:32:123 \\
00022 & \textbf{5} & 16 & 28 & 18 & 199 & 734 & 1.00:3.20:5.60:3.60:39.80:146.80 & 1:3:5:3:39:146 \\
00023 & \textbf{7} & 14 & 18 & 17 & 202 & 742 & 1.00:2.00:2.57:2.43:28.86:106.00 & 1:2:2:2:28:106 \\
00024 & \textbf{8} & 14 & 20 & 12 & 187 & 759 & 1.00:1.75:2.50:1.50:23.38:94.88 & 1:1:2:1:23:94 \\
00025 & \textbf{3} & 13 & 29 & 23 & 212 & 720 & 1.00:4.33:9.67:7.67:70.67:240.00 & 1:4:9:7:70:240 \\
00026 & \textbf{5} & 15 & 17 & 20 & 201 & 742 & 1.00:3.00:3.40:4.00:40.20:148.40 & 1:3:3:4:40:148 \\
00027 & \textbf{9} & 16 & 16 & 13 & 181 & 765 & 1.00:1.78:1.78:1.44:20.11:85.00 & 1:1:1:1:20:85 \\
00028 & \textbf{8} & 20 & 15 & 23 & 211 & 723 & 1.00:2.50:1.88:2.88:26.38:90.38 & 1:2:1:2:26:90 \\
00029 & \textbf{5} & 18 & 24 & 19 & 184 & 750 & 1.00:3.60:4.80:3.80:36.80:150.00 & 1:3:4:3:36:150 \\
00030 & \textbf{7} & 19 & 19 & 26 & 193 & 736 & 1.00:2.71:2.71:3.71:27.57:105.14 & 1:2:2:3:27:105 \\
00031 & \textbf{7} & 25 & 23 & 19 & 170 & 756 & 1.00:3.57:3.29:2.71:24.29:108.00 & 1:3:3:2:24:108 \\
00032 & \textbf{7} & 14 & 20 & 21 & 189 & 749 & 1.00:2.00:2.86:3.00:27.00:107.00 & 1:2:2:3:27:107 \\
00033 & \textbf{5} & 12 & 28 & 20 & 189 & 746 & 1.00:2.40:5.60:4.00:37.80:149.20 & 1:2:5:4:37:149 \\
00034 & \textbf{6} & 15 & 23 & 22 & 205 & 729 & 1.00:2.50:3.83:3.67:34.17:121.50 & 1:2:3:3:34:121 \\
00035 & \textbf{6} & 14 & 26 & 27 & 183 & 744 & 1.00:2.33:4.33:4.50:30.50:124.00 & 1:2:4:4:30:124 \\
00036 & \textbf{5} & 20 & 30 & 15 & 189 & 741 & 1.00:4.00:6.00:3.00:37.80:148.20 & 1:4:6:3:37:148 \\
00037 & \textbf{9} & 14 & 21 & 19 & 187 & 750 & 1.00:1.56:2.33:2.11:20.78:83.33 & 1:1:2:2:20:83 \\
00038 & \textbf{8} & 20 & 26 & 22 & 186 & 738 & 1.00:2.50:3.25:2.75:23.25:92.25 & 1:2:3:2:23:92 \\
00039 & \textbf{3} & 11 & 29 & 24 & 198 & 735 & 1.00:3.67:9.67:8.00:66.00:245.00 & 1:3:9:8:66:245 \\
00040 & \textbf{5} & 17 & 25 & 23 & 183 & 747 & 1.00:3.40:5.00:4.60:36.60:149.40 & 1:3:5:4:36:149 \\
00041 & \textbf{5} & 6 & 17 & 21 & 175 & 776 & 1.00:1.20:3.40:4.20:35.00:155.20 & 1:1:3:4:35:155 \\
00042 & 14 & \textbf{13} & 20 & 192 & 761 & & 1.08:1.00:1.54:14.77:58.54 & 1:1:1:14:58 \\
00043 & \textbf{7} & 17 & 22 & 23 & 187 & 744 & 1.00:2.43:3.14:3.29:26.71:106.29 & 1:2:3:3:26:106 \\
00044 & \textbf{6} & 10 & 30 & 17 & 204 & 733 & 1.00:1.67:5.00:2.83:34.00:122.17 & 1:1:5:2:34:122 \\
00045 & \textbf{6} & 11 & 27 & 14 & 170 & 772 & 1.00:1.83:4.50:2.33:28.33:128.67 & 1:1:4:2:28:128 \\
00046 & \textbf{4} & 18 & 14 & 16 & 202 & 746 & 1.00:4.50:3.50:4.00:50.50:186.50 & 1:4:3:4:50:186 \\
00047 & \textbf{6} & 14 & 23 & 16 & 189 & 752 & 1.00:2.33:3.83:2.67:31.50:125.33 & 1:2:3:2:31:125 \\
00048 & \textbf{4} & 15 & 20 & 15 & 201 & 745 & 1.00:3.75:5.00:3.75:50.25:186.25 & 1:3:5:3:50:186 \\
00049 & \textbf{4} & 13 & 27 & 27 & 214 & 715 & 1.00:3.25:6.75:6.75:53.50:178.75 & 1:3:6:6:53:178 \\
00050 & \textbf{5} & 12 & 19 & 21 & 181 & 762 & 1.00:2.40:3.80:4.20:36.20:152.40 & 1:2:3:4:36:152 \\
00051 & 10 & \textbf{6} & 19 & 18 & 192 & 755 & 1.67:1.00:3.17:3.00:32.00:125.83 & 1:1:3:3:32:125 \\
00052 & \textbf{6} & 10 & 22 & 28 & 189 & 745 & 1.00:1.67:3.67:4.67:31.50:124.17 & 1:1:3:4:31:124 \\
00053 & \textbf{12} & 18 & 17 & 21 & 211 & 721 & 1.00:1.50:1.42:1.75:17.58:60.08 & 1:1:1:1:17:60 \\
00054 & \textbf{3} & 12 & 14 & 26 & 205 & 740 & 1.00:4.00:4.67:8.67:68.33:246.67 & 1:4:4:8:68:246 \\
00055 & \textbf{6} & 12 & 27 & 17 & 193 & 745 & 1.00:2.00:4.50:2.83:32.17:124.17 & 1:2:4:2:32:124 \\
00056 & \textbf{8} & 27 & 22 & 28 & 187 & 728 & 1.00:3.38:2.75:3.50:23.38:91.00 & 1:3:2:3:23:91 \\
00057 & 7 & \textbf{4} & 18 & 29 & 189 & 753 & 1.75:1.00:4.50:7.25:47.25:188.25 & 1:1:4:7:47:188 \\
00058 & \textbf{5} & 16 & 17 & 24 & 204 & 734 & 1.00:3.20:3.40:4.80:40.80:146.80 & 1:3:3:4:40:146 \\
00059 & \textbf{3} & 12 & 18 & 21 & 205 & 741 & 1.00:4.00:6.00:7.00:68.33:247.00 & 1:4:6:7:68:247 \\
00060 & \textbf{6} & 16 & 23 & 22 & 198 & 735 & 1.00:2.67:3.83:3.67:33.00:122.50 & 1:2:3:3:33:122 \\
00061 & \textbf{11} & \textbf{11} & 18 & 19 & 182 & 759 & 1.00:1.00:1.64:1.73:16.55:69.00 & 1:1:1:1:16:69 \\
00062 & \textbf{8} & 15 & 26 & 19 & 193 & 739 & 1.00:1.88:3.25:2.38:24.12:92.38 & 1:1:3:2:24:92 \\
00063 & \textbf{4} & 10 & 22 & 22 & 171 & 771 & 1.00:2.50:5.50:5.50:42.75:192.75 & 1:2:5:5:42:192 \\
00064 & \textbf{7} & 17 & 30 & 23 & 193 & 730 & 1.00:2.43:4.29:3.29:27.57:104.29 & 1:2:4:3:27:104 \\
00065 & \textbf{8} & 16 & 22 & 17 & 189 & 748 & 1.00:2.00:2.75:2.12:23.62:93.50 & 1:2:2:2:23:93 \\
00066 & \textbf{4} & 15 & 18 & 19 & 178 & 766 & 1.00:3.75:4.50:4.75:44.50:191.50 & 1:3:4:4:44:191 \\
00067 & \textbf{5} & 10 & 32 & 22 & 183 & 748 & 1.00:2.00:6.40:4.40:36.60:149.60 & 1:2:6:4:36:149 \\
00068 & \textbf{3} & 20 & 14 & 13 & 208 & 742 & 1.00:6.67:4.67:4.33:69.33:247.33 & 1:6:4:4:69:247 \\
00069 & \textbf{9} & 10 & 22 & 15 & 185 & 759 & 1.00:1.11:2.44:1.67:20.56:84.33 & 1:1:2:1:20:84 \\
00070 & \textbf{9} & 14 & 18 & 15 & 195 & 749 & 1.00:1.56:2.00:1.67:21.67:83.22 & 1:1:2:1:21:83 \\
00071 & \textbf{11} & 15 & 18 & 29 & 183 & 744 & 1.00:1.36:1.64:2.64:16.64:67.64 & 1:1:1:2:16:67 \\
00072 & \textbf{3} & 14 & 30 & 18 & 181 & 754 & 1.00:4.67:10.00:6.00:60.33:251.33 & 1:4:10:6:60:251 \\
00073 & \textbf{5} & 11 & 17 & 14 & 210 & 743 & 1.00:2.20:3.40:2.80:42.00:148.60 & 1:2:3:2:42:148 \\
00074 & 5 & \textbf{4} & 22 & 17 & 174 & 778 & 1.25:1.00:5.50:4.25:43.50:194.50 & 1:1:5:4:43:194 \\
00075 & \textbf{9} & 13 & 25 & 21 & 193 & 739 & 1.00:1.44:2.78:2.33:21.44:82.11 & 1:1:2:2:21:82 \\
00076 & 10 & \textbf{9} & 21 & 22 & 178 & 760 & 1.11:1.00:2.33:2.44:19.78:84.44 & 1:1:2:2:19:84 \\
00077 & \textbf{4} & 22 & 22 & 25 & 182 & 745 & 1.00:5.50:5.50:6.25:45.50:186.25 & 1:5:5:6:45:186 \\
00078 & \textbf{5} & 18 & 20 & 29 & 206 & 722 & 1.00:3.60:4.00:5.80:41.20:144.40 & 1:3:4:5:41:144 \\
00079 & \textbf{6} & 15 & 24 & 17 & 190 & 748 & 1.00:2.50:4.00:2.83:31.67:124.67 & 1:2:4:2:31:124 \\
00080 & \textbf{5} & 17 & 22 & 10 & 200 & 746 & 1.00:3.40:4.40:2.00:40.00:149.20 & 1:3:4:2:40:149 \\
00081 & \textbf{11} & 27 & 25 & 22 & 175 & 740 & 1.00:2.45:2.27:2.00:15.91:67.27 & 1:2:2:2:15:67 \\
00082 & \textbf{4} & 16 & 23 & 25 & 165 & 767 & 1.00:4.00:5.75:6.25:41.25:191.75 & 1:4:5:6:41:191 \\
00083 & \textbf{6} & 16 & 23 & 10 & 181 & 764 & 1.00:2.67:3.83:1.67:30.17:127.33 & 1:2:3:1:30:127 \\
00084 & \textbf{6} & 12 & 18 & 20 & 179 & 765 & 1.00:2.00:3.00:3.33:29.83:127.50 & 1:2:3:3:29:127 \\
00085 & \textbf{4} & 10 & 26 & 19 & 191 & 750 & 1.00:2.50:6.50:4.75:47.75:187.50 & 1:2:6:4:47:187 \\
00086 & \textbf{3} & 21 & 24 & 21 & 183 & 748 & 1.00:7.00:8.00:7.00:61.00:249.33 & 1:7:8:7:61:249 \\
00087 & \textbf{4} & 12 & 25 & 19 & 215 & 725 & 1.00:3.00:6.25:4.75:53.75:181.25 & 1:3:6:4:53:181 \\
00088 & \textbf{11} & 14 & 28 & 19 & 200 & 728 & 1.00:1.27:2.55:1.73:18.18:66.18 & 1:1:2:1:18:66 \\
00089 & \textbf{7} & 14 & 29 & 21 & 177 & 752 & 1.00:2.00:4.14:3.00:25.29:107.43 & 1:2:4:3:25:107 \\
00090 & \textbf{7} & 17 & 27 & 28 & 197 & 724 & 1.00:2.43:3.86:4.00:28.14:103.43 & 1:2:3:4:28:103 \\
\bottomrule
\end{longtblr}}
\onecolumn

\twocolumn[]
{\small
\begin{longtblr}[
  caption = {income},
]{
  colspec = {lrrrr},
  rowhead = 1,
}
\toprule
Inst. & G1 & G2 & Ratio & Ideal \\
\midrule
00000 & \textbf{499} & 501 & 1.00:1.00 & 1:1 \\
00001 & 543 & \textbf{457} & 1.19:1.00 & 1:1 \\
00002 & 531 & \textbf{469} & 1.13:1.00 & 1:1 \\
00003 & 522 & \textbf{478} & 1.09:1.00 & 1:1 \\
00004 & 525 & \textbf{475} & 1.11:1.00 & 1:1 \\
00005 & 550 & \textbf{450} & 1.22:1.00 & 1:1 \\
00006 & 523 & \textbf{477} & 1.10:1.00 & 1:1 \\
00007 & 528 & \textbf{472} & 1.12:1.00 & 1:1 \\
00008 & 558 & \textbf{442} & 1.26:1.00 & 1:1 \\
00009 & 551 & \textbf{449} & 1.23:1.00 & 1:1 \\
00010 & 525 & \textbf{475} & 1.11:1.00 & 1:1 \\
00011 & 523 & \textbf{477} & 1.10:1.00 & 1:1 \\
00012 & 516 & \textbf{484} & 1.07:1.00 & 1:1 \\
00013 & 514 & \textbf{486} & 1.06:1.00 & 1:1 \\
00014 & 508 & \textbf{492} & 1.03:1.00 & 1:1 \\
00015 & 555 & \textbf{445} & 1.25:1.00 & 1:1 \\
00016 & 522 & \textbf{478} & 1.09:1.00 & 1:1 \\
00017 & 548 & \textbf{452} & 1.21:1.00 & 1:1 \\
00018 & 512 & \textbf{488} & 1.05:1.00 & 1:1 \\
00019 & 522 & \textbf{478} & 1.09:1.00 & 1:1 \\
00020 & 528 & \textbf{472} & 1.12:1.00 & 1:1 \\
00021 & 525 & \textbf{475} & 1.11:1.00 & 1:1 \\
00022 & 537 & \textbf{463} & 1.16:1.00 & 1:1 \\
00023 & 532 & \textbf{468} & 1.14:1.00 & 1:1 \\
00024 & 535 & \textbf{465} & 1.15:1.00 & 1:1 \\
00025 & 525 & \textbf{475} & 1.11:1.00 & 1:1 \\
00026 & 531 & \textbf{469} & 1.13:1.00 & 1:1 \\
00027 & 541 & \textbf{459} & 1.18:1.00 & 1:1 \\
00028 & 552 & \textbf{448} & 1.23:1.00 & 1:1 \\
00029 & 530 & \textbf{470} & 1.13:1.00 & 1:1 \\
00030 & 539 & \textbf{461} & 1.17:1.00 & 1:1 \\
00031 & 527 & \textbf{473} & 1.11:1.00 & 1:1 \\
00032 & 525 & \textbf{475} & 1.11:1.00 & 1:1 \\
00033 & 509 & \textbf{491} & 1.04:1.00 & 1:1 \\
00034 & 534 & \textbf{466} & 1.15:1.00 & 1:1 \\
00035 & 539 & \textbf{461} & 1.17:1.00 & 1:1 \\
00036 & 546 & \textbf{454} & 1.20:1.00 & 1:1 \\
00037 & 547 & \textbf{453} & 1.21:1.00 & 1:1 \\
00038 & 527 & \textbf{473} & 1.11:1.00 & 1:1 \\
00039 & 532 & \textbf{468} & 1.14:1.00 & 1:1 \\
00040 & 529 & \textbf{471} & 1.12:1.00 & 1:1 \\
00041 & 527 & \textbf{473} & 1.11:1.00 & 1:1 \\
00042 & 518 & \textbf{482} & 1.07:1.00 & 1:1 \\
00043 & 511 & \textbf{489} & 1.04:1.00 & 1:1 \\
00044 & 524 & \textbf{476} & 1.10:1.00 & 1:1 \\
00045 & 536 & \textbf{464} & 1.16:1.00 & 1:1 \\
00046 & 532 & \textbf{468} & 1.14:1.00 & 1:1 \\
00047 & 530 & \textbf{470} & 1.13:1.00 & 1:1 \\
00048 & 525 & \textbf{475} & 1.11:1.00 & 1:1 \\
00049 & 524 & \textbf{476} & 1.10:1.00 & 1:1 \\
00050 & 505 & \textbf{495} & 1.02:1.00 & 1:1 \\
00051 & 519 & \textbf{481} & 1.08:1.00 & 1:1 \\
00052 & 538 & \textbf{462} & 1.16:1.00 & 1:1 \\
00053 & 534 & \textbf{466} & 1.15:1.00 & 1:1 \\
00054 & 529 & \textbf{471} & 1.12:1.00 & 1:1 \\
00055 & 534 & \textbf{466} & 1.15:1.00 & 1:1 \\
00056 & 501 & \textbf{499} & 1.00:1.00 & 1:1 \\
00057 & 516 & \textbf{484} & 1.07:1.00 & 1:1 \\
00058 & \textbf{494} & 506 & 1.00:1.02 & 1:1 \\
00059 & 533 & \textbf{467} & 1.14:1.00 & 1:1 \\
00060 & \textbf{497} & 503 & 1.00:1.01 & 1:1 \\
00061 & 514 & \textbf{486} & 1.06:1.00 & 1:1 \\
00062 & \textbf{491} & 509 & 1.00:1.04 & 1:1 \\
00063 & 542 & \textbf{458} & 1.18:1.00 & 1:1 \\
00064 & 522 & \textbf{478} & 1.09:1.00 & 1:1 \\
00065 & 545 & \textbf{455} & 1.20:1.00 & 1:1 \\
00066 & 553 & \textbf{447} & 1.24:1.00 & 1:1 \\
00067 & 533 & \textbf{467} & 1.14:1.00 & 1:1 \\
00068 & 542 & \textbf{458} & 1.18:1.00 & 1:1 \\
00069 & 557 & \textbf{443} & 1.26:1.00 & 1:1 \\
00070 & 541 & \textbf{459} & 1.18:1.00 & 1:1 \\
00071 & 525 & \textbf{475} & 1.11:1.00 & 1:1 \\
00072 & 506 & \textbf{494} & 1.02:1.00 & 1:1 \\
00073 & 527 & \textbf{473} & 1.11:1.00 & 1:1 \\
00074 & 533 & \textbf{467} & 1.14:1.00 & 1:1 \\
00075 & 508 & \textbf{492} & 1.03:1.00 & 1:1 \\
00076 & 513 & \textbf{487} & 1.05:1.00 & 1:1 \\
00077 & 556 & \textbf{444} & 1.25:1.00 & 1:1 \\
00078 & 555 & \textbf{445} & 1.25:1.00 & 1:1 \\
00079 & 547 & \textbf{453} & 1.21:1.00 & 1:1 \\
00080 & 554 & \textbf{446} & 1.24:1.00 & 1:1 \\
00081 & 508 & \textbf{492} & 1.03:1.00 & 1:1 \\
00082 & 533 & \textbf{467} & 1.14:1.00 & 1:1 \\
00083 & 518 & \textbf{482} & 1.07:1.00 & 1:1 \\
00084 & 529 & \textbf{471} & 1.12:1.00 & 1:1 \\
00085 & 521 & \textbf{479} & 1.09:1.00 & 1:1 \\
00086 & 534 & \textbf{466} & 1.15:1.00 & 1:1 \\
00087 & \textbf{498} & 502 & 1.00:1.01 & 1:1 \\
00088 & 524 & \textbf{476} & 1.10:1.00 & 1:1 \\
00089 & 532 & \textbf{468} & 1.14:1.00 & 1:1 \\
00090 & 539 & \textbf{461} & 1.17:1.00 & 1:1 \\
00091 & 531 & \textbf{469} & 1.13:1.00 & 1:1 \\
00092 & 519 & \textbf{481} & 1.08:1.00 & 1:1 \\
00093 & 543 & \textbf{457} & 1.19:1.00 & 1:1 \\
00094 & 529 & \textbf{471} & 1.12:1.00 & 1:1 \\
00095 & 534 & \textbf{466} & 1.15:1.00 & 1:1 \\
00096 & 542 & \textbf{458} & 1.18:1.00 & 1:1 \\
00097 & \textbf{498} & 502 & 1.00:1.01 & 1:1 \\
00098 & 524 & \textbf{476} & 1.10:1.00 & 1:1 \\
00099 & 512 & \textbf{488} & 1.05:1.00 & 1:1 \\
00100 & 503 & \textbf{497} & 1.01:1.00 & 1:1 \\
00101 & 547 & \textbf{453} & 1.21:1.00 & 1:1 \\
00102 & 519 & \textbf{481} & 1.08:1.00 & 1:1 \\
00103 & 524 & \textbf{476} & 1.10:1.00 & 1:1 \\
00104 & 546 & \textbf{454} & 1.20:1.00 & 1:1 \\
00105 & 518 & \textbf{482} & 1.07:1.00 & 1:1 \\
00106 & 509 & \textbf{491} & 1.04:1.00 & 1:1 \\
00107 & 520 & \textbf{480} & 1.08:1.00 & 1:1 \\
00108 & 529 & \textbf{471} & 1.12:1.00 & 1:1 \\
00109 & 538 & \textbf{462} & 1.16:1.00 & 1:1 \\
00110 & 515 & \textbf{485} & 1.06:1.00 & 1:1 \\
00111 & 516 & \textbf{484} & 1.07:1.00 & 1:1 \\
00112 & 520 & \textbf{480} & 1.08:1.00 & 1:1 \\
00113 & 519 & \textbf{481} & 1.08:1.00 & 1:1 \\
00114 & 518 & \textbf{482} & 1.07:1.00 & 1:1 \\
00115 & 513 & \textbf{487} & 1.05:1.00 & 1:1 \\
00116 & 520 & \textbf{480} & 1.08:1.00 & 1:1 \\
00117 & 529 & \textbf{471} & 1.12:1.00 & 1:1 \\
00118 & 540 & \textbf{460} & 1.17:1.00 & 1:1 \\
00119 & 528 & \textbf{472} & 1.12:1.00 & 1:1 \\
00120 & 545 & \textbf{455} & 1.20:1.00 & 1:1 \\
00121 & 525 & \textbf{475} & 1.11:1.00 & 1:1 \\
00122 & 558 & \textbf{442} & 1.26:1.00 & 1:1 \\
00123 & 523 & \textbf{477} & 1.10:1.00 & 1:1 \\
00124 & 506 & \textbf{494} & 1.02:1.00 & 1:1 \\
00125 & 539 & \textbf{461} & 1.17:1.00 & 1:1 \\
00126 & 531 & \textbf{469} & 1.13:1.00 & 1:1 \\
00127 & 545 & \textbf{455} & 1.20:1.00 & 1:1 \\
00128 & 567 & \textbf{433} & 1.31:1.00 & 1:1 \\
00129 & 506 & \textbf{494} & 1.02:1.00 & 1:1 \\
00130 & 519 & \textbf{481} & 1.08:1.00 & 1:1 \\
00131 & 523 & \textbf{477} & 1.10:1.00 & 1:1 \\
00132 & 527 & \textbf{473} & 1.11:1.00 & 1:1 \\
00133 & 528 & \textbf{472} & 1.12:1.00 & 1:1 \\
00134 & 510 & \textbf{490} & 1.04:1.00 & 1:1 \\
00135 & 522 & \textbf{478} & 1.09:1.00 & 1:1 \\
00136 & \textbf{497} & 503 & 1.00:1.01 & 1:1 \\
00137 & 557 & \textbf{443} & 1.26:1.00 & 1:1 \\
00138 & 514 & \textbf{486} & 1.06:1.00 & 1:1 \\
00139 & 532 & \textbf{468} & 1.14:1.00 & 1:1 \\
00140 & 536 & \textbf{464} & 1.16:1.00 & 1:1 \\
00141 & 524 & \textbf{476} & 1.10:1.00 & 1:1 \\
00142 & 512 & \textbf{488} & 1.05:1.00 & 1:1 \\
00143 & 537 & \textbf{463} & 1.16:1.00 & 1:1 \\
00144 & 524 & \textbf{476} & 1.10:1.00 & 1:1 \\
00145 & 536 & \textbf{464} & 1.16:1.00 & 1:1 \\
00146 & 511 & \textbf{489} & 1.04:1.00 & 1:1 \\
00147 & 542 & \textbf{458} & 1.18:1.00 & 1:1 \\
00148 & 506 & \textbf{494} & 1.02:1.00 & 1:1 \\
00149 & 551 & \textbf{449} & 1.23:1.00 & 1:1 \\
00150 & 560 & \textbf{440} & 1.27:1.00 & 1:1 \\
00151 & 502 & \textbf{498} & 1.01:1.00 & 1:1 \\
00152 & 552 & \textbf{448} & 1.23:1.00 & 1:1 \\
00153 & 514 & \textbf{486} & 1.06:1.00 & 1:1 \\
00154 & 529 & \textbf{471} & 1.12:1.00 & 1:1 \\
00155 & 540 & \textbf{460} & 1.17:1.00 & 1:1 \\
00156 & 556 & \textbf{444} & 1.25:1.00 & 1:1 \\
00157 & 505 & \textbf{495} & 1.02:1.00 & 1:1 \\
00158 & 550 & \textbf{450} & 1.22:1.00 & 1:1 \\
00159 & 524 & \textbf{476} & 1.10:1.00 & 1:1 \\
00160 & 567 & \textbf{433} & 1.31:1.00 & 1:1 \\
00161 & 561 & \textbf{439} & 1.28:1.00 & 1:1 \\
00162 & 539 & \textbf{461} & 1.17:1.00 & 1:1 \\
00163 & 514 & \textbf{486} & 1.06:1.00 & 1:1 \\
00164 & 548 & \textbf{452} & 1.21:1.00 & 1:1 \\
00165 & 525 & \textbf{475} & 1.11:1.00 & 1:1 \\
00166 & 525 & \textbf{475} & 1.11:1.00 & 1:1 \\
00167 & 514 & \textbf{486} & 1.06:1.00 & 1:1 \\
00168 & 523 & \textbf{477} & 1.10:1.00 & 1:1 \\
00169 & 516 & \textbf{484} & 1.07:1.00 & 1:1 \\
00170 & 536 & \textbf{464} & 1.16:1.00 & 1:1 \\
00171 & 525 & \textbf{475} & 1.11:1.00 & 1:1 \\
00172 & 524 & \textbf{476} & 1.10:1.00 & 1:1 \\
00173 & 504 & \textbf{496} & 1.02:1.00 & 1:1 \\
00174 & 540 & \textbf{460} & 1.17:1.00 & 1:1 \\
00175 & 530 & \textbf{470} & 1.13:1.00 & 1:1 \\
00176 & 537 & \textbf{463} & 1.16:1.00 & 1:1 \\
00177 & 505 & \textbf{495} & 1.02:1.00 & 1:1 \\
00178 & 526 & \textbf{474} & 1.11:1.00 & 1:1 \\
00179 & 553 & \textbf{447} & 1.24:1.00 & 1:1 \\
00180 & 507 & \textbf{493} & 1.03:1.00 & 1:1 \\
00181 & 534 & \textbf{466} & 1.15:1.00 & 1:1 \\
00182 & 524 & \textbf{476} & 1.10:1.00 & 1:1 \\
00183 & 537 & \textbf{463} & 1.16:1.00 & 1:1 \\
00184 & 535 & \textbf{465} & 1.15:1.00 & 1:1 \\
00185 & 525 & \textbf{475} & 1.11:1.00 & 1:1 \\
00186 & \textbf{497} & 503 & 1.00:1.01 & 1:1 \\
00187 & 526 & \textbf{474} & 1.11:1.00 & 1:1 \\
00188 & \textbf{500} & \textbf{500} & 1.00:1.00 & 1:1 \\
00189 & 505 & \textbf{495} & 1.02:1.00 & 1:1 \\
00190 & 538 & \textbf{462} & 1.16:1.00 & 1:1 \\
00191 & 532 & \textbf{468} & 1.14:1.00 & 1:1 \\
00192 & 553 & \textbf{447} & 1.24:1.00 & 1:1 \\
00193 & 524 & \textbf{476} & 1.10:1.00 & 1:1 \\
00194 & 544 & \textbf{456} & 1.19:1.00 & 1:1 \\
\bottomrule
\end{longtblr}}
\onecolumn
\fi

\end{document}